\documentclass[11pt]{article}
\usepackage[a4paper,margin=1in]{geometry}
\usepackage[T1]{fontenc}
\usepackage[utf8]{inputenc}
\usepackage{amsmath,amssymb,amsthm}
\usepackage{graphicx,subfig}
\usepackage{enumitem}
\usepackage{bbm}
\usepackage{booktabs}
\usepackage{natbib}
\bibpunct[, ]{(}{)}{,}{a}{}{,}

\usepackage{hyperref}
\newtheorem{theorem}{Theorem}
\newtheorem{lemma}{Lemma}
\newtheorem{proposition}{Proposition}

\newtheorem{definition}{Definition}
\newtheorem{assumption}{Assumption}
\newtheorem{example}{Example}
\newtheorem{remark}{Remark}
\renewenvironment{proof}[1]{\par\noindent\textit{#1}\ }{\par}
\newcommand{\Halmos}{\ensuremath{\square}}

\newcommand{\PP}{\mathbb{P}}
\newcommand{\EE}{\mathbb{E}}
\newcommand{\ind}{\mathbbm{1}}
\newcommand{\RR}{\mathbb{R}}
\newcommand{\Var}{\operatorname{Var}}
\newcommand{\KL}{\operatorname{KL}}
\newcommand{\bs}[1]{\boldsymbol{#1}}
\newcommand{\algo}{\textsf{2LLM}}
\DeclareMathOperator*{\argmin}{arg\,min}
\hypersetup{
  hidelinks,
  pdftitle={Sequential Sampling for Binary Classification: Two LLMs are (Almost) All You Need},
  pdfauthor={Guokai Li, Jiaxin Liang, Mo Liu, Yanzhe Lei, Stefanus Jasin, Fenghua Yang, Preet Baxi}
}

\begin{document}
\begin{center}
{\LARGE Sequential Sampling for Binary Classification:\par
Two LLMs are (Almost) All You Need\par}
\vspace{1.5em}

\begingroup\small
{\normalsize Guokai Li, Jiaxin (Alys) Liang\par}
McGill University, Montreal, Quebec H3A 1G5, Canada\par
\{guokai.li, alys.liang\}@mcgill.ca\par
\vspace{1em}

{\normalsize Mo Liu\par}
University of North Carolina, Chapel Hill, NC 27599, USA\par
mo\_liu@unc.edu\par
\vspace{1em}

{\normalsize Yanzhe (Murray) Lei\par}
Queen's University, Kingston, ON K7L 3N6, Canada\par
yl64@queensu.ca\par
\vspace{1em}

{\normalsize Stefanus Jasin, Fenghua Yang, Preet Baxi\par}
University of Michigan, Ann Arbor, MI 48109, USA\par
\{sjasin, yfenghua, preetb\}@umich.edu\par
\endgroup
\end{center}

\begin{abstract}
We study how to orchestrate multiple large language models (LLMs) for reliable binary classification when models differ in accuracy, monetary cost, and response time. A decision-maker adaptively chooses which LLM to query, updates its belief after each response, and stops once sufficient confidence has been reached.
Our results reveal a sharp contrast between exact and near-optimal orchestration. In the worst case, an exactly optimal policy may need to query every available LLM. Yet under high reliability requirements, a simple policy using at most two LLMs is asymptotically optimal. The policy identifies one specialist for each possible outcome and adaptively switches between them according to the accumulated evidence. We derive a universal lower bound on the achievable cost and show that this two-LLM policy approaches the bound at the best possible rate. Two LLMs are not merely sufficient—they can be essential. When models have complementary strengths, relying on the best single LLM can be arbitrarily more costly than using two specialists. How the models are queried also matters: committing to a fixed query plan can waste substantial resources because it cannot stop early or redirect queries as evidence accumulates. In a weak-signal regime, we analytically show that the best static design costs four times as much as the optimal sequential policy. Numerical experiments using 18 LLMs on a public benchmark show that these benefits remain substantial at practically relevant confidence levels. The main message is simple: near-optimal LLM orchestration need not require coordinating a large portfolio of models. Identifying two complementary specialists and switching between them as evidence evolves can capture most of the value of adaptive inference.
\end{abstract}

\section{Introduction}
\label{sec:intro}

A defining feature of modern AI systems is their reliance on \emph{test-time compute}: allocating more computational resources at inference to improve the quality of the output. While classical machine learning invested almost all effort in training, the frontier of large language model (LLM) deployment has shifted decisively toward inference-time strategies. Reasoning models such as OpenAI's o1 routinely generate dozens of candidate solutions per query, then aggregate them via majority voting or verification \citep{snell2024scaling, wu2024inference}. Agentic workflows route tasks through cascades of specialized models and tools, each invocation adding cost and latency \citep{wang2024mixture}. Indeed, this multi-source orchestration is becoming the default architecture for frontier systems: OpenAI's GPT-5, for instance, integrates multiple LLM models behind a single router that dynamically selects which model to invoke for each query \citep{openai2025intro}. These developments have made \textit{scaling test-time compute} one of the most pressing questions in AI, and the question is fundamentally operational: The key challenge is not to select a single model to deploy, but \emph{how to orchestrate multiple, heterogeneous information sources sequentially} to reach a reliable decision as quickly and cheaply as possible.

More broadly, many digital platforms face a similar challenge of orchestrating multiple information sources. For example, in a modern payment-fraud pipeline, a single transaction may pass through device fingerprinting, behavioral scoring models, merchant risk engines, and third-party verification services before a final verdict is reached \citep{abdallah2016fraud}. In cybersecurity, a single suspicious event may pass through multiple detection layers---anomaly detectors, intrusion-detection systems, endpoint scanners, and forensic analysis tools---before an analyst decides whether to escalate \citep{khraisat2019survey}. Content-moderation platforms screen posts through keyword filters, lightweight classifiers, contextual language models, and human reviewers, escalating ambiguous items through successive rounds of review \citep{gorwa2020algorithmic}. In all of these examples, the decision-maker faces the same sequential problem: For each incoming case there is a single, fixed but unknown ground truth, and the system must resolve this uncertainty by querying a sequence of tools, each of which returns a noisy signal at some cost and after some delay.

Despite the practical importance of this orchestration problem, the existing literature has not yet provided enough understanding. Recent work on test-time compute has shown that strategies such as repeated sampling, adaptive self-consistency, verifier-guided selection, self-correction, search, and debate can substantially improve reasoning quality \citep{wang2022self,aggarwal2023lets,li2024escape,chen2024more,snell2024scaling,zhang2024generative, huang2025sample}. However, these works largely focused on testing the empirical performances without providing structural understanding for the sequential decision of \emph{which} resource to query and \emph{when} to stop.

At the other end, classical sequential testing and experiment-design models provide the mathematical backbone for adaptive information acquisition \citep{wald1945sequential,chernoff1959sequential,naghshvar2013active}. However, these works were not developed for test-time compute in LLMs and therefore do not capture practical considerations such as heterogeneous query costs, response latency, and asymmetric diagnostic power. These considerations are critical when optimizing system performance. Figure~\ref{fig:model_tradeoffs} illustrates this heterogeneity among selected LLMs by comparing model intelligence, cost efficiency, and output speed. The horizontal axis reports the cost per Intelligence Index task in USD on a logarithmic scale, the vertical axis reports the Artificial Analysis Intelligence Index \citep{aa2026}, and the color represents output speed in tokens per second. The figure highlights that model capability does not translate into a uniform ordering across operational dimensions: models with similar Intelligence Index values can differ substantially in cost efficiency and generation speed.

\begin{figure}[htbp]
    \centering
    \includegraphics[width=0.9\linewidth]{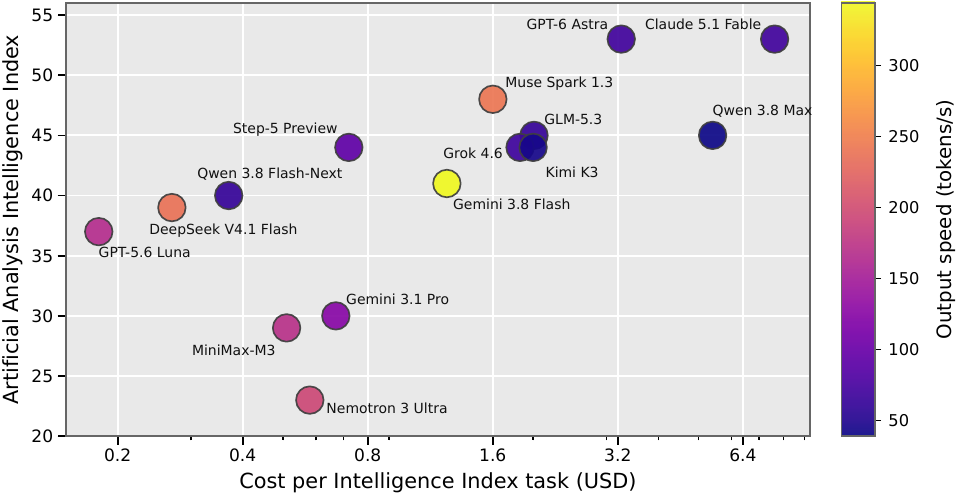}
    \caption{Cost--intelligence--speed trade-offs for selected LLMs. Model names are displayed next to their corresponding bubbles. }
    \label{fig:model_tradeoffs}
\end{figure}

This paper takes a first step toward filling this gap. We formulate a sequential hypothesis testing model in which a decision-maker faces a \textit{binary latent state} $\theta \in \{A , B\}$ and may adaptively query any of $M$ information sources. Each source $j$ is characterized by three primitives: A per-query monetary cost $c_j$, a random waiting time (i.e., latency) with a mean $\mu_j$, and a pair of asymmetric accuracies $(\gamma_{j,A}, \gamma_{j,B})$ with $\gamma_{j,\theta}\in (1/2, 1)$ that determine its diagnostic power under each hypothesis $\theta=A$ or $B$. The decision-maker updates a posterior belief after every query in a Bayesian manner and stops once the posterior error probability falls below a pre-defined threshold $\alpha$. The objective is to minimize the sum of expected monetary costs and a convex penalty on total waiting time, subject to the error constraint. Although the model is deliberately stylized, it captures the essential features shared by all the examples described above. Moreover, we are able to conduct a sharp asymptotic analysis in the high-confidence regime where error tolerance $\alpha$ approaches zero. This high-confidence regime is operationally relevant since digital platforms support many queries in real-time, and practical systems increasingly operate under stringent reliability requirements.

Our model does not attempt to capture open-ended answer generation or full agent-environment interaction. Instead, we focus on the task-level decision of how to spend additional test-time compute on one question with two possible answers when the solution quality needs to be improved. In many deployed AI systems, this is exactly where decision-makers face the most consequential trade-offs between reliability, latency, and cost: A platform may need to decide whether to accept or reject a generated answer, whether one response is better than the other, whether a code patch is ready to execute, whether an agent's action is safe enough to take, or whether a case should be escalated to a more expensive LLM or a human reviewer. In such settings, the latent state is effectively binary, but the platform may have access to several heterogeneous information sources: different LLMs, prompting strategies, verifiers, or human checks. These sources differ in price, latency, and statistical power.

Our analysis begins by asking what any admissible policy must achieve in order to meet the target error level. More precisely, we establish a universal lower bound on the total expected cost incurred by any admissible policy. Using a change-of-measure argument and martingale techniques, we show that the posterior-threshold stopping rule forces every policy to accumulate, under each hypothesis, on the order of $\log(1/\alpha)$ units of statistical evidence. This requirement is further translated into two deterministic information-budget constraints, one for each hypothesis. As a result, we propose a convex program whose optimal value constitutes the universal lower bound. The convex program has a natural interpretation as a query-allocation problem: It determines how the expected number of queries should be distributed across the $M$ information sources under each hypothesis in order to meet the error requirements while minimizing the combined monetary cost and delay penalty. Importantly, this lower-bound program reveals a simple structural property. Because the program contains only two linear information constraints and a convex objective, its optimal solution can be supported on at most one source for each hypothesis. Consequently, there exists an optimal policy that relies on at most two information sources, regardless of how many are available.

Inspired by this property, we then propose a simple class of \textit{sign-based two-source} policies: At each step, the decision-maker tracks the cumulative log-likelihood ratio and queries an $A$-specialist when the evidence favors hypothesis $A$, and a $B$-specialist when it favors $B$, stopping when the likelihood ratio crosses the posterior thresholds. Despite its simplicity, we show that an appropriately chosen pair of specialists achieves the universal lower bound up to an additive remainder. Moreover, the order of the remainder is tight: Even the optimal policy cannot improve the rate in general. In particular, we show that, as the evidence accumulates when the policy proceeds, an asymptotically optimal policy effectively alternates between two specialists; as the evidence becomes decisive, it commits to a single one. We further provide guidance on how to construct an asymptotically optimal two-source policy.

We further show that the proposed policy achieves a proper balance between cost performance and computational complexity. We first show that there exists a family of problem instances where any exact optimal policy must query all of the sources at least once with high probability. In other words, the computation cost is prohibitively high if one wants to achieve the minimum cost. We then show that using one source alone can incur a cost that is arbitrarily worse than the minimum cost, even in the high confidence regime, especially when different sources differ in their efficiency in gathering information. Lastly, we also show that there is a significant benefit in sequentially querying relative to a static design that chooses all sources' query counts before observing any responses. Taken together, these results suggest that adaptively querying two information sources is necessary to achieve near-optimal cost without incurring excessive computation costs.

Our work contributes to the emerging intersection of Operations Research (OR) and AI systems. The major contributions of this paper can be summarized as follows:

\begin{itemize}
    \item From a modeling perspective, we introduce a sequential hypothesis testing framework that captures a class of fundamental problems arising naturally in modern AI systems but is not well understood in the existing literature. In contrast with classical sequential testing, our model explicitly incorporates monetary cost, stochastic waiting time, and stopping decisions within one unified objective. In contrast with most existing work on test-time compute in computer science, we are able to obtain clean and insightful theoretical results.

    \item From a theoretical perspective, we provide a sharp asymptotic characterization of optimal performance in the high-confidence regime. The lower bound shows that the error rate requirement can be summarized by two information budgets, one for each hypothesis, while the upper bound shows that a simple sign-based policy attains this lower bound up to the best possible additive remainder. Classical sequential design results establish asymptotic optimality when the objective is expected sample size or a closely related stopping criterion. Our results go further by allowing source-specific costs, stochastic waiting times, and nonlinear delay penalties, and by showing that these operational features can be incorporated without losing sharp analytical structure.

   \item From a practical perspective, we show that while adaptively querying two sources can achieve asymptotically optimal cost, either dropping one more source or statically choosing the query counts will significantly increase the cost. This is a strong and actionable design principle. A near-optimal AI system does not need to orchestrate a large collection of sources in a complicated way; it needs to identify one source that is especially effective for each hypothesis, respectively, and then switch between them as evidence evolves. This insight helps explain why simple cascading may work well in practice, and suggests that the complexity of effective AI decision systems need not grow with the number of the available sources.

\end{itemize}

\subsection{Related Literature}
\label{sec:literature}

Our work connects to and draws on several streams of research. We will discuss them in turn and clarify how our contribution differs from or extends the existing results.

Our work is closely related to the literature on sequential hypothesis testing and sequential information acquisition. \cite{wald1945sequential} introduced the celebrated sequential probability ratio test (SPRT) and established its optimality for testing two simple hypotheses with a single source of observation. \cite{chernoff1959sequential} further extended these findings to the sequential design of experiments, where the decision-maker adaptively chooses which experiment to run. Later work \citep{naghshvar2013active} further developed active sequential hypothesis testing formulations that optimize information acquisition across multiple actions. \cite{dayanik2012multisource} study a continuous-time Bayesian binary testing problem with multiple simultaneously observed sources and characterize the optimal stopping region. \cite{song2017asymptotically} and \cite{tsopelakos2022sequential} study sequential multiple testing and anomaly detection across multiple independent data streams, where the objective is to identify which streams are anomalous under global error constraints and sampling limitations. There are also some works that incorporate operational frictions, such as switching costs \citep{vaidhiyan2015active} and delays \citep{lambez2021anomaly}. This literature provides the conceptual backbone of our analysis. Our model differs, however, in several important respects. We study a setting where the decision-maker chooses both which source to query and when to stop; each source has its own monetary cost, stochastic waiting time with a convex penalty on total waiting time, and asymmetric diagnostic power.

Our work also contributes to the rapidly growing computer-science literature on test-time compute for LLMs. Self-consistency \citep{wang2022self} established repeated sampling and answer aggregation as a powerful test-time strategy. Adaptive self-consistency and early-stopping methods show that it is often inefficient to work with a fixed budget and that one can reduce the number of model calls by stopping once sufficient agreement emerges \citep{aggarwal2023lets,li2024escape,huang2026optimal}. Other works study broader test-time scaling laws and compute-optimal inference strategies, showing that the value of additional calls depends on the complexity of the prompting technique and on how computation is split across generation, verification, and aggregation \citep{chen2024more,snell2024scaling,huang2025sample,singhi2025when}. Recent works on verifier-based reasoning and reward-model approaches similarly use additional test-time compute to search over, rank, or refine candidate solutions \citep{zhang2024generative,setlur2024rewarding}. Agentic methods such as ReAct, Tree of Thoughts, CRITIC, and debate further illustrate that performance improvement often comes from repeated calls, searches, critiques, and external feedback rather than from a single forward pass \citep{yao2022react,yao2023tree,gou2023critic,du2023improving}. Our paper is complementary to these works. Rather than proposing another protocol for repeated sampling, search, or self-correction, we propose a model that captures the key trade-off of how to orchestrate heterogeneous information sources sequentially when error rate, monetary cost, and latency all matter.

More broadly, the machine learning literature on model-cascade and budgeted-prediction studies how to escalate examples through sequences or trees of predictors with increasing computational cost, typically using an offline-designed cascade or partition of the input space \citep{viola2001rapid,xu2013cost}. Closely related developments in LLM-as-judge have moved beyond single judges to study cascaded selective evaluation, panels of diverse evaluators, and debate-based assessment \citep{kim2024debate, verga2024replacing,jung2025trust}. In the LLM setting, FrugalGPT revisits this idea with cascades over heterogeneous APIs \citep{chen2023frugalgpt}. A related but distinct line studies LLM routing, where a router makes a one-shot assignment of each query to a cheaper or stronger model based on predicted difficulty or quality--cost trade-offs \citep{ong2024routellm,ding2024hybrid}. Relative to these works, our model allows for sequential querying and re-routing as evidence accumulates; we also provide asymptotic optimality guarantees for the proposed policy.

Lastly, our work also lies at the intersection of OR and LLMs. This emerging literature uses OR tools such as queueing theory, online scheduling, robust optimization, and stochastic modeling to study how LLM systems should be run efficiently at inference time. Recent papers analyze how to batch and schedule many concurrent LLM requests under memory constraints, uncertain response lengths, and latency objectives \citep{ao2025optimizing,jaillet2025online,li2025throughput,wang2025serving,chen2025robust,bari2025optimal}. Collectively, these papers show that OR methods can provide new algorithmic insights and rigorous performance guarantees for important deployment problems in generative AI. The focus is mostly on how a shared system should allocate hardware and service capacity across many jobs. By contrast, we study a task-level problem within a single job: how to allocate sequential test-time compute across heterogeneous information sources until one instance can be certified with high confidence. \cite{huang2026optimal} is more closely related to our paper. They study Bayesian stopping for repeated sampling from a single LLM in order to identify the most consistent answer under informative priors. We, instead, study a binary hypothesis testing problem with multiple heterogeneous sources, each with its own diagnostic power, monetary cost, and stochastic response time. \citet{guo2026congestion} study classification tasks routed through LLM cascades in a steady-state system with endogenous queueing delays. Their analysis optimizes routing while accounting for congestion across jobs, whereas ours optimizes successive queries within a job while accounting for uncertainty in both evidence and accumulated response time. \cite{belloni2026online} formulate contextual LLM cascading as an online Pandora’s box problem and develop a policy with sublinear regret that learns which models to query and which generated output to deploy using feedback observed only for the deployed output. While their framework selects among generated outputs and learns across requests, we study sequential evidence accumulation for binary classification, minimizing query and waiting costs subject to accuracy requirements and establishing conditions under which two models suffice for asymptotic optimality. These modeling differences lead to a different mathematical structure and a different set of insights.

\subsection{Organization}
\label{sec:organization}
The remainder of the paper is organized as follows. Section~\ref{sec:model} introduces the model and its assumptions. Section~\ref{sec:finite-alpha-all-llms} shows that exact optimal policies may require every available model. Section~\ref{sec:2-LLM} develops the deterministic lower bound and the two-LLM policy, proves asymptotic optimality, and establishes the tight order of the gap to the benchmark. Sections~\ref{subsec:one-vs-two} and~\ref{sec:adaptivity} examine the value of a second model and the value of sequential querying, respectively. Section~\ref{sec:numerics} reports the numerical experiments, and Section~\ref{sec:conclusion} concludes. Proofs and additional computational details appear in the appendix.

\section{Model}\label{sec:model}
A decision-maker must identify a fixed but unknown ground truth $\theta\in\{A,B\}$ by sequentially querying a collection of information sources. For concreteness, we refer to these sources as LLMs, although the model applies to any collection of noisy binary classifiers. Let $[M]:=\{1,\ldots,M\}$ denote the available LLMs, and let $D\in\{A,B\}$ denote the final decision. The decision-maker's prior is
\[
P(\theta = A) = \xi_A \in (0,1), \qquad
P(\theta = B) = \xi_B := 1 - \xi_A \in (0,1).
\]

To reach a satisfactory decision, the decision-maker adaptively queries the LLMs sequentially until the posterior error probability falls below a threshold $\alpha \in (0, 1/2)$. At round $t=1,2,\ldots$, the decision-maker selects an LLM $j_t\in[M]$ using the observations available so far. The selected LLM returns an output $Y_t\in\{A,B\}$ after a random waiting time $T_{j_t,t}$. Queries are made one at a time, and an LLM may be queried repeatedly. Each LLM $j$ has three characteristics:
\begin{enumerate}
    \item A per-query monetary cost $c_j>0$.

    \item Asymmetric accuracies $\gamma_{j,A},\gamma_{j,B}\in(1/2,1)$: the probability of returning $A$ under $\theta=A$ is $\gamma_{j,A}$, whereas the probability of returning $B$ under $\theta=B$ is $\gamma_{j,B}$.

    \item A nonnegative waiting time with known mean $\mu_j:=\EE[T_{j,t}]\in(0,\infty)$.
\end{enumerate}

Allowing $\gamma_{j,A}\ne\gamma_{j,B}$ captures differences in a model's ability to distinguish the two hypotheses. For example, a model may reliably identify violations but reject compliant cases more often. The relevant trade-off is therefore not captured by a single average accuracy. We treat the accuracies and waiting-time distributions as known inputs for the class of tasks under consideration; they can be learnt from historical observation, so we do not explicitly model the process of learning this information.

\subsection{Admissible Policies and Posterior Dynamics}\label{sec:admin_policy}

Let $\mathcal H_t:=\sigma(j_1,Y_1,\ldots,j_t,Y_t)$ denote the history of model choices and outputs, and let $\mathcal G_t:=\sigma(j_1,Y_1,T_{j_1,1},\ldots,j_t,Y_t,T_{j_t,t})$ also include realized waiting times. For randomized policies, we augment these histories with an independent random seed available to the decision-maker. Under our assumption (see Assumption~\ref{assum:indep}), waiting times provide no additional evidence about $\theta$ once the choices and outputs are observed, so the posterior is the same under $\mathcal H_t$ and $\mathcal G_t$. Nevertheless, realized delays may affect which LLM the decision-maker queries next.

We consider non-anticipative policies that stop at the first round when either hypothesis has a posterior probability of at least $1-\alpha$, where $\alpha\in(0,1/2)$. Given the stopping criteria, the total query time is
\begin{align}
\tau := \inf\{t\ge 0: P(\theta=A\mid Y_{1:t}, j_{1:t})\ge 1-\alpha \text{ or }  P(\theta=B\mid Y_{1:t}, j_{1:t})\ge 1-\alpha\}, \label{eq:stopping-time}
\end{align}
which is also a stopping time. The corresponding decision rule is
\begin{align}
    D :=
\begin{cases}
A, & \text{if } \PP(\theta=A\mid Y_{1:\tau}, j_{1:\tau})\ge 1-\alpha ,\\
B, & \text{if } \PP(\theta=B\mid Y_{1:\tau}, j_{1:\tau})\ge 1-\alpha .
\end{cases}\label{eq:decision-rule}
\end{align}
\noindent This rule ensures $\PP(D\ne\theta\mid\mathcal G_\tau)\le\alpha$, and hence the probability of making a wrong inference satisfies $\PP(D\ne\theta)\le\alpha$. If the prior already meets the confidence requirement, then $\tau=0$. Our analysis concerns the nontrivial case $\alpha<\min\{\xi_A,\xi_B\}$, which holds eventually as $\alpha\to0$ for any fixed prior.

We now formally define the class of admissible policies.

\vspace{1mm}
\begin{definition}[Admissible Policy]\label{def:admissible}
For a given $\alpha\in(0,1/2)$, an admissible policy $\pi=\{\phi_t\}_{t\ge1}$ selects $j_t=\phi_t(\mathcal G_{t-1})\in[M]$ nonanticipatively and uses the stopping time and terminal decision in \eqref{eq:stopping-time} and \eqref{eq:decision-rule}, respectively. Independent policy randomization is included in the history as well.
\end{definition}

\vspace{1mm}
Let $\Pi(\alpha)$ denote this class. The stopping criterion is prescribed; the control decision is which LLM to query before the criterion is met. The realized stopping time remains endogenous because it depends on the chosen LLMs and their outputs.

In characterizing the posterior dynamic, we choose to work with the log-likelihood ratio (LLR) instead of the posterior probability directly. This is due to three reasons. First, the cumulative LLR accumulates additively across queries (see~\ref{def:Lt}), with  each query to LLM $j$ contributing a specific gain regardless of the current belief state. In contrast, the posterior updates nonlinearly. Second, as we show in Lemma~\ref{lemma1}, the posterior threshold stopping rule translates into a simple linear boundary crossing condition on the cumulative LLR, connecting our problem to classical random walk theory. Third, the expected increment of the log-likelihood ratio per query under each hypothesis  is a fixed property of LLM $j$ that does not depend on the cumulative LLR. This last point is crucial: it means the information contributed by each LLM can be characterized by two numbers, independent of how the process has evolved, which is what allows the optimal policy to be analyzed through the deterministic convex program developed in Section~\ref{subsec:lower_bound}.

We start by defining the LLR for a single LLM. Define the prior log-odds in favor of $A$ over $B$ as
\[
\delta := \log \frac{\xi_A}{\xi_B},
\]
which serves as the initial offset for the log-likelihood ratio process that accumulates evidence in favor of $A$ versus $B$. For any single LLM $j\in [M]$ and output $y\in\{A,B\}$, the LLR of $\theta=A$ versus $\theta=B$ can be written as
\begin{align}
\ell_j(y)
:= \log\frac{\PP(Y=y\mid \theta=A,j)}{\PP(Y=y\mid \theta=B,j)}=
\begin{cases}
\log\displaystyle\frac{\gamma_{j,A}}{1-\gamma_{j,B}}, & \text{if } y=A,\\[2mm]
\log\displaystyle\frac{1-\gamma_{j,A}}{\gamma_{j,B}}, & \text{if } y=B.
\end{cases} \label{eq:llr-asym}
\end{align}
Since $\gamma_{j,A},\gamma_{j,B}\in(1/2,1)$, we have $\ell_j(A) > 0$ and $\ell_j(B) < 0$. Recall that $j_t$ denotes the LLM chosen at round $t$, and $Y_t$ is its output. Then, starting from $t=1$, the \emph{cumulative log-likelihood ratio (LLR) up to $t$} can be written as
\begin{align}
    L_t := \sum_{s=1}^t \ell_{j_s}(Y_s). \tag{LLR}\label{def:Lt}
\end{align}
\noindent
For ease of exposition, let $\PP_\theta(\,\cdot\,) := \PP(\,\cdot\mid \theta)$ and $\EE_\theta[\cdot] := \EE[\cdot\mid \theta]$ for $\theta\in\{A,B\}$. Unsubscripted probabilities and expectations average over the prior, so
\[
\PP(\,\cdot\,) = \xi_A \PP_A(\,\cdot\,) + \xi_B \PP_B(\,\cdot\,),\qquad
\EE[\cdot] = \xi_A \EE_A[\cdot] + \xi_B \EE_B[\cdot].
\]

Using $L_t$, we can interpret stopping rule via a threshold structure. Formally, the following lemma translates the posterior confidence requirement into two stopping boundaries.
\begin{lemma}\label{lemma1}
Under Assumption~\ref{assum:indep}, for every time $t\ge 1$, we have
\begin{align}
 \PP(\theta=A\mid Y_{1:t}, j_{1:t}) = \frac{e^{\delta+L_t}}{1+e^{\delta+L_t}}, \qquad
 \PP(\theta=B\mid Y_{1:t}, j_{1:t}) = \frac{1}{1+e^{\delta+L_t}}. \label{pos_prob}
\end{align}
Moreover, the stopping time and the terminal decision can both be represented equivalently using a threshold structure
\begin{align}
\tau := \inf\{t\ge0: L_t\ge A_\alpha \text{ or } L_t\le -B_\alpha\}, \qquad
D :=
\begin{cases}
A, & \text{if } L_\tau \ge A_\alpha,\\
B, & \text{if } L_\tau \le -B_\alpha,
\end{cases}\label{threshold stopping time} \tag{Threshold Decision Rule}
\end{align}
where $
A_\alpha := \log\frac{1-\alpha}{\alpha} - \delta$ and $B_\alpha := \log\frac{1-\alpha}{\alpha} + \delta$.
\end{lemma}

The structure of the threshold policy has an intuitive explanation: querying will continue until sufficient evidence has been accumulated towards a hypothesis, measured by cumulative information gain crossing a certain threshold. In particular, $L_t$ measures the evidence supplied by the queries, while $\delta+L_t$ incorporates the prior. A prior favoring $A$ reduces the additional evidence needed to declare $A$ and increases the evidence needed to declare $B$.

To quantify an LLM's contribution to the query process, we associate each LLM $j\in [M]$ with a pair of \textit{expected information gains} under different ground truth $\theta$:
\begin{align*}
    I_{j,A} &:= \EE_A[\ell_j(Y)]
= \sum_{y\in\{A,B\}}\PP(Y=y\mid\theta=A,j)\cdot \ell_j(y),\\
I_{j,B} &:= -\EE_B[\ell_j(Y)]
= -\sum_{y\in\{A,B\}}\PP(Y=y\mid\theta=B,j)\cdot \ell_j(y).
\end{align*}
For simplicity, we will call these quantities \textit{information rates}, measured per query rather than per unit of time. Observe that both $I_{j,A}$ and $I_{j,B}$ are strictly positive because $\gamma_{j,A},\gamma_{j,B}\in(1/2,1)$. Under $A$, the process $L_t$ has a positive drift; under $B$, it is $-L_t$ that has a positive drift. Each rate depends on both accuracy parameters because an output is informative only to the extent that it distinguishes the hypotheses.

The definition of information rate is consistent with the KL divergence: If $P_{j,A}$ and $P_{j,B}$ are  the output distributions of LLM $j$ under the two hypotheses, then
\[
I_{j,A} = \KL(P_{j,A}\,\|\,P_{j,B})>0,\qquad
I_{j,B} = \KL(P_{j,B}\,\|\,P_{j,A})>0,
\]
where $\KL(P\|Q)$ denotes the Kullback–Leibler divergence.

We further define \emph{cost-efficiency} and \emph{time-efficiency} for each $\theta\in\{A, B\}$ as
\begin{align*}
    \kappa_{j,\theta} := \frac{c_j}{I_{j,\theta}},\qquad \eta_{j,\theta} := \frac{\mu_j}{I_{j,\theta}}, \text{ respectively }
\end{align*}
where $c_j$ is the per-query monetary cost, and $\mu_j$ is the mean waiting time. These quantities are central to the design of our proposed policy. Smaller values indicate greater efficiency. In particular, the fastest model per query need not be the fastest at resolving a hypothesis: a slower but more informative model can have a smaller $\eta_{j,\theta}$. These ratios, rather than accuracy or latency alone, guide the policy developed in Section~\ref{sec:2-LLM}.

\subsection{The Objective Function: Balancing Query Cost and Latency}
The objective of the decision maker is to achieve the accuracy threshold $1-\alpha$ at a low cost and high speed. For a policy $\pi$ with stopping time $\tau$, the total query cost and total waiting time are defined as
\begin{align}
    C_\pi := \sum_{t=1}^{\tau} c_{j_t}, \qquad
W_\pi := \sum_{t=1}^{\tau} T_{j_t,t}.\notag
\end{align}
To characterize the cost of waiting, we let $g : \RR_+ \to \RR_+$ be a convex and non-decreasing function with $g(0)=0$. The objective is to minimize the total weighted cost of a policy $\pi \in \Pi(\alpha)$ at error tolerance $\alpha$, defined as
\begin{align}\label{equ:obj}
 \mathcal{R}(\pi;\alpha) := \EE[C_\pi \, + \, g(W_\pi)]
 = \EE\left[\sum_{t=1}^{\tau} c_{j_t} \, + \, g\!\left(\sum_{t=1}^{\tau} T_{j_t,t}\right)\right] \tag{Objective}
\end{align}
\noindent Minimizing $\mathcal{R}(\pi;\alpha)$ strikes a trade-off between information gain, monetary cost, and total waiting time. To characterize such a trade-off, it is essential to carefully analyze the behavior of $L_t$ and the stopping time $\tau$ and further establish their connections to the objective $\mathcal{R}(\pi, \alpha)$.

Our model is analogous to Wald’s sequential probability ratio test (SPRT). However, there are at least two key features that distinguish our setting from classical sequential testing models. First, in standard formulations of sequential hypothesis testing, there is a single information channel, and the only control decision concerns \emph{when} to stop. In contrast, our framework involves multiple heterogeneous and asymmetric information sources. The decision-maker must decide not only \emph{when} to stop but also \emph{which} LLM to query. While the statistical characteristics of each LLM are known, their costs, waiting time distributions, and diagnostic efficiencies differ across hypotheses. The challenge, therefore, lies in optimally orchestrating multiple LLMs whose contributions vary across states. This results in a high-dimensional stochastic control problem in which both the stopping time and the flow of information are endogenous and depend on the entire query history. Second, the structure of the objective function introduces an additional layer of complexity. Query costs scale linearly with the number of queries, but the random total waiting time is penalized through a convex function. This nonlinearity prevents a direct application of standard optimal stopping or dynamic programming arguments. In particular, a near-optimal policy must carefully control the distribution of the stopping time and the accumulated waiting time, rather than only their expectations.

\vspace{1mm}
\begin{remark}\label{rmk:inf_tau}
It suffices to consider policies with $\EE[\tau]<\infty$. Indeed, if  $\EE[\tau]=\infty$, then $c_j>0$ implies,
\[
\EE[C_\pi]
\;\ge\;
(\xi_A \EE_A[\tau] + \xi_B \EE_B[\tau]) \cdot \min_{j \in [M]} c_j
\;=\; \infty,
\]
A policy with an infinite expected number of queries, therefore, has infinite cost. Under the assumptions below, repeatedly querying any fixed LLM yields a finite-cost threshold policy. Thus, infinite-cost policies cannot be optimal.
\end{remark}

\subsection{Assumptions and Discussion}
We impose the following assumptions on the information and waiting-time processes.
\begin{assumption}
\label{assum:indep}
We assume:
\begin{enumerate}
    \item[(i)] \textsc{Output independence.} Given the ground truth and the selected LLM, the next output has a history-independent distribution:
    \[
    \PP(Y_t = y \mid \theta, j_t=j, \mathcal{G}_{t-1}) = \PP(Y = y \mid \theta, j),
    \qquad y \in \{A,B\}, \ t \ge 1 .
    \]
    \item[(ii)] \textsc{Stationary waiting times.} For each LLM $j \in [M]$, the waiting times $(T_{j,t})_{t \ge 1}$ are i.i.d.\ with mean $\mu_j> 0$ and finite variance $\sigma_j^2\ge 0$.
    \item[(iii)] \textsc{Separation of information and latency.} Waiting times are  generated independently  of the ground truth, output-generating randomness, and policy randomization. In particular, for a measurable set $E\subseteq\RR_+$,
    \[
    \PP(T_{j_t,t}\in E\mid\theta,\mathcal G_{t-1},j_t,Y_t) =\PP(T_{j_t,t}\in E\mid j_t).
    \]

    \item[(iv)] \textsc{Cross-LLM independence.} The waiting-time sequences associated with different LLMs are mutually independent.
\end{enumerate}
\end{assumption}

These assumptions capture the idea that each query behaves as a fresh experiment whose output realization depends only on the selected LLM and the ground truth $\theta$. Part (i) ensures that evidence accumulates additively over time, which yields a tractable likelihood-ratio representation. Parts (ii)–(iv) separate informational uncertainty from system-level latency. In particular, inference delays are driven by computational and network factors, not by which hypothesis is true or which label is produced. This separation allows us to analyze information acquisition and delay costs in parallel and underpins the martingale and concentration arguments developed later.

We next formalize the tail property imposed on the waiting-time distributions.

\begin{assumption}[Sub-Gaussian Waiting Times]
\label{assum:wait}
There exists a constant $\psi>0$ such that, for every $j\in[M]$, $t\ge1$, and $\lambda\in\RR$,
\[
\EE\!\left[e^{\lambda(T_{j,t}-\mu_j)}\right]
\le \exp\!\left(\frac{\psi^2\lambda^2}{2}\right).
\]
\end{assumption}

Assumption~\ref{assum:wait}  controls the tails of individual response times. The concentration of the total delay also depends on the random query count and the routing policy; it does not follow from the individual tail bound alone.

Lastly, to enable our asymptotic analysis in Section~\ref{sec:2-LLM}, we introduce a mild regularity assumption on the  waiting-time cost function $g(\cdot)$.

\begin{assumption}[Power Waiting Cost]
\label{assum:g}
The waiting penalty is $g(x)=cx^\rho$, where $c\ge0$ and $\rho\ge1$.
\end{assumption}

Assumption~\ref{assum:g} captures one representative class of convex non-decreasing functions that characterize the waiting time penalty. It can be further generalized such that $g(x)$ is a non-decreasing, regularly varying function with order $\rho\ge 1$ (see \citealt{bingham1989regular}). Similar results can be derived following the proof framework in this paper.

The implications of Assumption \ref{assum:wait} and Assumption \ref{assum:g} are that, under any reasonable policy $\pi$, the total waiting time $W_\pi$ should scale on the same order as the number of queries $\tau$ which, as we will later show, grows only logarithmically in $1/\alpha$. If $g(\cdot)$ were allowed to grow exponentially fast, then even moderate random fluctuations in $W_\pi$
would make the waiting time penalty dominate the objective function, rendering the query cost meaningless. Assumption \ref{assum:g} rules out such explosive growth, while sub-Gaussian tails prevent large deviations
of $W_\pi$. Together, they ensure that the waiting-time penalty remains commensurate with query costs and information accumulation, so that $g(W_\pi)$ cannot overwhelm $C_\pi$ in the asymptotic
regime $\alpha\rightarrow 0$. This matches how LLMs are used in practice: users and system builders care about both the per-call price (e.g., pay-per-token or per-request fees) and the time-to-response, and adoption decisions are typically made by trading these two considerations off against each other. Moreover, in most production settings, response times are engineered to stay within a reasonable, finite range (via timeouts, load balancing, batching, and autoscaling), so while latency is stochastic, it typically exhibits limited tail risk around a stable mean, which is consistent with sub-Gaussian concentration.\vspace{2mm}

\begin{remark}[No monotonicity assumptions]
In this paper, we do \emph{not} assume any monotonic relationship between $c_j$,
$(\gamma_{j,A},\gamma_{j,B})$, and $\mu_j$ across different LLMs. Thus, an LLM may be relatively cheap but inaccurate, accurate but slow, or exhibit any other combination of characteristics. Our results do not rely on any ordering of the LLMs.
\end{remark}

In Table~\ref{tab:notation}, we provide a summary of the notations.
\begin{table}[htbp]
\centering
\caption{Summary of notation}\label{tab:notation}
\begin{tabular}{ll}
\toprule
Symbol & Definition\\
\midrule
$M$ & Number of LLMs\\
$\theta$ & Ground-truth hypothesis\\
$\xi_\theta$ & Prior probability of hypothesis $\theta\in\{A,B\}$\\
$\delta$ & Prior log-odds in favor of $A$\\
$\gamma_{j,\theta}$ & Accuracy of LLM $j$ under hypothesis $\theta$\\
$T_{j,t}$ & Waiting time of LLM $j$ at round $t$\\
$\mu_j$ & Mean waiting time of LLM $j$\\
$c_j$ & Per-query cost of LLM $j$\\
$\ell_j(y)$ & Single-query log-likelihood ratio\\
$L_t$ & Cumulative log-likelihood ratio through round $t$\\
$I_{j,\theta}$ & Expected evidence per query in favor of the true hypothesis\\
$A_\alpha,\,-B_\alpha$ & Upper and lower stopping boundaries for $L_t$\\
$\kappa_{j,\theta}=c_j/I_{j,\theta}$ & Monetary cost per unit of information \\
$\eta_{j,\theta}=\mu_j/I_{j,\theta}$ & Mean waiting time per unit of information\\
$\tau,\,C_\pi,\,W_\pi$ & Number of queries, total query cost, and total waiting time \\
\bottomrule
\end{tabular}
\end{table}

\section{Exact Optimal Policies May Query All the LLMs}
\label{sec:finite-alpha-all-llms}

The availability of many LLMs creates a natural design question: can the decision-maker select a small subset before querying begins without sacrificing optimality? At a fixed confidence requirement, a model's usefulness depends not only on its average information contribution, but also on whether its possible outputs provide enough evidence to stop from the current belief. We show that these state-dependent advantages can make every available LLM necessary for exact optimality.

Let $N_j(t)$ denote the number of queries sent to LLM $j$ through round $t$ and define $\bar\alpha:=(1+e^4)^{-1}$. The following theorem constructs instances in which every exact optimal policy queries all available models with high probability under hypothesis $A$.

\begin{theorem}[Worst-case use of all available LLMs]
\label{thm:finite-alpha-all-llms}
For each \(\alpha\in(0,\bar\alpha]\), define $H_\alpha:=\log\frac{1-\alpha}{\alpha}$ and $M_\alpha:=\bigl\lfloor \sqrt{H_\alpha}\bigr\rfloor$.
Then there exists a problem instance with \(M_\alpha\) LLMs such that every exact optimal policy $\pi^\star\in\arg\min_{\pi\in\Pi(\alpha)}\mathcal{R}(\pi;\alpha)$
 satisfies
\[
\PP_A\left(N_j(\tau)\ge 1 \text{ for all } j\in[M_\alpha]\right)
\ge
1-M_\alpha\frac{\alpha}{1-\alpha}.
\]
Moreover, \(M_\alpha\to\infty\) and \(M_\alpha\frac{\alpha}{1-\alpha}\to 0\) as \(\alpha\rightarrow 0\).
\end{theorem}

Theorem~\ref{thm:finite-alpha-all-llms} rules out a constant bound, independent of both the instance and the error tolerance, on the number of LLMs used by an exact optimal policy. Moreover, the need for all models is not confined to rare histories: conditional on $A$, the probability of using every model approaches one along the constructed family.

We explain the construction to show why different models become useful at different beliefs. Take a symmetric prior, so $\delta=0$ and $A_\alpha=B_\alpha=H_\alpha$. Set $c_j=1$ and $T_{j,t}\equiv1$ for every model, and let $g(w)=w$. Then $\mathcal R(\pi;\alpha)=2\EE[\tau]$, so minimizing total cost is equivalent to minimizing the expected number of queries. Equal prices and response times remove cost and speed differences across models. The choice of model therefore turns entirely on how its outputs move the belief toward a stopping boundary. By Lemma~\ref{lemma1}, querying continues while the LLR lies in $(-H_\alpha,H_\alpha)$.

Let $\lambda_\alpha:=H_\alpha/M_\alpha$ and define the ladder of LLR states
\[
x_i:=i\lambda_\alpha,\qquad i=0,1,\ldots,M_\alpha,
\]
where $x_0=0$ and $x_{M_\alpha}=H_\alpha$. The accuracy parameters are chosen so that, at each continuation state $x_i$ on this ladder, an $A$ output from any LLM moves the LLR to $x_{i+1}$. A $B$ output from LLM $j$ instead moves it to $x_i-(M_\alpha+j-1)\lambda_\alpha$. Appendix~\ref{sec:app-finite-alpha-all-llms} verifies that these transitions arise from accuracies in $(1/2,1)$. Thus, all models supply the same additional evidence after an $A$ output, but a larger-index model supplies stronger contrary evidence after a $B$ output. This stronger $B$ output is also less likely, given the current belief.

At $x_i$, LLM $i+1$ is the smallest-index model whose $B$ output reaches the lower stopping boundary. This is the central feature of the construction: as the accumulated evidence moves farther toward $A$, a larger downward move is needed for a $B$ output to end the test. LLM $i+1$ supplies exactly that move.

A smaller-index model $j\le i$ produces $B$ with a slightly higher probability, but that response leaves the LLR inside the continuation region, requiring further queries. The construction ensures that the saving from making an $A$ output less likely is too small to compensate for this additional continuation after a $B$ output. A larger-index model $j>i+1$, when available, also stops after a $B$ output, but is less likely to produce it. Its $A$ output leads to the same continuation state $x_{i+1}$, so it increases the probability of requiring more queries without improving the stopping outcome after $B$. Overshooting the lower boundary is not itself penalized; it simply provides no additional benefit once the confidence requirement is met.

The Bellman comparisons in Appendix~\ref{sec:app-finite-alpha-all-llms} establish that LLM $i+1$ is uniquely optimal at $x_i$. At the last continuation state $x_{M_\alpha-1}$, LLM $M_\alpha$ stops after either output, whereas every smaller-index model still requires additional queries after a $B$ output. The choice $\bar\alpha=(1+e^4)^{-1}$ ensures that the comparisons hold throughout the stated range of error tolerances.

Figure~\ref{fig:illustration_all_LLM} illustrates the construction with $\alpha=(1+e^{25})^{-1}$, so $H_\alpha=25$, $M_\alpha=5$, and $\lambda_\alpha=5$. At $x_2=10$, all five models move to $x_3=15$ after an $A$ output, so querying continues. The figure compares three representative models through their $B$ outputs: LLM~2 moves to $-20$, still inside the continuation region; LLM~3 reaches the lower boundary $-25$; and LLM~5 moves beyond it to $-35$. Thus, LLM~3 is the first model whose $B$ output suffices to stop. The preceding comparisons explain why it is uniquely optimal among all five models, not merely the three shown.

\begin{figure}[ht]
    \centering
    \includegraphics[width=\linewidth]{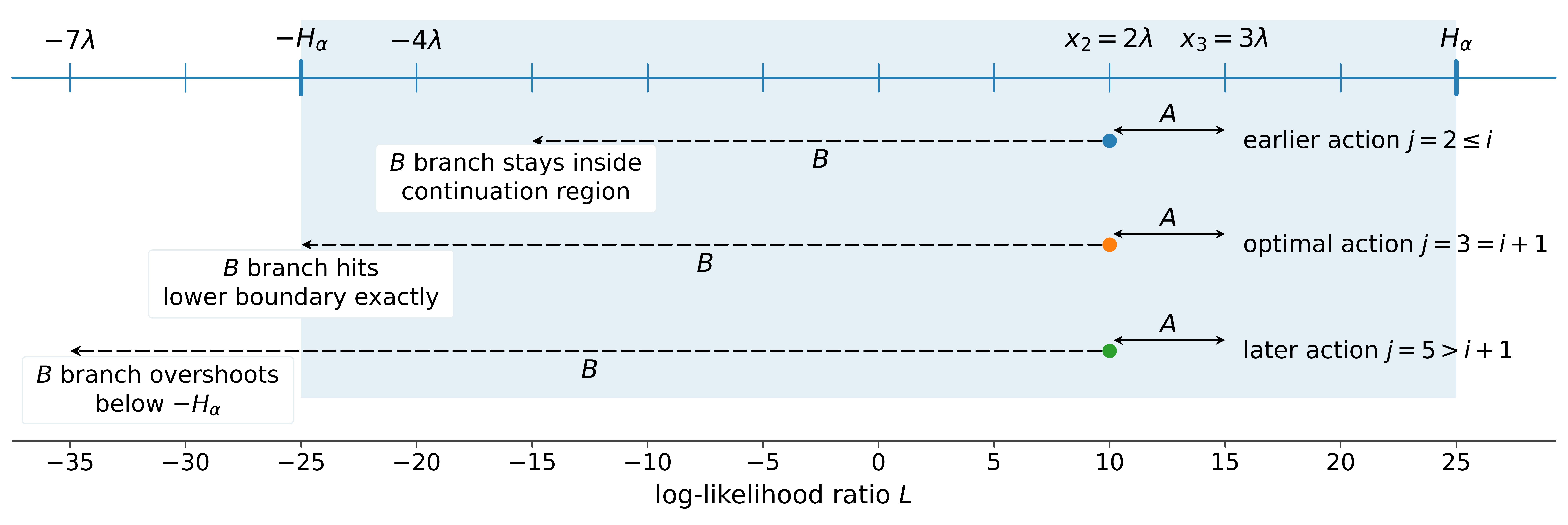}
    \caption{An illustrative example within the family of instances as defined in Theorem~\ref{thm:finite-alpha-all-llms}}

    \label{fig:illustration_all_LLM}
\end{figure}

To obtain the probability bound, consider the path on which the first $M_\alpha$ outputs are all $A$. Because each ladder state has a unique optimal model, every exact optimal policy follows
\[
0=x_0\to x_1\to \cdots \to x_{M_\alpha}=A_\alpha.
\]
In other words, every exact optimal policy must query LLMs \(1,2,\dots,M_\alpha\) in that order before stopping. Conditional on $\theta=A$, this path has a probability $\prod_{j=1}^{M_\alpha}\gamma_{j,A}.$ Our construction further ensures
\[
1-\gamma_{j,A}\le e^{-H_\alpha}=\frac{\alpha}{1-\alpha},
\]
and therefore
\[
\prod_{j=1}^{M_\alpha}\gamma_{j,A}
\ge
1-\sum_{j=1}^{M_\alpha}(1-\gamma_{j,A})
\ge
1-M_\alpha\frac{\alpha}{1-\alpha}.
\]

In summary, Theorem~\ref{thm:finite-alpha-all-llms} shows that, in the non-asymptotic setting, the exact optimal policy may use all available LLMs in the worst case. Computing such an exact optimal policy would, therefore, require solving a high-dimensional stochastic dynamic programming problem, which is intractable in practice. This finding motivates the design goal of the rest of the paper. Since the exact optimal policy with finite-\(\alpha\) can be structurally complex, could a simpler adaptive policy have near-optimal performance in a practically relevant regime? Fortunately, as we show later, a policy that  only uses two LLMs is asymptotically optimal in the high confidence regime.

\section{The Two-LLM Policy}\label{sec:2-LLM}

We now develop a simple policy by first identifying the information that any admissible policy must acquire. In particular, we first translate the posterior confidence requirement into deterministic information budgets and construct a deterministic lower bound program for the optimal objective. Motivated by the structure of the optimal solution to this lower bound program, we propose a policy that identifies an efficient LLM for each hypothesis and then routes queries to one of the two LLMs according to the current posterior. We show that, although the proposed policy has to allocate queries without knowing the truth, as $\alpha\to0$, its expected cost approaches the optimal cost of the lower bound program in which queries are allocated based on the ground truth and, therefore, is asymptotically optimal among all admissible policies. We further show that the gap between the cost of our policy and the lower bound program is tight.

\subsection{A Deterministic Lower Bound}\label{subsec:lower_bound}
We start with the lower bound. The LLR representation in Section~\ref{sec:admin_policy} makes it possible to separate the amount of evidence required from the way that evidence is acquired. Every query contributes a model-specific expected information gain under each hypothesis. Consequently, although a policy can follow many different histories, its expected evidence at stopping can be summarized by the expected number of times it uses each LLM.

Fix an admissible policy $\pi$ and LLM $j$. Define the number of queries sent to that LLM by time $t$ and stopping time $\tau$ as
\begin{align}\label{eq:def-nj}
N_j^{\pi}(t):=\sum_{s=1}^t \mathbf 1\{j_s^{\pi}=j\}.
\end{align}
We further define $n_{j,\theta}^{\pi} := \EE_\theta[N_j^{\pi}(\tau)]$ as the expected number of queries of LLM~$j$ under $\theta \in \{A,B\}$. The two vectors $\bs n_A^\pi$ and $\bs n_B^\pi$ describe the allocations generated by the same policy under different ground truths; they are not choices that the policy can make after observing $\theta$.

Under Assumption~\ref{assum:indep}, conditioning on the ground truth $\theta$, at stopping time $\tau$, policies with finite expected query counts satisfy
\begin{equation}\label{eq:L_tau}
\begin{aligned}
  \EE_A[L_\tau^{\pi}] &= \sum_{j=1}^M I_{j,A} \EE_A[N_j^{\pi}(\tau)] = \sum_{j=1}^M I_{j,A} n_{j,A}^{\pi}, \\
  \EE_B[-L_\tau^{\pi}] &= \sum_{j=1}^M I_{j,B} \EE_B[N_j^{\pi}(\tau)] = \sum_{j=1}^M I_{j,B} n_{j,B}^{\pi}.
\end{aligned}
\end{equation}
The negative sign under $B$ aligns evidence with the true hypothesis, so both sums are positive. These identities turn a requirement on the stopped LLR into a requirement on expected query allocations.

The next proposition formalizes this intuition and shows that the expected information gathered from the LLMs must be at least over a certain threshold.

\begin{proposition}
\label{prop:info-budget-asym}
There exists a constant $C_0>0$, independent of $\alpha$ and $\pi$, such that, for all sufficiently small $\alpha$ and every admissible policy satisfying $\EE_A[\tau],\EE_B[\tau]<\infty$,
\begin{equation}\label{eq:E-info-lb}
 \begin{aligned}
 \EE_A[L_\tau^{\pi}] &=\sum_{j=1}^M I_{j,A} n_{j,A}^{\pi}
  \ge S_A(\alpha) := A_\alpha - C_0, \\
  \EE_B[-L_\tau^{\pi}] & =\sum_{j=1}^M I_{j,B} n_{j,B}^{\pi} \ge S_B(\alpha) := B_\alpha - C_0.
  \end{aligned}
\end{equation}

\end{proposition}

Note that as $\alpha \to 0$, both $A_\alpha$ and $B_\alpha$ are in the order of $\log(1/\alpha)$. As a result, Proposition~\ref{prop:info-budget-asym} says that each hypothesis imposes an information budget of order $\log(1/\alpha)$ over any admissible policy. The budgets are slightly below the stopping thresholds (by a constant $C_0$) because a policy can occasionally stop at the wrong boundary. Under $A$, for example, those paths contribute negative rather than positive evidence to $\EE_A[L_\tau]$. Their probability is $O(\alpha)$ and the LLR overshoot is bounded, so their effect can be absorbed into a constant slack. Thus, $S_\theta(\alpha)=\log(1/\alpha)+O(1)$ for either hypothesis. The prior changes the constant offset in the required evidence, but not its leading logarithmic order. In the remainder of the lower-bound analysis, we assume that $\alpha$ is small enough such that both budgets are positive.

In the following proposition, we further rewrite the objective function in terms of $n_{j, A}$ and $ n_{j,B}$.
\begin{proposition}
\label{prop:cost-wait-decomp}
For any policy $\pi$ with $\EE_\theta[\tau]<\infty$ and $\theta\in\{A,B\}$, we have
\begin{align}
\label{eq:cost-decomp-theta}
\EE_\theta[C_\pi] = \sum_{j=1}^M c_j n_{j, \theta}^\pi,\quad \EE_\theta[W_\pi] = \sum_{j=1}^M \mu_j \, n_{j, \theta}^\pi
\end{align}
\end{proposition}

If we further use the expected values to replace the random variable in the objective function, we obtain the following deterministic function
\begin{equation}\label{eq:defn-F}
     F(\bs{n}_A^\pi, \bs{n}_B^\pi)   := \sum_{\theta \in \{A,B\} } \xi_\theta \left[ \sum_{j=1}^M c_j n_{j, \theta}^\pi  +  g \left(\sum_{j=1}^M \mu_j n_{j, \theta}^\pi\right) \right].
\end{equation}

Combining the above results, we propose the following convex problem:
\begin{align}
    \Psi_\alpha:= \min_{\bs{n}_A, \bs{n}_B} \ \ & F (\bs{n}_A, \bs{n}_B)\tag{LB}\label{LB2} \\
   \mbox {s.t.}  \quad   & \sum_{j=1}^M I_{j,\theta} n_{j,\theta} = S_\theta(\alpha), \ \ \ \forall \theta \in \{A,B\} \notag \\
 & n_{j,A}, n_{j,B} \ge 0  \ \ \ \forall j\in [M]. \notag
\end{align}

In the following proposition, we establish that $\Psi_\alpha$ is a tractable lower bound of the optimal objective value. \begin{proposition}\label{prop:lower_bound}
For $\alpha$ such that $S_A(\alpha),S_B(\alpha)>0$, every admissible policy satisfies $\mathcal R(\pi;\alpha)\ge\Psi_\alpha$. Moreover, $F(\bs n_A,\bs n_B)$ is convex in $(\bs n_A,\bs n_B)$.
\end{proposition}

Program~\eqref{LB2} has a natural interpretation from the perspective of acquiring information. A query to model $j$ supplies $I_{j,\theta}$ units of expected evidence under $\theta$, at monetary cost $c_j$ and mean delay $\mu_j$. The program determines the least expensive combination that supplies the required budget under each hypothesis. Two further simplifications make this benchmark optimistic. Only expected information is required, rather than an actual sample path that reaches a stopping boundary, and the waiting penalty is evaluated at each conditional mean, rather than at realized delay. Every admissible policy induces allocations feasible for the corresponding inequality-constrained relaxation, and its true cost is no smaller than their deterministic cost. These observations explain why the program is a lower bound. They also identify that an asymptotically optimal policy needs to learn which allocation to follow while keeping the costs of excessive evidence and noisy delay on the lower order. Note that we change the inequality constraints in \eqref{eq:E-info-lb} to equality constraints. This is intuitive since the weighted cost of any non-binding query allocation can be further reduced.

To study the structure of the allocation, it is useful to work with fractions of information rather than fractions of queries. Define
\begin{align*}
    w_{j, \theta} = \frac{I_{j, \theta}\cdot n_{j, \theta}}{S_\theta(\alpha)}.
\end{align*}
The fractions sum to one under each hypothesis. Allocating a fraction $w_{j,\theta}$ to a more informative model requires fewer queries, so this normalization allows models to be compared on a common evidence scale. Using the cost-efficiency and time-efficiency ratios defined in Section~\ref{sec:model}, let
\begin{equation*}
\label{eq:def-kappa-tau-avg}
  \kappa_\theta(\bs{w}_\theta) := \sum_{j=1}^M \kappa_{j,\theta} w_{j, \theta},
  \qquad
  \eta_\theta(\bs{w}_\theta) := \sum_{j=1}^M \eta_{j,\theta} w_{j, \theta}.
\end{equation*}
These are the average monetary cost and mean waiting time per unit of information. Using these definitions, we can rewrite the objective function of \eqref{LB2} in terms of $\bs{w}_{\theta}:=(w_{1, \theta}, w_{2, \theta},\dots, w_{M, \theta}), $ as follows:
\begin{align}
  G_\alpha(\bs{w}_A, \bs{w}_B) := \sum_{\theta\in\{A,B\}} \xi_{\theta} \left[  S_{\theta}(\alpha) \kappa_{\theta}(\bs{w}_{\theta}) +  g\! \left( S_{\theta}(\alpha) \eta_{\theta}(\bs{w}_{\theta}) \right) \right] . \label{eq:F-in-w}
\end{align}
\noindent
Thus, \eqref{LB2} can be equivalently written as the following compact formulation
\begin{align}
   \Psi_\alpha=  \min_{(\bs{w}_A, \bs{w}_B) \in \Delta_M \times \Delta_M} G_\alpha\left(\bs{w}_A,\bs{w}_B\right) \label{ALO}\tag{ALO}
\end{align}
\noindent
where $\Delta_M := \left\{\bs{w} \in \RR_+^M : \sum_{j=1}^M w_j = 1\right\}$ is the probability simplex in $ \RR^M$. Now, leveraging the Karush-Kuhn-Tucker (KKT)  conditions, we can establish the property of \eqref{eq:F-in-w} when $\alpha$ is relatively small.
\begin{proposition}
\label{prop:lower-mixture}
Let $c>0$. Under Assumptions~\ref{assum:indep}-~\ref{assum:g}, there exists a finite constant $\hat\alpha>0$, such that for all $\alpha<\hat\alpha$, we have
    \begin{align*}
        \min_{(\bs{w}_A, \bs{w}_B)\in\Delta_M\times\Delta_M}
        G_\alpha(\bs{w}_A, \bs{w}_B) = \min_{(j_A, j_B)\in[M]\times[M]} G_\alpha(\bs{e}_{j_A}, \bs{e}_{j_B}),
    \end{align*}
    where $\bs{e}_j$ is the $j$-th unit vector. Furthermore, the optimal entries satisfy
    \begin{equation*}
        j_{\theta}^\star\in \argmin_{j\in [M]} \left\{S_{\theta}(\alpha)\kappa_{j, \theta} + g\!\left(S_{\theta}(\alpha)\eta_{j, \theta}\right)\right\}, \,  \theta \in\{A,B\}.
    \end{equation*}
\end{proposition}
\vspace{2mm}

Proposition~\ref{prop:lower-mixture} establishes a useful simplification of the deterministic lower-bound problem: for sufficiently small $\alpha$, there exists an optimal allocation that assigns each hypothesis-specific information budget to a single LLM. The selected LLM balances the monetary cost and waiting penalty of supplying that budget. Specifically, if LLM $j$ supplies all $S_\theta(\alpha)$ units of information under hypothesis $\theta$, the relaxed allocation uses $S_\theta(\alpha)/I_{j,\theta}$ queries, incurring monetary cost $S_\theta(\alpha)\kappa_{j,\theta}$ and mean waiting time $S_\theta(\alpha)\eta_{j,\theta}$. The specialist therefore minimizes the combined cost of acquiring the required evidence, rather than being selected by its per-query price, accuracy, or response time alone.

Because both efficiency ratios depend on the hypothesis, the LLM offering the best cost--delay trade-off under $A$ need not be the one offering the best trade-off under $B$. This gives rise to the two specialists, which may coincide if the same LLM minimizes both criteria. The remaining challenge is to recover the benefit of these hypothesis-specific allocations through a policy that does not know the ground truth.

\subsection{Sign-based Two-LLM Policy}
Based on the deterministic lower bound in Section~\ref{subsec:lower_bound}, we design an asymptotically optimal admissible policy as $\alpha$ goes to zero. A classical idea in designing near-optimal control policy when a tractable benchmark is known is to mimic the behavior of the optimal solution to that benchmark.  In our setting, the obstacle is informational rather than computational: In \eqref{LB2}, we allocate the queries based on the underlying truth $\theta\in \{A, B\}$; in contrast, an admissible policy is nonanticipating and does not know $\theta$. We use the posterior to bridge this gap. When the evidence favors $A$, the policy queries the $A$-specialist; when it favors $B$, it queries the $B$-specialist.

Let $j_A^\star$ and $j_B^\star$ denote the $A$-specialist and the $B$-specialist defined by Proposition~\ref{prop:lower-mixture}, respectively. In the following definition, we describe a sign-based two-LLM sequential (\algo{}) policy.
\vspace{2mm}
\begin{definition}[Sign-based two-LLM Sequential Policy]\label{def:2llm}
Fix an error tolerance $\alpha$ and the optimal pair of LLMs $(j_A^\star,j_B^\star)$ in Proposition~\ref{prop:lower-mixture}.
\begin{itemize}
    \item \emph{Initialization:} Set $L_0=0$.
    \item \emph{LLM selection and update:} Before stopping, at period $t\ge 1$, choose
    \[
    j_t =
    \begin{cases}
    j_A^\star, & \text{if } \delta + L_{t-1}\ge 0,\\[3pt]
    j_B^\star, & \text{if } \delta + L_{t-1}< 0,
    \end{cases}
    \]
    observe $Y_t$ from LLM $j_t$, and update
    $L_t = L_{t-1} + \ell_{j_t}(Y_t).$
    \item \emph{Stopping rule and decision:} Stop at the time
    \[
    \tau := \inf\{t\ge 0: L_t \ge A_\alpha \ \text{or}\ L_t \le -B_\alpha\},
    \]
    and set
    \[
    D =
    \begin{cases}
    A, & \text{if } L_\tau \ge A_\alpha,\\
    B, & \text{if } L_\tau \le -B_\alpha.
    \end{cases}
    \]
\end{itemize}
\end{definition}

According to Definition~\ref{def:2llm}, the choice of LLM at each step depends on the sign of the posterior log-odds $L_t+\delta$ (equivalent to a posterior probability $1/2$), so queries are directed toward the specialist aligned with the current direction of evidence. Policy~\algo{}  has a simple interpretation. When the accumulated evidence favors $A$, the policy queries the LLM that is most efficient at gathering information in favor of $A$; otherwise, it queries the LLM that is most efficient for certifying $B$. When the posterior is close to $1/2$, the LLR is close to $-\delta$, and the policy may switch between the two specialists as the sign changes. As the evidence becomes more decisive, the policy increasingly queries the specialist associated with the favored hypothesis while retaining the same sign-based rule throughout. The policy uses mean waiting times to select the specialists but does not use realized waiting times in online routing.

In the following, we refer to a~\algo{} policy that only queries $j_A^\star$ and $j_B^\star$ as $\pi(j_A^\star, j_B^\star)$ whenever we want to make the choice of LLMs explicit.

\subsection{Asymptotic Optimality of Policy~\algo{}}

We next establish that this simple routing rule attains the deterministic lower bound up to a lower-order term. This is stronger than showing that the policy performs well within a restricted two-model class: the comparison is with every admissible policy, including those that use all available models and react to realized delays.

\begin{theorem}[Asymptotic optimality]
\label{thm:two-llm-upper}
Let $c>0$. As $\alpha\to0$, we have $\Psi_\alpha=\Theta((\log(1/\alpha))^\rho)$. Moreover, a pair $(j_A^\star(\alpha),j_B^\star(\alpha))$ selected as in Proposition~\ref{prop:lower-mixture} yields a policy $\pi:=\pi(j_A^\star(\alpha),j_B^\star(\alpha))$ satisfying, as $\alpha\rightarrow 0$,
\begin{align}
\mathcal{R}(\pi ;\alpha)  = \Psi_\alpha \;+\; O\!\Big((\log(1/\alpha))^{\rho-1}\Big).
\label{eq:two-llm-additive-rate}
\end{align}
\end{theorem}

Let $\mathcal R^*(\alpha):=\inf_{\pi\in\Pi(\alpha)}\mathcal R(\pi;\alpha)$. Since $\Psi_\alpha\le\mathcal R^*(\alpha)\le\mathcal R(\pi;\alpha)$, Theorem~\ref{thm:two-llm-upper} implies
\[
\begin{aligned}
0\le\mathcal R(\pi;\alpha)-\mathcal R^*(\alpha)
&=O\!\left((\log(1/\alpha))^{\rho-1}\right),\\
\frac{\mathcal R(\pi;\alpha)}{\mathcal R^*(\alpha)}
&=1+O\!\left(\frac1{\log(1/\alpha)}\right).
\end{aligned}
\]
With linear waiting cost, the additive gap is bounded. With a superlinear penalty, the additive bound can grow as confidence increases, but it remains small relative to the optimal total cost.

Theorem~\ref{thm:two-llm-upper} establishes that Policy~\algo{} is asymptotically optimal as the target error level $\alpha \rightarrow 0$. This conclusion has both theoretical and operational implications. Theoretically, it reveals a simple structure in the high confidence regime. Although the decision-maker may have access to many heterogeneous LLMs, near-optimal performance can be achieved by using at most two specialists. The complexity of sequential hypothesis experimentation, therefore, needs \textit{not} to grow with the number of available LLMs. Moreover, we note that Policy~\algo{} does not use the realized waiting time but achieves asymptotic optimality among all latency-adaptive policies.

Operationally, the theorem provides direct guidance for policy design. To implement a near optimal strategy, the  decision-maker only needs to identify one strong model for certifying $A$ and another for certifying $B$, and then switch between them as the posterior evolves. Adding more LLMs expands the set of candidate specialists and may potentially reduce the cost if a more efficient specialist is introduced, but it does not require more complex orchestration. As a result, the policy remains simple and scalable under stringent reliability requirements.

Our proof follows a matching-bounds approach centered around the deterministic universal lower bound $\Psi_\alpha$ discussed in Section~\ref{subsec:lower_bound}. We show that, as $\alpha \to 0$, the expected query cost $\EE[C_\pi]$ and the expected waiting time penalty  $\EE[g(W_\pi)]$ induced by $\pi = \pi(j_A^\star, j_B^\star) $ match their counterpart in the lower bound $\Psi_\alpha$ up to lower-order terms, respectively. To do this, we decompose the cumulative log likelihood ratio (LLR) into a drift component that matches the counterpart in the lower bound and a martingale component that captures the stochasticity. Under $\theta\in\{A,B\}$, let
\begin{align}
  M_t^{(A)} \;=\;  L_t \;-\; \sum_{s=1}^t I_{j_s,A} \text{ and } M_t^{(B)} \;=\; -L_t  \; - \;\sum_{s=1}^t I_{j_s,B}. \label{decomp}
\end{align}
We show that (see Lemma~\ref{lem:drift-decomp}) ${M_t^{(\theta)}}$ is a martingale under $\PP_\theta$ with bounded increments. This decomposition isolates the predictable evidence accumulation term, which directly relates to the deterministic benchmark, from a martingale fluctuation term that captures overshoot and deviation effects at the stopping time. Under $A$, for example, every chosen model contributes a strictly positive expected increment to $L_t$. Remaining on the wrong side of the routing boundary after many queries therefore requires a large unfavorable martingale fluctuation. Using a concentration argument, we show that the expected number of rounds spent on the wrong side is $O(1)$. The same reasoning applies under $B$. The final query may also supply slightly more evidence than is needed to stop. Accounting for these effects, the policy's expected query cost and mean waiting time remain within a constant of their respective benchmark values under each hypothesis.

The second part of the proof controls the fluctuations of the waiting times accumulated over those queries. When the waiting penalty is nonlinear, finishing early on some paths does not fully offset the additional cost of finishing late on others. Matching the benchmark's mean delay is therefore not sufficient. The proof must also control two sources of variation: how many queries are needed before the evidence is decisive, and how long each query takes. The concentration bounds for accumulated evidence control the first, while the sub-Gaussian waiting-time assumption controls the second. Together, they ensure that unusually long query sequences do not dominate the expected waiting cost.

\subsection{Tightness of the Asymptotic Order}

The preceding theorem leaves a natural question: is its remainder caused by the simplicity of~\algo{}, or by an unavoidable difference between the lower bound program and the optimal policy? The next result shows that, for strictly convex power penalties, even an exact optimal policy remains separated from $\Psi_\alpha$ by a term of this order.

\begin{theorem}[Tightness of the convergence rate]\label{thm:two-llm-tight}
Suppose $c>0$ and $\rho>1$. For every fixed problem instance satisfying Assumptions~\ref{assum:indep}--\ref{assum:g}, there exist constants $K>0$ and $\underline\alpha>0$ such that, for all $\alpha\in(0,\underline\alpha)$,
\begin{equation}
\label{eq:rate-tightness}
    \inf_{\pi\in\Pi(\alpha)}
    \mathcal R(\pi;\alpha)
    \ge
    \Psi_\alpha
    +
    K(\log(1/\alpha))^{\rho-1}.
\end{equation}
\end{theorem}

Theorem~\ref{thm:two-llm-tight} shows that the remaining gap to the benchmark is not simply a consequence of restricting the policy to two LLMs, but a fundamental gap between an optimal policy and the deterministic benchmark because the latter neglects uncertainty in its formulation. To understand this difference, decompose the gap for an admissible policy $\pi$ as
\[
\begin{aligned}
\mathcal R(\pi;\alpha)-\Psi_\alpha
={}&
\underbrace{
F(\bs n_A^\pi,\bs n_B^\pi)-\Psi_\alpha
}_{\text{allocation gap}} +
\underbrace{
\sum_{\theta\in\{A,B\}}\xi_\theta
\left[
\EE_\theta[g(W_\pi)]
-g\!\left(\EE_\theta[W_\pi]\right)
\right]
}_{\text{fluctuation gap}}.
\end{aligned}
\]
The allocation gap measures the gap between policy $\pi$'s expected query allocation costs and that under the optimal deterministic allocation. It reflects both which models are queried and how much evidence is collected before stopping. The fluctuation gap is the additional waiting penalty caused by variation in completion time. With a strictly convex penalty, finishing earlier on some paths does not fully offset the extra cost of finishing later on others.

Given our construction, minimizing the total weighted cost requires striking the right balance between these two gaps. For example, a policy that always queries the optimal deterministic LLMs has no allocation gap, but its waiting time is essentially a random sum over $\Theta(\log(1/\alpha))$ queries, and its fluctuation gap is of order $(\log(1/\alpha))^{\rho-1}$. On the other hand, if we want to reduce the variability of the waiting time, any policy will need to query another LLM with different delay profiles. Any such query is not time-efficient, and since the total waiting time grows in the order of $\log(1/\alpha)$, the marginal delay penalty grows in the order of $(\log(1/\alpha))^{\rho-1}$.

\section{Value of Querying Two LLMs} \label{subsec:one-vs-two}

The previous section establishes that two LLMs are sufficient for asymptotic optimality, but sufficiency does not imply that both are necessary. We now ask when a single model can attain the same leading-order performance and how much is lost when it cannot. A single-model policy is still sequential: it can stop as soon as the posterior meets the confidence requirement. Thus, this section isolates the value of adapting \emph{which} model is queried, rather than the value of adapting \emph{when} to stop.

Let $\Pi^{(1)}(\alpha)$ be the class of policies that select one LLM before observing any outputs, use that model for every query, and apply the posterior stopping rule in Section~\ref{sec:admin_policy}. The model can be chosen using the prior, error tolerance, and all known model characteristics, but cannot depend on the unknown ground truth. Motivated by the lower bound program~\eqref{LB2}, we define the following lower bound.
\begin{equation}
\label{eq:Phi-one-llm}
  \Psi_\alpha^{(1)}
  := \min_{1\le j\le M}
     \Bigl\{
       \xi_A S_A(\alpha)\,\kappa_{j,A}
      +\xi_B S_B(\alpha)\,\kappa_{j,B}
      + \xi_A g\left(S_A(\alpha)\,\eta_{j,A}\right) + \xi_B g\left(S_B(\alpha)\,\eta_{j,B}
        \right)
     \Bigr\}.
\end{equation}
This is the lowest cost of supplying both hypothesis-specific information budgets with the same model. It allows the required number of queries to differ under $A$ and $B$, but not the queried LLM. Equivalently, it restricts \eqref{eq:F-in-w} to the diagonal pairs $(\bs e_j,\bs e_j)$. Therefore, the lower-bound argument in Proposition~\ref{prop:lower_bound} gives
\[
\inf_{\pi\in\Pi^{(1)}(\alpha)}\mathcal R(\pi;\alpha)
\ge\Psi_\alpha^{(1)}\ge\Psi_\alpha.
\]
The single-LLM case of the proof of Theorem~\ref{thm:two-llm-upper} also shows that repeatedly querying a minimizer in \eqref{eq:Phi-one-llm} attains $\Psi_\alpha^{(1)}$ up to $O((\log(1/\alpha))^{\rho-1})$. Comparing the two deterministic benchmarks therefore reveals the leading-order value of allowing hypothesis-specific specialization, rather than merely comparing two potentially loose bounds.

A single LLM suffices when the two hypothesis-specific benchmark problems have a common optimal model. In that case, for some $j^\sharp$ and all sufficiently small $\alpha$,
\begin{equation}\label{eq:one-llm-optimal}
\Psi_\alpha^{(1)}
=G_\alpha(\bs e_{j^\sharp},\bs e_{j^\sharp})
=\Psi_\alpha.
\end{equation}
A sufficient condition is that this model weakly dominates every other model in both cost-efficiency and time-efficiency under both hypotheses:
\[
\kappa_{j^\sharp,A}\le\kappa_{j,A},\qquad
\kappa_{j^\sharp,B}\le\kappa_{j,B},\qquad
\eta_{j^\sharp,A}\le\eta_{j,A},\qquad
\eta_{j^\sharp,B}\le\eta_{j,B},\qquad j\in[M].
\]
The two specialists may then be chosen as $j_A^\star=j_B^\star=j^\sharp$, and, for $c>0$, the resulting single-model policy is asymptotically optimal by Theorem~\ref{thm:two-llm-upper}. This conclusion is useful because it identifies when routing has no leading-order benefit: one model supplies evidence efficiently regardless of which hypothesis turns out to be true.

In many of the other scenarios, however, there does not exist a single LLM that achieves the lowest per-unit information cost and waiting time under both hypotheses simultaneously.  For instance, suppose there are two LLMs such that $\kappa_{1,A} \ll \kappa_{2,A}$ and $\eta_{1,A} \ll \eta_{2,A}$, but $\kappa_{2,B} \ll \kappa_{1,B}$ and $\eta_{2,B} \ll \eta_{1,B}.$ In other words, LLM~$1$ is much more efficient under hypothesis~$A$ (since it provides information faster and cheaper) while LLM~$2$ is much more efficient under hypothesis~$B$.  Any single-LLM design must compromise between these conflicting efficiencies and will incur a larger value that is at least $\Psi_\alpha^{(1)}$. In fact, we can construct a family of examples where the gap between $\Psi_\alpha^{(1)}$ and $\Psi_\alpha$ is very large.

\vspace{1mm}
\begin{example}[Quantifying the value of querying two LLMs]
\label{ex:two-llm-advantage}
Let $c=1$ such that $g(x)=x^\rho$. Consider a symmetric prior $\xi_A = \xi_B = 1/2$ and two LLMs defined in Table~\ref{tab:parameter}
\begin{table}[htbp]
\centering
\caption{Parameters for Example~\ref{ex:two-llm-advantage}}\label{tab:parameter}
\begin{tabular}{cccccc}
\toprule
LLM $j$ & $I_{j,A}$ & $I_{j,B}$ & $c_j$ & $\mu_j$ & Specialization\\
\midrule
1 & 5 & 1 & 10 & 10 & $A$-specialist\\
2 & 1 & 5 & 10 & 10 & $B$-specialist\\
\bottomrule
\end{tabular}
\end{table}
\end{example}

Let $S(\alpha):=S_A(\alpha)=S_B(\alpha)=\log((1-\alpha)/\alpha)-C_0$, which is positive for sufficiently small $\alpha$. Selecting either model as the sole source gives
\[
\Psi_\alpha^{(1)}=6S(\alpha)
+\tfrac12g(2S(\alpha))+\tfrac12g(10S(\alpha)).
\]
Under its favorable hypothesis, the selected model requires benchmark mean delay $2S(\alpha)$; under the other hypothesis, it requires $10S(\alpha)$. By contrast, model 1 dominates under $A$ and model 2 dominates under $B$, so the two-specialist benchmark is
\[
\Psi_\alpha=2S(\alpha)+g(2S(\alpha)).
\]
Consequently,
\[
\frac{\Psi_\alpha^{(1)}}{\Psi_\alpha}
=\frac{6S(\alpha)+\tfrac12(2S(\alpha))^\rho+\tfrac12(10S(\alpha))^\rho}
{2S(\alpha)+(2S(\alpha))^\rho}
\longrightarrow\frac{1+5^\rho}{2}.
\]
The best single-model policy therefore has asymptotic cost three times as large when $\rho=1$, and thirteen times as large when $\rho=2$. The fivefold disadvantage under the unfavorable hypothesis is paid on half the instances. A convex waiting penalty amplifies this disadvantage, which explains why the ratio is much larger under quadratic than linear waiting cost. Since the models have identical per-query costs and mean latencies, the gain comes entirely from obtaining more useful evidence per query under the relevant hypothesis.

The ratios can be made arbitrarily large by increasing the contrast between the two directional information rates. Such contrasts are compatible with the accuracy model: keeping one accuracy strictly between $1/2$ and 1 while letting the other approach 1 makes one directional KL divergence grow without bound while the reverse divergence stays bounded; swapping the accuracies produces a complementary model. From a managerial perspective, this example shows us that the value of a second LLM depends on complementarity, not diversity alone. Access to another model is valuable when it improves information acquisition in a direction that the first model handles inefficiently. $\hfill \square$

\section{Value of Adaptivity}\label{sec:adaptivity}

The single-model policies in Section~\ref{subsec:one-vs-two} still adapt their stopping time to the observed evidence. We now compare sequential querying with a static design that fixes every model's query count before observing any outputs. Throughout this section, we retain the special case of zero waiting cost, $g(w)\equiv0$, or equivalently $c=0$ in Assumption~\ref{assum:g}. The objective becomes
\[
\mathcal R(\pi;\alpha)=\EE[C_\pi]
=\EE\!\left[\sum_{t=1}^{\tau}c_{j_t}\right].
\]
Waiting times may still be positive, but their penalty is zero. This specialization isolates the monetary value of responding to the evidence without combining it with the curvature of the delay penalty. The static benchmark is allowed to choose both the total query count and its allocation across all available models. Any advantage of sequential querying, therefore, reflects the ability to react to outputs, rather than a comparison with an arbitrarily chosen fixed budget.

Formally, a static design is a vector $\bs n=(n_1,\ldots,n_M)\in\{0,1,\ldots\}^M$, where $n_j$ is the number of queries planned for LLM $j$. All planned queries are completed before the terminal decision is made. Under Assumption~\ref{assum:indep}, their order does not affect the distribution of accumulated evidence; with $g\equiv0$, the cost is simply $\sum_j c_jn_j$. Let $\omega:=\sum_jn_j$, and define the positive, hypothesis-aligned mean evidence and its variance by
\begin{align*}
\mu_A(\bs n)&:=\EE_A[L_\omega]=\sum_{j=1}^M n_j I_{j,A},\qquad
\mu_B(\bs n):=\EE_B[-L_\omega]=\sum_{j=1}^M n_j I_{j,B},\\
\sigma_\theta^2(\bs n)&:=\Var_\theta(L_\omega)=\sum_{j=1}^M n_j V_{j,\theta},\qquad
V_{j,\theta}:=\Var_\theta(\ell_j(Y))>0.
\end{align*}
After the $\omega$ queries, the design selects the hypothesis with the larger posterior probability:
\begin{equation}\label{eq:static_decision}
D:=\begin{cases}
A,&L_\omega+\delta\ge0,\\
B,&L_\omega+\delta<0.
\end{cases}
\end{equation}
We require $\PP(D\ne\theta)\le\alpha$ and denote the feasible static designs by $\Pi^{\mathrm{st}}(\alpha)$. This unconditional Bayes-error requirement is weaker than requiring posterior error at most $\alpha$ on every stopped history, as in the sequential class. In other words, the comparison does not impose the sequential class' stronger sample-path level requirement on static designs. The static policy class $\Pi^{\mathrm{st}}(\alpha)$ is, therefore, not a subclass of the admissible sequential policy $\Pi(\alpha)$.

Let $\mathcal R^{\mathrm{st},*}(\alpha):=\inf_{\pi\in\Pi^{\mathrm{st}}(\alpha)}\mathcal R(\pi;\alpha)$ denote the optimal static cost. Write $\xi_{\min}:=\min\{\xi_A,\xi_B\}$ and take $\alpha<\xi_{\min}$, so a feasible design must make at least one query. Since the Bayes error is $\xi_A\PP_A(D=B)+\xi_B\PP_B(D=A)$, feasibility implies
\begin{equation}
\PP_A(L_\omega+\delta<0)\le\frac{\alpha}{\xi_A},\qquad
\PP_B(L_\omega+\delta\ge0)\le\frac{\alpha}{\xi_B}.
\end{equation}
These are necessary conditions, not a sufficient pair of conditions for Bayes error at most $\alpha$.

To understand what a fixed query budget must achieve, consider its terminal evidence relative to its fluctuations. Under a static design, the LLR is a sum of independent, though generally nonidentically distributed, increments. A normal approximation suggests
\[
\begin{aligned}
\PP_A(L_\omega+\delta<0)&\approx
\Phi\!\left(-\frac{\mu_A(\bs n)+\delta}{\sigma_A(\bs n)}\right),\\
\PP_B(L_\omega+\delta\ge0)&\approx
\Phi\!\left(-\frac{\mu_B(\bs n)-\delta}{\sigma_B(\bs n)}\right),
\end{aligned}
\]
where $\Phi$ is the standard normal distribution function. Thus, a large expected amount of evidence is not sufficient by itself: it must also be large relative to its standard deviation. This distinction is important because a fixed horizon must control the probability of an unfavorable evidence realization at that horizon. The following proposition retains the Berry--Esseen approximation error, making the necessary conditions rigorous rather than treating the normal approximation as exact.

\begin{proposition}\label{prop:socp_static_lb_partI}
Consider a static design with at least one query and terminal decision \eqref{eq:static_decision}. Define
\[
\mathcal T_\theta(\bs n):=\sum_{j=1}^M n_j\EE_\theta
\left[\left|\ell_j(Y)-\EE_\theta[\ell_j(Y)]\right|^3\right].
\]
If $\PP(D\ne\theta)\le\alpha$, then
\begin{equation}\label{eq:normal_constraint}
\begin{aligned}
\frac{\mu_A(\bs n)+\delta}{\sigma_A(\bs n)}
&\ge\nu\!\left(\frac{\alpha}{\xi_A}
+\frac{\mathcal T_A(\bs n)}{\sigma_A(\bs n)^3}\right),\\
\frac{\mu_B(\bs n)-\delta}{\sigma_B(\bs n)}
&\ge\nu\!\left(\frac{\alpha}{\xi_B}
+\frac{\mathcal T_B(\bs n)}{\sigma_B(\bs n)^3}\right),
\end{aligned}
\end{equation}
where
\[
\nu(x):=\begin{cases}
\infty,&x\le0,\\
\Phi^{-1}(1-x),&0<x<1,\\
-\infty,&x\ge1.
\end{cases}
\]
\end{proposition}

The correction $\mathcal T_\theta(\bs n)/\sigma_\theta(\bs n)^3$ measures the possible error of the normal approximation. With very few queries, this correction can be large enough that the constraints are uninformative. We therefore also retain a necessary information condition as follows: For every feasible design,
\[
\PP_A(D=B)+\PP_B(D=A)\le\frac{\alpha}{\xi_{\min}}.
\]
The distributions of the complete output vector under $A$ and $B$ must consequently have total variation distance at least $1-\alpha/\xi_{\min}$. Pinsker's inequality in either direction gives
\[
\mu_\theta(\bs n)\ge b_\alpha
:=2\left(1-\frac{\alpha}{\xi_{\min}}\right)^2>0,
\qquad\theta\in\{A,B\}.
\]
This requirement rules out designs whose cumulative signal is too small to distinguish the hypotheses, even when the normal-approximation constraints happen to be vacuous. Combining these necessary conditions and relaxing integrality yields
\begin{equation}\label{eq:static_normal}
\begin{aligned}
\Psi_\alpha^{\mathrm{st}}:=\min_{\bs n\in\RR_+^M}\quad
&\sum_{j=1}^M c_j n_j\\
\text{s.t.}\quad
&\frac{\mu_A(\bs n)+\delta}{\sigma_A(\bs n)}
\ge\nu\!\left(\frac{\alpha}{\xi_A}
+\frac{\mathcal T_A(\bs n)}{\sigma_A(\bs n)^3}\right),\\
&\frac{\mu_B(\bs n)-\delta}{\sigma_B(\bs n)}
\ge\nu\!\left(\frac{\alpha}{\xi_B}
+\frac{\mathcal T_B(\bs n)}{\sigma_B(\bs n)^3}\right),\\
&\mu_\theta(\bs n)\ge b_\alpha,\qquad\theta\in\{A,B\}.
\end{aligned}
\end{equation}
Every feasible static design is feasible for this program, so $\Psi_\alpha^{\mathrm{st}}\le\mathcal R^{\mathrm{st},*}(\alpha)$. The positive information constraints also exclude the zero-query vector, for which the standardized expressions are undefined.

We next study a weak-signal regime in which each individual output is only slightly more informative about the truth. This regime makes the distinction between fixed-budget and sequential querying especially clear: many observations are needed, and the cost of continuing after a favorable evidence realization can be compared on the same scale as the total query budget.

\vspace{1mm}
\begin{definition}[Weak-signal scaling]
Fix $M$, positive constants $a_j,b_j$ for $j\in[M]$, the prior, monetary query costs, and waiting-time distributions. For
\[
0<\epsilon<\frac{1}{2\max_{j\in[M]}\{a_j,b_j\}},
\]
set
\[
\gamma_{j,A}=\tfrac12+a_j\epsilon,\qquad
\gamma_{j,B}=\tfrac12+b_j\epsilon,\qquad j=1,\ldots,M.
\]
\end{definition}
\vspace{1mm}

Weak-signal scaling means taking $\epsilon\to0$ while these other primitives remain fixed. As $\epsilon$ decreases, the output distributions under the two hypotheses approach each other. The following theorem further characterizes the value of adaptivity explicitly in a limiting regime.

\vspace{1mm}
\begin{theorem}[Value of adaptivity with zero waiting cost]\label{thm:VoA}
Suppose $g(w)\equiv0$, and let $\mathcal R^*(\alpha):=\inf_{\pi\in\Pi(\alpha)}\mathcal R(\pi;\alpha)$. Under weak-signal scaling,
\[
\begin{aligned}
\lim_{\alpha\to0}\liminf_{\epsilon\to0}
\frac{\mathcal R^{\mathrm{st},*}(\alpha)}{\mathcal R^*(\alpha)}
 = \lim_{\alpha\to0}\limsup_{\epsilon\to0}
\frac{\mathcal R^{\mathrm{st},*}(\alpha)}{\mathcal R^*(\alpha)} = 4.
\end{aligned}
\]
\end{theorem}
\vspace{1mm}

The theorem establishes a factor-of-four cost difference when outputs first become weakly informative and the error tolerance then approaches zero. To see where the constant four comes from, write $s_j:=a_j+b_j$ . The expansions in the appendix give
\[
I_{j,\theta}=2s_j^2\epsilon^2+O(\epsilon^4),\qquad
V_{j,\theta}=4s_j^2\epsilon^2+O(\epsilon^4),
\qquad\theta\in\{A,B\}.
\]
For a static design that uses model $j$, the standardized mean evidence is approximately $s_j\epsilon\sqrt n$. To make the terminal error small, this quantity must reach a normal quantile of order $\sqrt{2 \log(1/\alpha) }$. A sequential policy instead stops when the realized evidence reaches a boundary of order $\log(1/\alpha)$; dividing that requirement by the drift gives its leading expected query count. The two scales are therefore
\[
\begin{aligned}
\text{Static query count}&\ \approx\ \frac{2\log(1/\alpha)}{s_j^2\epsilon^2},\\
\text{Sequential expected query count}&\ \approx\ \frac{\log(1/\alpha)}{2s_j^2\epsilon^2}.
\end{aligned}
\]
These are leading-scale interpretations in the ordered limits, not finite-parameter formulas. Their ratio explains the factor four. A fixed horizon must be long enough to control unfavorable evidence fluctuations, and that entire budget is paid even on paths that become decisive much earlier. Sequential stopping uses more queries on difficult paths and fewer on favorable ones, avoiding the cost of continuing after the required confidence has been reached.

Together, Sections~\ref{subsec:one-vs-two} and~\ref{sec:adaptivity} identify two distinct benefits of responding to evidence. Model selection exploits its direction, assigning queries to a specialist for the more plausible hypothesis. Sequential stopping exploits its amount, ending the process once enough evidence has accumulated. Efficient orchestration can benefit from both, even though the two sections isolate them in different comparisons.

\section{Numerical Experiments}
\label{sec:numerics}

In this section, we empirically evaluate the performance of Policy~\algo{} using a pool of 18 distinct LLMs. We consider a binary classification task based on the $\mu$-MATH (Meta U-MATH) benchmark \citep{chernyshev-etal-2025-u}, which contains 1,084 labeled tasks, each containing a university-level math question, one candidate solution, and ground truth labels for the correctness of the solution. The ground-truth correctness labels were obtained through a combination of automatic verification and manual annotation by mathematical experts. LLMs are asked to judge whether the candidate solution is correct or not without knowing the ground truth. The selected 18 LLMs exhibit heterogeneous verification accuracy, waiting time distribution, and inference cost. We measure these model-specific quantities based on their performance and API usage in our experiments.

We identify hypothesis $A$ with a correct candidate solution and hypothesis $B$ with an incorrect candidate solution. Correct and incorrect instances are sampled in equal proportions, so we use the symmetric prior $\xi_A=\xi_B=1/2$. Hence $\delta=0$ and the two stopping boundaries are symmetric:
\[
A_\alpha=B_\alpha=H_\alpha:=\log\frac{1-\alpha}{\alpha}.
\]
The waiting-time penalty is $g(w)=C_w w^\rho$, with $C_w>0$ and $\rho\ge1$. Under this task mapping, the true-positive rate of LLM $j$ is $\gamma_{j,A}$ and its true-negative rate is $\gamma_{j,B}$; together they determine the two information rates $I_{j,A}$ and $I_{j,B}$. The remaining inputs are the per-question cost $c_j$ and mean response time $\mu_j$.

\begin{figure}[hb]
    \centering
    \includegraphics[width=0.9\linewidth]
    {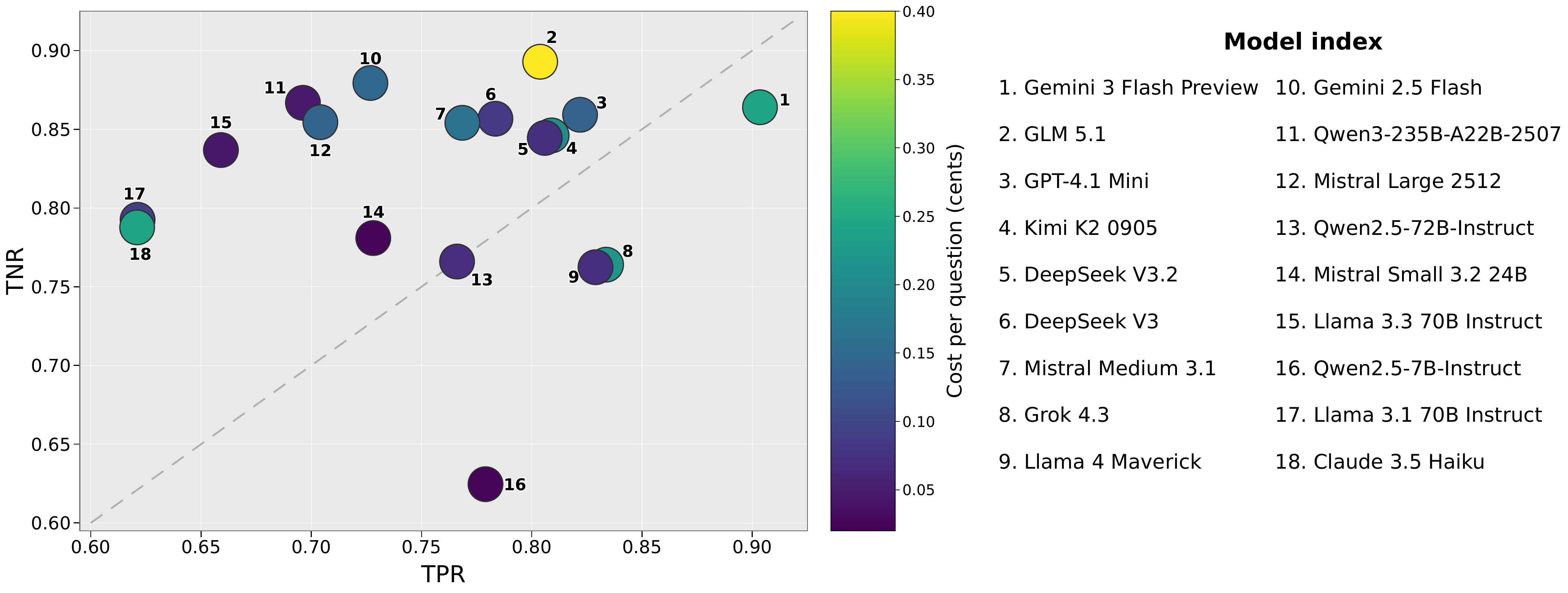}
    \caption{Model characteristics for the 18 LLMs used in the numerical experiments.}
    \label{fig:numerical-model-characteristics}
\end{figure}

In Figure~\ref{fig:numerical-model-characteristics}, the horizontal axis reports the true-positive rate $\gamma_{j,A}$, the vertical axis reports the true-negative rate $\gamma_{j,B}$, the color represents the per-question cost, and the numbers identify the corresponding models. The figure shows substantial heterogeneity along both accuracy dimensions, while monetary cost is not monotonically related to either one. This is precisely the setting in which a single accuracy ranking is insufficient: a model can be informative for certifying one hypothesis without being as informative for the other. At the same time, Figure~\ref{fig:numerical-model-characteristics} alone does not determine the specialists, because specialist selection depends on information-normalized monetary cost and latency, not on raw accuracy or price in isolation.

Proposition~\ref{prop:lower-mixture} motivates selecting one specialist for each hypothesis by balancing monetary cost and mean latency per unit of information. Under the symmetric prior, the lower-bound budgets in Section~4 are $S_A(\alpha)=S_B(\alpha)=H_\alpha-C_0$. Recall that the constant $C_0$ is a conservative slack introduced by the proof of the universal lower bound. For the numerical implementation, we simply set $C_0 = 0$, which preserves the leading high-confidence scale because $S_\theta(\alpha)=H_\alpha+O(1)$ while avoiding calibration of a proof-specific slack constant.

Specifically, for $\theta\in\{A,B\}$ we select
\begin{equation*}
 j_{\theta}^{\star}(\alpha,C_w,\rho)
 \in
 \operatorname*{arg\,min}_{j\in[M]}
 \left\{
 \frac{H_\alpha c_j}{I_{j,\theta}}
 +
 C_w\left(\frac{H_\alpha\mu_j}{I_{j,\theta}}\right)^\rho
 \right\}.
\end{equation*}
We refer to these as the $A$- and $B$-specialists. Once the pair is selected, online routing follows Definition~\ref{def:2llm}: because $\delta=0$, the policy queries the $A$-specialist when the current LLR is nonnegative and the $B$-specialist otherwise. Realized waiting times affect the realized cost, but not the online routing decision.

We evaluate policy performance using pathwise Monte Carlo simulation. For each simulated path, we first sample the true hypothesis according to the equal prior. At each query, the policy selects the $A$- or $B$-specialist according to the sign of the current LLR, samples the selected model's binary response using its true-positive and true-negative rates, samples a model-specific waiting time, updates the LLR, and continues until one of the stopping boundaries is crossed.

For a path with a stopping time $\tau$, the realized total cost is
\[
\sum_{t=1}^{\tau}c_{j_t}
+
C_w\left(\sum_{t=1}^{\tau}T_{j_t,t}\right)^\rho.
\]
Averaging this quantity over simulated paths directly estimates the expected nonlinear waiting-time cost of Policy~\algo{}. For the experiments with $\rho=1$ and varying $C_w$, we use 200,000 simulated paths per parameter setting. For the experiments with varying $\rho$, we use 1,000,000 simulated paths per parameter setting. The sample size is larger for the latter because the convex waiting-time penalty for $\rho>1$ increases the variability of the simulated path costs. Further implementation details, including how we construct the waiting-time distributions, are provided in Appendix~\ref{append:monte-carlo}.

To assess the performance of~\algo{} under finite $\alpha$, we compute a benchmark based on solving approximate Bellman equations, where we replace the random waiting time with its mean, i.e., $T_{j,t} \equiv \mu_j$ for all $j$ and $t$. This deterministic-latency model allows the decision maker to use all 18 LLMs and preserve the nonlinear dependence on accumulated delay, but it removes the model-specific waiting-time fluctuations. In particular, when $\rho=1$, the expected waiting penalty is linear, so replacing the random waiting time by its mean does not change the expected objective. When $\rho>1$, the waiting penalty depends on the accumulated waiting time. Because the original stochastic policy class may also condition future choices on realized delays, we use this value as a deterministic-latency benchmark rather than as a formally exact lower bound for the full stochastic control problem. We further approximately solve the Bellman equations by discretizing the state space; the precise Bellman recursion and computation details are reported in Appendix~\ref{append:bellman-defn}.

\subsection{Numerical Results}

We report the percentage relative difference between the pathwise Monte Carlo estimate of the two-LLM policy and the corresponding Bellman benchmark:
\[
\operatorname{Gap}(\alpha,\rho)
=
100\%
\times
\frac{
\widehat{\mathcal R}_{\algo{}}^{\mathrm{MC}}(\alpha,\rho)
-
\mathcal R_{\mathrm{Bellman}}^{\mathrm{det}}(\alpha,\rho)
}{
\mathcal R_{\mathrm{Bellman}}^{\mathrm{det}}(\alpha,\rho)
}.
\]
Here, $\widehat{\mathcal R}_{\algo}^{\mathrm{MC}}$ is the pathwise Monte Carlo estimate for the sign-based two-LLM policy and $\mathcal R_{\mathrm{Bellman}}^{\mathrm{det}}$ is the value of the Bellman problem evaluated with model-specific mean latencies. The interpretation of this quantity follows the benchmark discussion above. When $\rho=1$, expected waiting time is preserved when using the average value, so the gap directly compares Policy~\algo{} with a lower bound on the optimal policy (up to Bellman-grid discretization and Monte Carlo error). When $\rho>1$, the numerator compares the stochastic pathwise evaluation of Policy~\algo{} with a deterministic Bellman benchmark.

We first vary the waiting-cost weight $C_w$ while fixing $\rho=1$. This experiment changes the relative importance of waiting time while retaining the single-dimensional Bellman formulation. Figure~\ref{fig:relative-gap-comparison}(a) shows that the largest differences occur at $x=1$, when the stopping boundaries are relatively close and local routing details matter most. The curves drop sharply once the confidence requirement becomes more stringent and, at larger $x$, cluster at only a few percentage points across a broad range of waiting-cost weights. The decline is not perfectly monotone for every $C_w$, which is unsurprising at finite $\alpha$: changing $C_w$ can change the selected specialists, and the reported values also reflect numerical and Monte Carlo approximation. The overall pattern is nevertheless consistent with the high-confidence result in Theorem~2: for fixed model characteristics and waiting-cost parameters, the cost of restricting to two specialists becomes increasingly inexpensive as $\alpha$ decreases.

\begin{figure}[h]
    \centering
    \subfloat[$\rho=1$ with varying $C_w$.\label{fig:gap-varying-Cw}]{
        \includegraphics[width=0.48\textwidth]{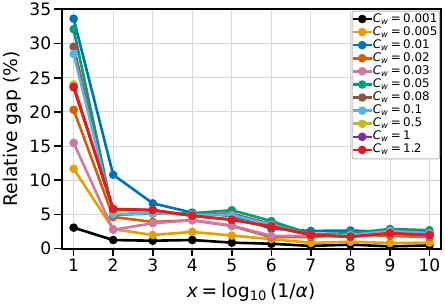}
    }
    \hfill
    \subfloat[Varying $\rho$ with $C_w=0.01$.\label{fig:gap-varying-rho}]{
        \includegraphics[width=0.48\textwidth]{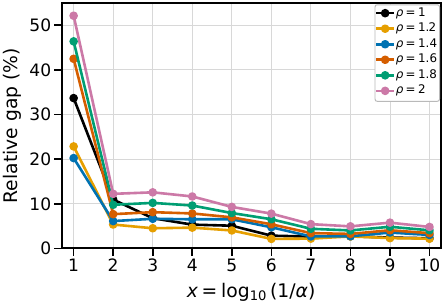}
    }
    \vspace{8pt}
    \caption{Relative gap between the pathwise Monte Carlo estimate of Policy~\algo{} and the corresponding Bellman benchmark.}
    \label{fig:relative-gap-comparison}
\end{figure}

We next fix $C_w=0.01$ and vary the curvature parameter $\rho$. Figure~\ref{fig:relative-gap-comparison}(b) shows that the reported differences are generally larger when $\rho>1$ than when $\rho=1$. This comparison should be interpreted with care. Under Monte Carlo evaluation, Policy~\algo{} pays the convex penalty on the realized total waiting time, so paths with higher waiting time receive disproportionate weight. The deterministic-latency Bellman benchmark removes precisely these fluctuations before applying the penalty. Therefore, part of the larger reported gap for $\rho>1$ reflects this fluctuation mechanism, rather than the limitation of using only two LLMs; see the discussion following Theorem~\ref{thm:two-llm-tight}. As a result, the curves are not strictly ordered monotonically in $\rho$. Despite these computational effects, the reported gap declines substantially as the confidence requirement becomes more stringent for every tested value of $\rho$.

Taken together, the two panels of Figure~\ref{fig:relative-gap-comparison} provide complementary finite-$\alpha$ evidence about the proposed policy. In both cases, the diminishing gaps at higher confidence levels are qualitatively consistent with the asymptotic prediction that Policy~\algo{} becomes increasingly effective as $\alpha\to0$.

\section{Conclusion}\label{sec:conclusion}
This paper studies a sequential decision problem motivated by modern AI systems, where a decision-maker must orchestrate multiple heterogeneous information sources under cost, waiting time, and accuracy trade-offs. We formulate a Bayesian sequential hypothesis testing model that captures key operational features of LLM deployment, including asymmetric diagnostic power, stochastic waiting times, and convex delay penalties.

Our main result is a sharp asymptotic characterization of optimal performance in the high-confidence regime. We establish a universal lower bound through a deterministic convex program and show that an asymptotically optimal policy needs at most two information sources. Motivated by this structure, we propose a simple sign-based two-LLM policy and prove that it achieves the lower bound up to a tight remainder.

These results yield a simple design insight: near-optimal sequential information acquisition does not require coordinating many models, but rather identifying two complementary specialists and switching between them as evidence evolves. The numerical study further illustrates the finite-$\alpha$ performance
of this two-LLM policy relative to the Bellman benchmarks. Taken together, our findings show that effective AI decision systems can remain structurally simple even as the number of available models grows.

\bibliographystyle{informs2014}
\bibliography{reference}

\newpage
\renewcommand{\thepage}{ec\arabic{page}}
\setcounter{page}{1}
\appendix
\newgeometry{left=0.65in,right=0.65in,top=1in,bottom=1in}
\small
\begin{center}
    \textbf{\Large Appendix}
\end{center}

\setcounter{equation}{0}
\setcounter{theorem}{0}
\setcounter{proposition}{0}
\setcounter{lemma}{0}
\setcounter{corollary}{0}
\setcounter{definition}{0}
\renewcommand{\thetheorem}{\thesection.\arabic{theorem}}
\renewcommand{\theproposition}{\thesection.\arabic{proposition}}
\renewcommand{\thelemma}{\thesection.\arabic{lemma}}
\renewcommand{\thecorollary}{\thesection.\arabic{corollary}}
\renewcommand{\thedefinition}{\thesection.\arabic{definition}}
\renewcommand{\theequation}{\thesection.\arabic{equation}}

\renewcommand{\theHtheorem}{\thesection.\arabic{theorem}}
\renewcommand{\theHproposition}{\thesection.\arabic{proposition}}
\renewcommand{\theHlemma}{\thesection.\arabic{lemma}}
\renewcommand{\theHcorollary}{\thesection.\arabic{corollary}}
\renewcommand{\theHdefinition}{\thesection.\arabic{definition}}
\renewcommand{\theHequation}{\thesection.\arabic{equation}}

The appendix provides omitted proofs in the main text, and is organized by sections.

\section{Omitted Proofs in Section~\ref{sec:model}}
\subsection{Proof of Lemma~\ref{lemma1}}

We use the likelihood of the \emph{joint sequential history}, rather than conditioning on the completed sequence of model choices. The latter sequence is adaptive and can reveal earlier observations. Conditional on each pre-query history, the same selection rule is used under both hypotheses. The selected model's waiting-time law is also the same under both hypotheses and is independent of the fresh output. Thus, the selection, randomization, and waiting-time factors cancel from the likelihood ratio. Iterating the remaining output factors gives
\[
\frac{d\PP_A|_{\mathcal G_t}}{d\PP_B|_{\mathcal G_t}}
=\prod_{s=1}^{t}\frac{P_{j_s,A}(Y_s)}{P_{j_s,B}(Y_s)}
=\exp\!\left(\sum_{s=1}^{t}\ell_{j_s}(Y_s)\right)=e^{L_t}.
\]
All output probabilities are strictly positive, so these finite-history measures are mutually absolutely continuous. The independent seed has the same law under the two hypotheses and contributes a likelihood ratio of one, including at $t=0$.

Bayes' rule now yields
\[
\PP(\theta=A\mid\mathcal G_t)
=\frac{\xi_A e^{L_t}}{\xi_A e^{L_t}+\xi_B}
=\frac{e^{\delta+L_t}}{1+e^{\delta+L_t}},
\qquad
\PP(\theta=B\mid\mathcal G_t)=\frac{1}{1+e^{\delta+L_t}}.
\]
The right-hand sides are $\mathcal H_t$-measurable. Taking conditional expectations given $\mathcal H_t\subseteq\mathcal G_t$ therefore gives the same posterior formulas on $\mathcal H_t$, proving~\eqref{pos_prob}. This also explains why past response times may affect routing without supplying additional evidence about the ground truth.

Finally, with $H_\alpha:=\log((1-\alpha)/\alpha)$,
\begin{align*}
\frac{e^{\delta+L_t}}{1+e^{\delta+L_t}}\ge1-\alpha
&\quad\Longleftrightarrow\quad L_t\ge H_\alpha-\delta=A_\alpha,\\
\frac{1}{1+e^{\delta+L_t}}\ge1-\alpha
&\quad\Longleftrightarrow\quad L_t\le-H_\alpha-\delta=-B_\alpha.
\end{align*}
Applying these equivalences at the first crossing proves the stopping and decision rules. They also cover $t=0$ when the prior already meets the confidence requirement.\hfill\Halmos

\vspace{0.1in}

\section{Omitted Proofs in Section~\ref{sec:finite-alpha-all-llms}}
\label{sec:app-finite-alpha-all-llms}

Define
\[
\bar\alpha:=\frac{1}{1+e^4},
\qquad
H_\alpha:=\log\frac{1-\alpha}{\alpha},
\qquad
M_\alpha:=\bigl\lfloor \sqrt{H_\alpha}\bigr\rfloor,
\qquad
\lambda_\alpha:=\frac{H_\alpha}{M_\alpha}.
\]
In the following, we fix $\alpha\in (0, \bar{\alpha}]$.
Because \(H_\alpha\ge 4\), we have \(M_\alpha\ge 2\) and \(\lambda_\alpha\ge \sqrt{H_\alpha}\ge 2\).
We take a symmetric prior \(\xi_A=\xi_B=1/2\), so that \(\delta=0\) and \(A_\alpha=B_\alpha=H_\alpha\). We also set \(c_j=1\) and \(T_{j,t}\equiv 1\) for all $j\in [M_\alpha]$, and \(g(w)=w\). Hence, for every admissible policy \(\pi\in\Pi(\alpha)\),
\[
\mathcal{R}(\pi;\alpha)=\EE[C_\pi]+\EE[W_\pi]=2\,\EE[\tau].
\]
Thus minimizing \(\mathcal{R}(\pi;\alpha)\) is equivalent to minimizing \(\EE[\tau]\).

For each \(j\in[M_\alpha]\), define
\[
\gamma_{j,A}
:=
\frac{1-e^{-(M_\alpha+j-1)\lambda_\alpha}}
     {1-e^{-(M_\alpha+j)\lambda_\alpha}},
\qquad
\gamma_{j,B}
:=
\frac{1-e^{-\lambda_\alpha}}
     {1-e^{-(M_\alpha+j)\lambda_\alpha}}.
\]
We also define the chain states
\[
x_i:=i\lambda_\alpha,
\qquad i=0,1,\dots,M_\alpha.
\]
Then \(x_0=0\) and \(x_{M_\alpha}=H_\alpha=A_\alpha\).

For a continuation state \(L\in(-H_\alpha,H_\alpha)\), define
\[
p(L):=\frac{e^L}{1+e^L},
\]
the posterior probability of \(A\) under the symmetric prior. Let \(V(L)\) be the minimum expected remaining number of queries when the current LLR state is \(L\), and set \(V(L):=0\) whenever \(|L|\ge H_\alpha\). Since each query has unit cost and \(g(w)=w\), the Bellman equation is
\[
V(L)=\min_{j\in[M_\alpha]}Q(L,j),
\qquad |L|<H_\alpha,
\]
where
\[
Q(L,j)
:=
1
+
q_j(L)V\bigl(L+\ell_j(A)\bigr)
+
\bigl(1-q_j(L)\bigr)V\bigl(L+\ell_j(B)\bigr),
\]
and
\[
q_j(L)
:=
p(L)\gamma_{j,A}+(1-p(L))(1-\gamma_{j,B}).
\]

\subsection{Supporting Results}
\begin{lemma}
\label{lem:finite-alpha-all-llms-valid}
For every \(j\in[M_\alpha]\), the primitive parameters satisfy
\[
\gamma_{j,A},\gamma_{j,B}\in(1/2,1),
\]
and the induced LLR increments satisfy
\[
\ell_j(A)=\lambda_\alpha,
\qquad
\ell_j(B)=-(M_\alpha+j-1)\lambda_\alpha.
\]
\end{lemma}

\begin{proof}{Proof of Lemma \ref{lem:finite-alpha-all-llms-valid}}
Fix \(j\in[M_\alpha]\) and write
\[
D_j:=(M_\alpha+j-1)\lambda_\alpha,
\qquad
U_j:=\lambda_\alpha.
\]
Because \(\lambda_\alpha\ge 2\), we have \(U_j\ge 2\) and \(D_j\ge 2\). Since \(0<1-e^{-(U_j+D_j)}<1\),
\[
\gamma_{j,A}
=
\frac{1-e^{-D_j}}{1-e^{-(U_j+D_j)}}
>
1-e^{-D_j}
\ge
1-e^{-2}
>
\frac12.
\]
Similarly,
\[
\gamma_{j,B}
=
\frac{1-e^{-U_j}}{1-e^{-(U_j+D_j)}}
>
1-e^{-U_j}
\ge
1-e^{-2}
>
\frac12.
\]
Also, \(D_j<U_j+D_j\) and \(U_j<U_j+D_j\), so
\[
1-e^{-D_j}<1-e^{-(U_j+D_j)},
\qquad
1-e^{-U_j}<1-e^{-(U_j+D_j)},
\]
which implies \(\gamma_{j,A}<1\) and \(\gamma_{j,B}<1\).

Next, we have
\begin{align*}
\frac{\gamma_{j,A}}{1-\gamma_{j,B}}
=
\frac{\dfrac{1-e^{-D_j}}{1-e^{-(U_j+D_j)}}}
     {\dfrac{e^{-U_j}(1-e^{-D_j})}{1-e^{-(U_j+D_j)}}}
=
e^{U_j}, \qquad \frac{1-\gamma_{j,A}}{\gamma_{j,B}}
=
\frac{\dfrac{e^{-D_j}(1-e^{-U_j})}{1-e^{-(U_j+D_j)}}}
     {\dfrac{1-e^{-U_j}}{1-e^{-(U_j+D_j)}}}
=
e^{-D_j},
\end{align*}
which imply $\ell_j(A)=U_j=\lambda_\alpha$ and $\ell_j(B)=-D_j=-(M_\alpha+j-1)\lambda_\alpha$.  \hfill \Halmos
\end{proof}

\vspace{0.1in}

\vspace{0.1in}
\begin{lemma}
\label{lem:finite-alpha-all-llms-q}
For every \(L\in(-H_\alpha,H_\alpha)\) and every \(j\in[M_\alpha]\), we have $1-\gamma_{j,B}=e^{-\lambda_\alpha}\gamma_{j,A}$
and hence
\[
q_j(L)
=
\bigl(p(L)+e^{-\lambda_\alpha}(1-p(L))\bigr)\gamma_{j,A}.
\]
In particular, for each fixed \(L\), the map \(j\mapsto q_j(L)\) is strictly increasing.
\end{lemma}

\begin{proof}{Proof of Lemma \ref{lem:finite-alpha-all-llms-q}}
The identity $1-\gamma_{j,B}=e^{-\lambda_\alpha}\gamma_{j,A}$ follows directly from the definitions. Substituting into the formula for \(q_j(L)\) gives
\[
q_j(L)
=
p(L)\gamma_{j,A}+(1-p(L))e^{-\lambda_\alpha}\gamma_{j,A}
=
\bigl(p(L)+e^{-\lambda_\alpha}(1-p(L))\bigr)\gamma_{j,A}.
\]
Thus it remains to show that \(\gamma_{j,A}\) is strictly increasing in \(j\). Define
\[
h(t):=\frac{1-e^{-t}}{1-e^{-(t+\lambda_\alpha)}},
\qquad t>0.
\]
Then \(\gamma_{j,A}=h((M_\alpha+j-1)\lambda_\alpha)\). Let \(a:=e^{-\lambda_\alpha}\in(0,1)\) and \(u:=e^{-t}\in(0,1)\). Then, we have
\[
h(t)=\frac{1-u}{1-au}, \qquad \text{and} \qquad \frac{d}{du}\left(\frac{1-u}{1-au}\right)
=
-\frac{1-a}{(1-au)^2}<0.
\]
Since \(u=e^{-t}\) is strictly decreasing in \(t\), \(h(t)\) is strictly increasing in \(t\). Therefore \(\gamma_{j,A}\), and hence \(q_j(L)\), is strictly increasing in \(j\).\hfill \Halmos
\end{proof}

\vspace{0.1in}
\begin{lemma}
\label{lem:finite-alpha-all-llms-geometry}
Fix \(i\in\{0,1,\dots,M_\alpha-1\}\) and \(j\in[M_\alpha]\).
\begin{enumerate}[label=(\roman*)]
    \item The good successor state from \(x_i\) under action \(j\) is
    \[
    x_i+\ell_j(A)=x_{i+1}.
    \]
    \item If \(j\ge i+1\), then the bad successor state exits the continuation region:
    \[
    x_i+\ell_j(B)\le -H_\alpha.
    \]
    \item If \(j\le i\), then the bad successor state remains in the continuation region and is negative:
    \[
    x_i+\ell_j(B)\in(-H_\alpha,0).
    \]
\end{enumerate}
\end{lemma}

\begin{proof}{Proof of Lemma \ref{lem:finite-alpha-all-llms-geometry}}
Part (i) follows from \(\ell_j(A)=\lambda_\alpha\). If \(j\ge i+1\), then
\[
x_i+\ell_j(B)
=
i\lambda_\alpha-(M_\alpha+j-1)\lambda_\alpha
=
-(M_\alpha+j-1-i)\lambda_\alpha
\le -M_\alpha\lambda_\alpha
=
-H_\alpha,
\]
which proves (ii). If \(j\le i\), then
\[
x_i+\ell_j(B)=-(M_\alpha+j-1-i)\lambda_\alpha<0.
\]
Moreover, because \(j\ge 1\) and \(i\le M_\alpha-1\),
\[
M_\alpha+j-1-i\le M_\alpha-1,
\]
so
\[
x_i+\ell_j(B)=-(M_\alpha+j-1-i)\lambda_\alpha\ge -(M_\alpha-1)\lambda_\alpha>-M_\alpha\lambda_\alpha=-H_\alpha.
\]
This proves (iii). \hfill \Halmos
\end{proof}

\vspace{0.1in}
\begin{lemma}
\label{lem:finite-alpha-all-llms-upper}
For each \(i\in\{0,1,\dots,M_\alpha\}\), we have $V(x_i)\le M_\alpha-i$.

\end{lemma}

\begin{proof}{Proof of Lemma \ref{lem:finite-alpha-all-llms-upper}}
The claim is immediate for \(i=M_\alpha\) because \(x_{M_\alpha}=H_\alpha\) is a stopping state and \(V(x_{M_\alpha})=0\). Now fix \(i\in\{0,1,\dots,M_\alpha-1\}\). Consider the policy that queries LLM \(i+1\) first, then LLM \(i+2\) if the first outcome is \(A\), then LLM \(i+3\) if the first two outcomes are both \(A\), and so on.

By Lemma~\ref{lem:finite-alpha-all-llms-geometry}(i), each \(A\) outcome moves the state from \(x_r\) to \(x_{r+1}\). By Lemma~\ref{lem:finite-alpha-all-llms-geometry}(ii), a \(B\) outcome from state \(x_r\) under action \(r+1\) exits immediately. Therefore this policy uses at most one query at each of the states \(x_i,x_{i+1},\dots,x_{M_\alpha-1}\), so it stops after at most \(M_\alpha-i\) queries pathwise. Hence, we have \(V(x_i)\le M_\alpha-i\). \hfill \Halmos
\end{proof}

\vspace{0.1in}
\begin{proposition}
\label{prop:finite-alpha-all-llms-chain-opt}
If \(\alpha\in(0,\bar\alpha]\), then for every \(i\in\{0,1,\dots,M_\alpha-1\}\), the unique Bellman-optimal action at state \(x_i\) is LLM \(i+1\), i.e., $Q(x_i,i+1)<Q(x_i,j)$ for every $j\neq i+1$.
\end{proposition}

\begin{proof}{Proof of Proposition~\ref{prop:finite-alpha-all-llms-chain-opt}}
Fix \(i\in\{0,1,\dots,M_\alpha-1\}\). We consider two cases.

\medskip
\noindent
\textbf{Case 1: comparison with later actions \(j>i+1\).}

Fix such a \(j\). By Lemma~\ref{lem:finite-alpha-all-llms-geometry}(i) and (ii), both actions \(i+1\) and \(j\) have the same good successor \(x_{i+1}\) and both make the bad successor stop immediately. Hence, we can derive
\[
Q(x_i,i+1)=1+q_{i+1}(x_i)V(x_{i+1}),
\qquad
Q(x_i,j)=1+q_j(x_i)V(x_{i+1}).
\]
Because \(j\ge i+2\), necessarily \(i\le M_\alpha-2\), so \(x_{i+1}<H_\alpha\) and therefore \(V(x_{i+1})\ge 1\). By Lemma~\ref{lem:finite-alpha-all-llms-q}, \(q_j(x_i)>q_{i+1}(x_i)\), and thus
\[
Q(x_i,j)-Q(x_i,i+1)
=
\bigl(q_j(x_i)-q_{i+1}(x_i)\bigr)V(x_{i+1})
>0.
\]

\medskip
\noindent
\textbf{Case 2: comparison with earlier actions \(j\le i\).}

Fix such a \(j\). By Lemma~\ref{lem:finite-alpha-all-llms-geometry}(i) and (iii), we have $x_i+\ell_j(A)=x_{i+1}$ and $x_i+\ell_j(B)\in(-H_\alpha,0)$, which implies
\[
Q(x_i,j)
=
1+q_j(x_i)V(x_{i+1})
+
\bigl(1-q_j(x_i)\bigr)V\bigl(x_i+\ell_j(B)\bigr), \qquad Q(x_i,i+1)=1+q_{i+1}(x_i)V(x_{i+1}).
\]
Subtracting the above equations yields
\[
Q(x_i,j)-Q(x_i,i+1)
=
\bigl(q_j(x_i)-q_{i+1}(x_i)\bigr)V(x_{i+1})
+
\bigl(1-q_j(x_i)\bigr)V\bigl(x_i+\ell_j(B)\bigr).
\]
Since \(x_i+\ell_j(B)\) is a continuation state, at least one further query is needed before stopping, i.e., $V\bigl(x_i+\ell_j(B)\bigr)\ge 1$.
Hence, we have
\begin{equation}
\label{eq:finite-alpha-all-llms-basic-lower}
Q(x_i,j)-Q(x_i,i+1)
\ge
-\bigl(q_{i+1}(x_i)-q_j(x_i)\bigr)V(x_{i+1})
+
\bigl(1-q_j(x_i)\bigr).
\end{equation}
By Lemma~\ref{lem:finite-alpha-all-llms-upper}, we have $V(x_{i+1})\le M_\alpha-i-1\le M_\alpha-1$.

We now bound the two terms on the right-hand side.

First, by Lemma~\ref{lem:finite-alpha-all-llms-q}, we have
\[
q_{i+1}(x_i)-q_j(x_i)
=
\bigl(p(x_i)+e^{-\lambda_\alpha}(1-p(x_i))\bigr)
\bigl(\gamma_{i+1,A}-\gamma_{j,A}\bigr)
\le
\gamma_{i+1,A}-\gamma_{j,A}.
\]
Because \(\gamma_{j,A}\le\gamma_{i+1,A}\le 1\), we have
\[
\gamma_{i+1,A}-\gamma_{j,A} \le 1-\gamma_{j,A}
=
\frac{e^{-(M_\alpha+j-1)\lambda_\alpha}(1-e^{-\lambda_\alpha})}
     {1-e^{-(M_\alpha+j)\lambda_\alpha}}
\le
e^{-(M_\alpha+j-1)\lambda_\alpha}
\le
e^{-M_\alpha\lambda_\alpha}
=
e^{-H_\alpha}.
\]
which implies that $q_{i+1}(x_i)-q_j(x_i)\le e^{-H_\alpha}$.

Second, again by Lemma~\ref{lem:finite-alpha-all-llms-q}, we have
\[
1-q_j(x_i)>(1-p(x_i))(1-e^{-\lambda_\alpha}).
\]
Since \(x_i=i\lambda_\alpha\le (M_\alpha-1)\lambda_\alpha\), we can derive
\[
1-p(x_i)=\frac{1}{1+e^{x_i}}\ge \frac{1}{1+e^{(M_\alpha-1)\lambda_\alpha}}.
\]
Because \(\lambda_\alpha\ge 1\) and \(1+e^{(M_\alpha-1)\lambda_\alpha}\le 2e^{(M_\alpha-1)\lambda_\alpha}\), it follows that
\[
1-q_j(x_i)
>
\frac{1-e^{-\lambda_\alpha}}{1+e^{(M_\alpha-1)\lambda_\alpha}}
\ge
\frac{1-e^{-1}}{2}e^{-(M_\alpha-1)\lambda_\alpha}.
\]
Substituting these estimates into \eqref{eq:finite-alpha-all-llms-basic-lower} gives
\[
Q(x_i,j)-Q(x_i,i+1)
>
-(M_\alpha-1)e^{-H_\alpha}
+
\frac{1-e^{-1}}{2}e^{-(M_\alpha-1)\lambda_\alpha}.
\]
Factoring out \(e^{-(M_\alpha-1)\lambda_\alpha}\) and using \(H_\alpha=M_\alpha\lambda_\alpha\), we obtain
\[
Q(x_i,j)-Q(x_i,i+1)
>
e^{-(M_\alpha-1)\lambda_\alpha}
\left(
\frac{1-e^{-1}}{2}
-
(M_\alpha-1)e^{-\lambda_\alpha}
\right).
\]
The quantity in parentheses is strictly positive when $\alpha\in (0, \bar{\alpha}]$, so \(Q(x_i,j)>Q(x_i,i+1)\).

Combining Cases 1 and 2 shows that action \(i+1\) is the unique Bellman-optimal action at state \(x_i\). \hfill \Halmos
\end{proof}
\vspace{0.1in}

\subsection{Proof of Theorem~\ref{thm:finite-alpha-all-llms}}
By Proposition~\ref{prop:finite-alpha-all-llms-chain-opt}, for each \(i=0,1,\dots,M_\alpha-1\), the unique Bellman-optimal action at state \(x_i\) is LLM \(i+1\).

Now consider the event
\[
E:=\{Y_1=A,\,Y_2=A,\,\dots,\,Y_{M_\alpha}=A\}.
\]
On \(E\), the LLR evolves as
\[
x_0\to x_1\to x_2\to\cdots\to x_{M_\alpha}=H_\alpha=A_\alpha,
\]
so every exact optimal policy must query LLMs \(1,2,\dots,M_\alpha\) in order before stopping. Therefore, we have
\[
E\subseteq \{N_j(\tau)\ge 1 \text{ for all } j\in[M_\alpha]\},
\]
which implies
\[
\PP_A\!\left(N_j(\tau)\ge 1 \text{ for all } j\in[M_\alpha]\right)\ge \PP_A(E).
\]
Under \(\PP_A\), whenever LLM \(j\) is queried, the probability of outcome \(A\) is \(\gamma_{j,A}\), so
\[
\PP_A(E)=\prod_{j=1}^{M_\alpha}\gamma_{j,A}.
\]
Moreover, we can derive
\[
1-\gamma_{j,A}
=
\frac{e^{-(M_\alpha+j-1)\lambda_\alpha}(1-e^{-\lambda_\alpha})}
     {1-e^{-(M_\alpha+j)\lambda_\alpha}}
\le
e^{-(M_\alpha+j-1)\lambda_\alpha}
\le
e^{-M_\alpha\lambda_\alpha}
=
e^{-H_\alpha}
=
\frac{\alpha}{1-\alpha}.
\]
Using the union bound, we have
\[
1-\prod_{j=1}^{M_\alpha}\gamma_{j,A}
\le
\sum_{j=1}^{M_\alpha}(1-\gamma_{j,A})
\le
M_\alpha\frac{\alpha}{1-\alpha}.
\]
Therefore, we can derive
\[
\PP_A\!\left(N_j(\tau)\ge 1 \text{ for all } j\in[M_\alpha]\right)
\ge \PP_A(E) \ge
1-M_\alpha\frac{\alpha}{1-\alpha}.
\]
Finally, we have \(\alpha M_\alpha/(1-\alpha)\to 0 \) as \(\alpha\rightarrow 0\) because \(M_\alpha=\lfloor \sqrt{\log((1-\alpha)/\alpha)}\rfloor\) and $\sqrt{\log(1/\alpha)} \cdot \alpha\to 0$.\hfill \Halmos

\vspace{0.1in}

\section{Omitted Proofs in Section~\ref{sec:2-LLM}}\label{sec:concentration}

\subsection{Supporting Results}
Write
\[
h_\alpha:=\log(1/\alpha),\qquad H_\alpha:=\log((1-\alpha)/\alpha),
\qquad I_{\min,\theta}:=\min_{j\in[M]}I_{j,\theta}>0.
\]
We work in the nontrivial regime $\alpha<\min\{\xi_A,\xi_B\}$, so both stopping boundaries are positive. Define
\begin{align}
C_\ell&:=\max_{j\in[M]}\max\{|\ell_j(A)|,|\ell_j(B)|\}<\infty,\label{C_l}\\
t_\alpha&:= \frac{\min\{A_\alpha, B_\alpha\}}{C_\ell}<\infty, \\
v_\ell^2&:=\max_{j\in[M]}\max\{\Var_A(\ell_j(Y)),\Var_B(\ell_j(Y))\}<\infty.\label{v_l}
\end{align}
Finiteness follows from $M<\infty$ and strictly interior accuracies. When a deterministic-time process is used after the policy stops, extend its LLM-selection rule nonanticipatively beyond $\tau$; this extension has no effect on stopped quantities.

\subsubsection{Likelihood Ratio Martingale and Change of Measure}\

Let $\Lambda_t:=e^{L_t}$, with $\Lambda_0=1$. In this subsection, we provide some preliminary properties of $\Lambda_t$.

\begin{lemma}[Likelihood ratio martingale]
\label{lem:LR-martingale}
Under $\PP_B$, $(\Lambda_t)_{t\ge0}$ is a nonnegative $(\mathcal G_t)$-martingale with
\[
\EE_B[\Lambda_t\mid\mathcal G_{t-1}]=\Lambda_{t-1}\quad(t\ge1),
\qquad\EE_B[\Lambda_t]=1\quad(t\ge0).
\]
Under $\PP_A$, $(\Lambda_t^{-1})_{t\ge0}$ is a nonnegative martingale with mean one.
\end{lemma}

\begin{proof}{Proof of Lemma~\ref{lem:LR-martingale}.}
Fix $t\ge1$ and work under $\PP_B$. Conditional on $\theta=B$ and $\mathcal{G}_{t-1}$, we have
\begin{align*}
\EE_B[\Lambda_t \mid \mathcal{G}_{t-1}]
&= \EE_B\big[e^{L_{t-1} + \ell_{j_t}(Y_t)} \,\big|\, \mathcal{G}_{t-1}\big] = e^{L_{t-1}} \sum_{y\in\{A,B\}} \PP_B(Y_t=y\mid \mathcal{G}_{t-1})\, e^{\ell_{j_t}(y)} \\
&= e^{L_{t-1}} \sum_{y\in\{A,B\}} \PP(Y=y\mid \theta=A,j_t) = e^{L_{t-1}} = \Lambda_{t-1},
\end{align*}
where the third equality is because $e^{\ell_{j_t}(y)} = \PP(Y=y\mid\theta=A,j_t)/\PP(Y=y\mid\theta=B,j_t)$.
Thus $(\Lambda_t)_{t\ge0}$ is a martingale under $\PP_B$ with $\EE_B[\Lambda_0]=1$.
By the optional stopping theorem (see Theorem~4.8.6 in \citealt{durrett2019probability}), we have $\EE_B[\Lambda_t]=1$ for all $t$.
The argument under $\PP_A$ is analogous: for $Z_t := \Lambda_t^{-1}=e^{-L_t}$ we have
\[
\EE_A[Z_t\mid\mathcal{G}_{t-1}] = Z_{t-1},\qquad \EE_A[Z_t]=1. \tag*{\Halmos}
\]
\end{proof}

\vspace{0.1in}

\begin{lemma}[Change-of-measure identity at stopping times]
\label{lem:change-of-measure}
Let $\tau$ be any almost surely finite stopping time with respect to $(\mathcal{G}_t)_{t\ge0}$. Then for every bounded $\mathcal{G}_\tau$-measurable random variable $X$,
\[
\EE_A[X] = \EE_B[\Lambda_\tau X],\qquad
\EE_B[X] = \EE_A[\Lambda_\tau^{-1} X].
\]
Furthermore, for every event $F\in\mathcal{G}_\tau$,
\[
\PP_A(F) = \EE_A[\mathbf{1}_F] = \EE_B[\Lambda_\tau \mathbf{1}_F],\qquad
\PP_B(F) = \EE_B[\mathbf{1}_F] = \EE_A[\Lambda_\tau^{-1} \mathbf{1}_F].
\]
\end{lemma}

\begin{proof}{Proof of Lemma~\ref{lem:change-of-measure}.}
For every deterministic $t$, $\Lambda_t=e^{L_t}$ is the likelihood ratio
of $\PP_A$ with respect to $\PP_B$ on $\mathcal G_t$. Hence, for every
bounded $\mathcal G_t$-measurable random variable $H$, $\EE_A[H]=\EE_B[\Lambda_t H]$.
Now let $X$ be bounded and $\mathcal G_\tau$-measurable. Since $X\mathbf 1_{\{\tau=t\}}$ is $\mathcal G_t$-measurable, the deterministic-time identity gives
\[
\EE_A[X\mathbf 1_{\{\tau=t\}}]
=
\EE_B[\Lambda_t X\mathbf 1_{\{\tau=t\}}].
\]
Summing over $t\ge0$ and using the almost-sure finiteness of $\tau$,
\[
\begin{aligned}
\EE_A[X]
&=\sum_{t=0}^\infty
  \EE_A[X\mathbf 1_{\{\tau=t\}}] =\sum_{t=0}^\infty
  \EE_B[\Lambda_t X\mathbf 1_{\{\tau=t\}}] =\EE_B[\Lambda_\tau X].
\end{aligned}
\]
The interchange is justified since $X$ is bounded; equivalently, one may
apply the argument separately to $X^+$ and $X^-$. The second identity is
proved symmetrically using $\Lambda_t^{-1}$ under $\PP_A$. Taking
$X=\mathbf 1_F$ gives the event identities. \hfill \Halmos
\end{proof}

\vspace{0.1in}

\subsubsection{Error Probability Bounds}\

In this subsection, we show that the probability of concluding a wrong answer is upper bounded.
\begin{lemma}[Exponential error bounds]
\label{lem:error-bounds}
For the posterior-threshold rule,
\[
\PP_A(D=B)\le e^{-B_\alpha}\le c_{\mathrm{err}}\alpha,
\qquad
\PP_B(D=A)\le e^{-A_\alpha}\le c_{\mathrm{err}}\alpha,
\]
where $c_{\mathrm{err}}:=2\max\{\xi_A/\xi_B,\xi_B/\xi_A\}$ is independent of $\alpha$ and the policy.
\end{lemma}

\begin{proof}{Proof of Lemma~\ref{lem:error-bounds}.}
We prove the first inequality; the second is analogous.
Under $\theta=A$, the error event $\{D=B\}$ occurs if and only if $L_\tau\le-B_\alpha$. By the change-of-measure identity (Lemma~\ref{lem:change-of-measure}),
\[
\PP_A(D=B)
= \PP_A(L_\tau\le-B_\alpha)
= \EE_B[\Lambda_\tau \mathbf{1}_{\{L_\tau\le-B_\alpha\}}].
\]
On the event $\{L_\tau\le-B_\alpha\}$, we have $\Lambda_\tau = e^{L_\tau} \le e^{-B_\alpha}$, so
\[
\PP_A(D=B)
\le e^{-B_\alpha} \EE_B[\mathbf{1}_{\{L_\tau\le-B_\alpha\}}]
\le e^{-B_\alpha} \EE_B[1]
= e^{-B_\alpha}= \frac{\alpha}{1-\alpha}\,\frac{\xi_B}{\xi_A}
\le c_{\mathrm{err}}\cdot \alpha,
\]
where $c_{\mathrm{err}}=2\cdot \max\{\xi_A/\xi_B, \xi_B/\xi_A\}$ is a constant independent of $\alpha$. Similarly, we can deduce that $\PP_B(D=A)\le e^{-A_\alpha}\le c_{\mathrm{err}} \cdot \alpha$.\hfill \Halmos
\end{proof}

\subsubsection{Martingale Decomposition of Log-Likelihood Ratio}\

In this subsection, we show that $M_t^{(\theta)}$ in \eqref{decomp} is a martingale, which will be used to establish concentration results. Recall that, by \eqref{eq:def-nj}, $N_j(t)$ denotes the total number of queries sent to LLM $j$ until period $t$, $N_j(t) := \sum_{\ell=1}^t \mathbf{1}_{\{j_t=j\}}$.

\begin{lemma}[Drift-martingale decomposition]
\label{lem:drift-decomp}
For each $t\ge0$,
\[
L_t = \sum_{s=1}^t I_{j_s,A} + M_t^{(A)}
= \sum_{j=1}^M I_{j,A} N_j(t) + M_t^{(A)},
\]
where $(M_t^{(A)})_{t\ge0}$ is a $(\mathcal{G}_t)$-martingale under $\PP_A$ with $\EE_A[M_t^{(A)}]=0$ and bounded increments $|M_t^{(A)} - M_{t-1}^{(A)}|\le 2C_\ell$. Similarly,
\[
-L_t = \sum_{s=1}^t I_{j_s,B} + M_t^{(B)}
= \sum_{j=1}^M I_{j,B} N_j(t) + M_t^{(B)},
\]
where $(M_t^{(B)})_{t\ge0}$ is a $(\mathcal{G}_t)$-martingale under $\PP_B$ with $\EE_B[M_t^{(B)}]=0$ and bounded increments $|M_t^{(B)} - M_{t-1}^{(B)}|\le 2C_\ell$.
\end{lemma}
\begin{proof}{Proof of Lemma~\ref{lem:drift-decomp}.}
    We prove the claim under $\PP_A$; the argument under $\PP_B$ is identical with signs adjusted.
Define
\[
D_s^{(A)} := \ell_{j_s}(Y_s) - I_{j_s,A},\qquad
M_t^{(A)} := \sum_{s=1}^t D_s^{(A)},\qquad t\ge0,
\]
with $M_0^{(A)}:=0$. Then, we have
\[
L_t = \sum_{s=1}^t \ell_{j_s}(Y_s)
= \sum_{s=1}^t I_{j_s,A} + M_t^{(A)}.
\]

Under $\PP_A$, we have $\EE_A[\ell_{j_s}(Y_s)\mid \mathcal{G}_{s-1}] = I_{j_s,A}$, so
\[
\EE_A\big[D_s^{(A)} \mid \mathcal{G}_{s-1}\big]
= \EE_A\big[\ell_{j_s}(Y_s) - I_{j_s,A} \mid \mathcal{G}_{s-1}\big] = 0.
\]

Hence $(M_t^{(A)})_{t\ge0}$ is a $(\mathcal{G}_t)$-martingale under $\PP_A$. The martingale $(M_t^{(A)})$ has bounded increments $|M_t^{(A)} - M_{t-1}^{(A)}|$ because
\begin{align*}
|M_t^{(A)} - M_{t-1}^{(A)}|\le |L_t-L_{t-1}|+|I_{j_t, A}| \le 2C_\ell \quad \mbox{almost surely for all $t\ge 1$.}
\end{align*}
The identities for expectations follow immediately by taking $\EE_A[\cdot]$ on both sides and using $\EE_A[M_t^{(A)}]=0$.
As for the equality, it suffices to note the following equation:
\[
\sum_{s=1}^t I_{j_s,A}
= \sum_{j=1}^M I_{j,A} \sum_{s=1}^t \mathbf{1}_{\{j_s=j\}}
= \sum_{j=1}^M I_{j,A} N_j(t).
\]
 \hfill \Halmos
\end{proof}

\subsubsection{Concentration of Log-Likelihood Ratio}\

In this subsection, we present concentration results for the cumulative log-likelihood ratio $L_t$. We first state the version of Freedman's inequality that we will use.
\begin{lemma}[Freedman's inequality; Theorem~1.6 in \citealt{freedman1975tail}]
\label{lem:Freedman}
Let $(M_t)_{t\ge0}$ be a real-valued martingale with respect to a filtration $(\mathcal{F}_t)_{t\ge0}$, with $M_0=0$, and define the martingale differences $X_t := M_t - M_{t-1}$ for all $t\ge1$.
Assume that there exists a constant $b>0$ such that $|X_t| \le b$ almost surely for all $t\ge 1$.
Define the predictable quadratic variation
$V_t := \sum_{s=1}^t \EE[X_s^2 \mid \mathcal{F}_{s-1}]$ with the convention $V_0:=0$.
Then for every $t\ge1$, every $x\ge0$, and every deterministic constant $v\ge0$ such that $V_t \le v$, we have
\[
\PP(|M_t| \ge x) \;\le\; 2\,\exp\!\Big(-\frac{x^2}{2(v + b x)}\Big),
\qquad x\ge0.
\]
\end{lemma}

\vspace{0.1in}

Then, we prove the concentration inequalities.
\begin{lemma}[Freedman-type concentration for the LLR martingales]
\label{lem:Freedman-LLR}
Under Assumptions~1-3, for each $\theta\in\{A,B\}$, each integer $t\ge1$, and each $x>0$,
\[
\PP_\theta\big(|M_t^{(\theta)}|\ge x\big)
\;\le\;
2 \exp\!\Big(-\frac{x^2}{2(v_\ell^2\,t + 2C_\ell x)}\Big).
\]
Consequently, under $\PP_A$,
\[
\PP_A\big(\big|L_t - \textstyle\sum_{s=1}^t I_{j_s,A}\big|\ge x\big)
\;\le\;
2 \exp\!\Big(-\frac{x^2}{2(v_\ell^2\,t + 2C_\ell x)}\Big),
\]
and under $\PP_B$,
\[
\PP_B\big(\big|-L_t - \textstyle\sum_{s=1}^t I_{j_s,B}\big|\ge x\big)
\;\le\;
2 \exp\!\Big(-\frac{x^2}{2(v_\ell^2\,t + 2C_\ell x)}\Big).
\]
\end{lemma}

\begin{proof}{Proof of Lemma~\ref{lem:Freedman-LLR}.}
Fix $\theta\in\{A,B\}$ and $t\ge1$. From Lemma~\ref{lem:drift-decomp}, we know that $(M_s^{(\theta)})_{s\ge0}$ is a $(\mathcal{G}_s)$-martingale under $\PP_\theta$ with $M_0^{(\theta)}=0$.
Then, we apply Lemma~\ref{lem:Freedman} to the martingale $(M_s^{(\theta)})_{s=0}^t$ with respect to the filtration $(\mathcal{G}_s)_{s=0}^t$ under the probability measure $\PP_\theta$.
The martingale difference $D_s^{(\theta)}=M_s^{(\theta)} - M_{s-1}^{(\theta)}$ is defined in Lemma~\ref{lem:drift-decomp}.
From Lemma~\ref{lem:drift-decomp}, we have $|D_s^{(\theta)}| \le 2C_\ell$ almost surely for all $s\ge1.$
Thus the increment bound in Lemma~\ref{lem:Freedman} holds with $b = 2C_\ell$.

The predictable quadratic variation is, by \eqref{v_l},
\[
V_t^{(\theta)}
= \sum_{s=1}^t \EE_\theta[(D_s^{(\theta)})^2 \mid \mathcal{G}_{s-1}] \le v_\ell^2 t \quad\text{almost surely}.
\]

By Lemma~\ref{lem:Freedman} with the choices $b = 2C_\ell$ and $v = v_\ell^2 t$, for each $x\ge0$,
\[
\PP_\theta(|M_t^{(\theta)}| \ge x)
\le 2\exp\!\Big(-\frac{x^2}{2(v_\ell^2 t + 2C_\ell x)}\Big).
\]

Under $\PP_A$, the decomposition~\eqref{decomp} gives
\begin{align*}
\PP_A\big(\big|L_t - \textstyle\sum_{s=1}^t I_{j_s,A}\big|\ge x\big)
&= \PP_A\big(|M_t^{(A)}|\ge x\big) \le 2 \exp\!\Big(-\frac{x^2}{2(v_\ell^2 t + 2C_\ell x)}\Big). \tag*{\Halmos}
\end{align*}

\end{proof}

\subsubsection{Overshoot and Expected LLR at Stopping}\

In this subsection, we investigate the relationship between $L_\tau$ and $(A_\alpha, B_\alpha)$. When the stopping time $\tau$ is reached, the LLR $L_\tau$ typically overshoots the boundary by a small amount. The following lemma provides a uniform bound on this overshoot.

\begin{lemma}[Uniform overshoot bound]
\label{lem:overshoot}
For the posterior-threshold stopping rule, the overshoot is uniformly bounded:
\[
|L_\tau - A_\alpha| \le C_\ell \quad \text{conditioning on } \{L_\tau\ge A_\alpha\},
\]
and
\[
|L_\tau + B_\alpha| \le C_\ell \quad \text{conditioning on } \{L_\tau\le-B_\alpha\},
\]
where $C_\ell$ is the constant in \eqref{C_l}.
\end{lemma}

\begin{proof}{Proof of Lemma~\ref{lem:overshoot}.}
By definition of $\tau$, we have $|L_{t}| < \max\{A_\alpha,B_\alpha\}$ for all $t<\tau$, and at time $\tau$ the process $L_\tau$ first crosses one of the boundaries. On the event $\{L_\tau\ge A_\alpha\}$, we must have $L_{\tau-1} < A_\alpha$ (otherwise $\tau$ would have occurred earlier), and
\[
L_\tau = L_{\tau-1} + \ell_{j_\tau}(Y_\tau) < A_\alpha + C_\ell,
\]
by \eqref{C_l}. Since $L_\tau\ge A_\alpha$ by assumption, we have
\[
A_\alpha \le L_\tau < A_\alpha + C_\ell,
\]
which gives $|L_\tau - A_\alpha| < C_\ell$. The bound on the other event is proved identically.\hfill \Halmos
\end{proof}

\vspace{0.1in}

The next lemma gives upper and lower bounds on the expected cumulative LLR at stopping under each hypothesis.

\begin{lemma}[Expected LLR at stopping]
\label{lem:ELtau-vs-boundaries}
Suppose the policy $\pi$ satisfies $\EE_A[\tau]<\infty$ and $\EE_B[\tau]<\infty$. Then for all sufficiently small $\alpha>0$,
\begin{align}
\label{eq:EA-Ltau-bounds}
A_\alpha - c_{\mathrm{err}}\alpha\cdot (A_\alpha+B_\alpha+C_\ell) \le \EE_A[L_\tau] \le A_\alpha + C_\ell,\\[1mm]
\label{eq:EB-Ltau-bounds}
B_\alpha - c_{\mathrm{err}}\alpha\cdot (A_\alpha+B_\alpha+C_\ell) \le \EE_B[-L_\tau] \le B_\alpha + C_\ell.
\end{align}
In particular, $\EE_A[L_\tau] = A_\alpha + O(1)$ and $\EE_B[-L_\tau] = B_\alpha + O(1)$ as $\alpha\rightarrow0$.
\end{lemma}

\begin{proof}{Proof of Lemma~\ref{lem:ELtau-vs-boundaries}.}
We prove~\eqref{eq:EA-Ltau-bounds}; the proof of~\eqref{eq:EB-Ltau-bounds} is analogous.
When $\theta=A$, we can decompose the expectation by the decision outcome:
\begin{align*}
\EE_A[L_\tau]
&= \EE_A[L_\tau \ind_{\{D=A\}}] + \EE_A[L_\tau \ind_{\{D=B\}}] = \EE_A[L_\tau \ind_{\{L_\tau\ge A_\alpha\}}] + \EE_A[L_\tau \ind_{\{L_\tau\le-B_\alpha\}}].
\end{align*}

On the event $\{L_\tau\ge A_\alpha\}$, Lemma~\ref{lem:overshoot} gives $A_\alpha \le L_\tau \le A_\alpha + C_\ell$. On the event $\{L_\tau\le-B_\alpha\}$, we have $-B_\alpha-C_\ell\le L_\tau\le-B_\alpha<0$. Therefore,
\begin{align*}
\EE_A[L_\tau]
&\ge A_\alpha \,\PP_A(L_\tau\ge A_\alpha) + (-B_\alpha-C_\ell)\,\PP_A(L_\tau\le-B_\alpha) \\
&\ge A_\alpha \,\PP_A(D=A) - (B_\alpha+C_\ell)\,\PP_A(D=B)\\
&= A_\alpha - (A_\alpha+B_\alpha+C_\ell)\,\PP_A(D=B)
\end{align*}

By Lemma~\ref{lem:error-bounds}, we have $\PP_A(D=B)\le c_{\mathrm{err}}\alpha$ and hence
\[
\EE_A[L_\tau] \ge A_\alpha - c_{\mathrm{err}}\alpha\cdot (A_\alpha+B_\alpha+C_\ell),
\]
which establishes the lower bound in~\eqref{eq:EA-Ltau-bounds}.
For the upper bound, note that on $\{L_\tau\ge A_\alpha\}$ we have $L_\tau\le A_\alpha + C_\ell$ and on $\{L_\tau\le-B_\alpha\}$ we have $L_\tau\le0$. Therefore,
\begin{align*}
\EE_A[L_\tau]
&\le (A_\alpha + C_\ell)\,\PP_A(L_\tau\ge A_\alpha) + 0\cdot\PP_A(L_\tau\le-B_\alpha) \\
&\le A_\alpha + C_\ell.
\end{align*}

Combining the bounds gives~\eqref{eq:EA-Ltau-bounds}. Moreover, since $x\log(1/x)\le 1/e$, we can deduce that $c_{\mathrm{err}}\alpha(A_\alpha+B_\alpha+C_\ell)$ is upper bounded by a constant. Thus, we have $\EE_A[L_\tau] = A_\alpha + O(1)$.

The proof of~\eqref{eq:EB-Ltau-bounds} is identical with $-L_\tau$ and $B_\alpha$ in place of $L_\tau$ and $A_\alpha$.
\hfill \Halmos
\end{proof}

\vspace{0.1in}

\subsubsection{Properties of Waiting Cost Function}\label{append:wait_func}\

Define $\phi(s):=s^\rho-1-\rho(s-1)$ for $s\ge 0$ and $q:=\max\{2, \rho\}$. In the following lemma, we provide bounds for the function $\phi(s)$.
\begin{lemma}\label{lem:wait_func}
   For $\rho\ge1$, $\phi(s)\ge0$. There is a finite $c_1>0$ such that
\[
\phi(s)\le c_1\bigl((s-1)^2+|s-1|^q\bigr),\qquad s\ge0.
\]
If $\rho>1$, there is $c_2>0$ such that
\[
\phi(s)\ge c_2\min\{|s-1|,1/2\}^2,\qquad s\ge0.
\]
\end{lemma}

\begin{proof}{Proof of Lemma~\ref{lem:wait_func}.}
    We start with the first statement. Since $\phi(s)$ is convex and $\phi'(1)=0$, we have $\phi(s)$ decreases when $s\in [0, 1)$ and increases when $s\in [1, \infty)$. Thus, we can deduce that $\phi(s)\ge \phi(1)=0$ for $s\ge 0$.

    Then we prove the second statement. The inequality holds for any constant when $s=1$. For $s\in (0, \infty)\setminus \{1\}$, we define the ratio
    \[
    R(s) = \frac{\phi(s)}{(s-1)^2+|s-1|^q}.
    \]
    Then, we consider several limits as follows:
    \begin{align*}
        &\lim_{s\to 0^+}R(s) = \frac{\rho-1}{2}\\
        &\lim_{s\to 1} R(s) = \lim_{s\to 1} \frac{\rho s^{\rho-1}-\rho}{2(s-1)+q|s-1|^{q-1}} = \lim_{s\to 1} \frac{\rho(\rho-1) s^{\rho-2}}{2+q(q-1)|s-1|^{q-2}} = \begin{cases}
            \frac{\rho(\rho-1)}{2} & q>2\\
            \frac{\rho(\rho-1)}{4} & q=2
        \end{cases},\\
        &\lim_{s\to \infty} R(s) = \begin{cases}
            0 & \rho<2\\
            1/2 & \rho =2\\
            1 & \rho>2
        \end{cases},
    \end{align*}
    which implies that $R(s)<\infty$ for $s\in (0, \infty)\setminus \{1\}$. Since $R(s)$ is continuous on $(0, 1)$ and $(1, \infty)$, there exists a constant $c_1$ such that the $\phi(s)\le c_1\left((s-1)^2+|s-1|^q\right)$.

    Lastly, if $\rho>1$, we have $\phi''(s)=\rho(\rho-1)s^{\rho-2}$. By Taylor expansion, for $s\in [1/2, 3/2]$, we have
    \[
        \phi(s) = \phi(1) + \phi'(1)(s-1) + \frac12\phi''(s_0)(s-1)^2 = \frac12 \rho(\rho-1)s_0^{\rho-2} (s-1)^2,
    \]
    where $s_0\in [1/2, 3/2]$. Thus, there exists a constant $c_2=\min_{x\in [1/2, 3/2]} \rho(\rho-1)x^{\rho-2}/2$ such that $\phi(s)\ge c_2(s-1)^2$. When $|s-1|>1/2$, we have $\phi(s)\ge \min\{\phi(1/2), \phi(3/2)\}\ge c_2\cdot (1/2)^2$.\hfill\Halmos
\end{proof}
\vspace{0.1in}

\subsection{Proof of Proposition~\ref{prop:info-budget-asym}}

We proceed in two steps.

\noindent
\textbf{Step 1: Drift identities under $\theta=A$ and $\theta=B$.}

By Lemma~\ref{lem:drift-decomp}, under $\PP_A$ we can write
\[
  L_t = \sum_{s=1}^t I_{j_s,A} + M_t^{(A)},
\]
where $(M_t^{(A)})_{t\ge 0}$ is a martingale under $\PP_A$ with respect to
the natural filtration $(\mathcal{G}_t)_{t\ge 0}$, and $I_{j_s,A}=\EE_A[\ell_{j_s}(Y)]$ is the conditional expectation of the LLR increment given $j_s$.

First, we have the following conditions:
\begin{enumerate}[label=(\roman*),leftmargin=8mm]
\item $\EE_A[\tau]<\infty$,
\item The increments of $M_t^{(A)}$ are bounded: $|M_t^{(A)} - M_{t-1}^{(A)}|\le |L_t-L_{t-1}|+|I_{j_t, A}| \le 2C_\ell$ almost surely for all $t\ge 1$.
\end{enumerate}

Given the above conditions, by the optional
stopping theorem (see Theorem~4.8.6 in \citealt{durrett2019probability}),
\[
  \EE_A[M_\tau^{(A)}] = \EE_A[M_0^{(A)}] = 0.
\]

Taking expectations at the stopping time $\tau$,
\begin{align*}
  \EE_A[L_\tau]
  &= \EE_A\!\left[\sum_{s=1}^\tau I_{j_s,A} + M_\tau^{(A)}\right]= \EE_A\!\left[\sum_{s=1}^\tau I_{j_s,A}\right]
     + \EE_A[M_\tau^{(A)}]= \EE_A\!\left[\sum_{s=1}^\tau I_{j_s,A}\right].
\end{align*}
Rearranging the sum by grouping terms according to which LLM is queried,
\[
  \sum_{s=1}^\tau I_{j_s,A}
  = \sum_{j=1}^M I_{j,A}
      \sum_{s=1}^\tau \mathbf{1}_{\{j_s=j\}}
  = \sum_{j=1}^M I_{j,A} N_j(\tau),
\]
where $N_j(\tau)$ denotes the total number of queries sent to LLM \(j\) before stopping.

To interchange expectation and summation, we verify the conditions for
Fubini's theorem. Since $M$ is finite, $I_{j,A} \in (0,C_\ell]$ for all $j$, and $\EE_A[\tau] < \infty$, we have
\[
  \EE_A\!\left[\sum_{j=1}^M I_{j,A} N_j(\tau)\right]
  \le \EE_A\!\left[C_\ell \cdot \tau\right]
  = C_\ell \cdot \EE_A[\tau]
  < \infty.
\]
Therefore, by Fubini's theorem, we can interchange expectation and the finite sum:
\[
  \EE_A[L_\tau]
  = \sum_{j=1}^M I_{j,A} \EE_A[N_j(\tau)]
  = \sum_{j=1}^M I_{j,A} n_{j,A}.
\]

An analogous argument under $\PP_B$ uses the fact that $-L_t$ has drift
$I_{j_s,B}$ under $\theta=B$. Specifically, by Lemma~\ref{lem:drift-decomp},
we can write
\[
  -L_t = \sum_{s=1}^t I_{j_s,B} + M_t^{(B)},
\]
where $(M_t^{(B)})$ is a martingale under $\PP_B$ with $M_0^{(B)}=0$.
The same verification as above shows that $\EE_B[-L_\tau]= \sum_{j=1}^M I_{j,B} n_{j,B}$.

\medskip\noindent
\textbf{Step 2: Lower bounds.}

According to Lemma~\ref{lem:ELtau-vs-boundaries}, we obtain
\[
  \EE_A[L_\tau] \ge A_\alpha - c_{\mathrm{err}}\alpha (A_\alpha+B_\alpha+C_\ell),
\]
Since $x\log \frac{1}{x}\le 1/e$, we can also deduce that
\[
c_{\mathrm{err}} \alpha (A_\alpha+B_\alpha+C_\ell) \le 2c_{\mathrm{err}} \alpha \log\frac{1}{\alpha} + c_{\mathrm{err}} \alpha C_\ell\le \frac{2c_{\mathrm{err}}}{e} + c_{\mathrm{err}}C_\ell.
\]
Thus, we can set $C_0=(2c_{\mathrm{err}}/e) + c_{\mathrm{err}}C_\ell$.

By Step~1, $\EE_A[L_\tau] = \sum_j I_{j,A} n_{j,A}$, so we have established
the result under $\PP_A$. A similar argument can induce the second inequality in \eqref{eq:E-info-lb}.
\hfill \Halmos

\vspace{0.1in}

\subsection{Proof of Proposition~\ref{prop:cost-wait-decomp}}
By definition, the total cost is
\[
C_\pi = \sum_{t=1}^\tau c_{j_t}
= \sum_{j=1}^M c_j \sum_{t=1}^\tau \ind_{\{j_t=j\}}
= \sum_{j=1}^M c_j N_j(\tau).
\]
Taking expectations under $\PP_\theta$ gives
\[
\EE_\theta[C_\pi]=\sum_{j=1}^M c_j\EE_\theta[N_j(\tau)]=\sum_{j=1}^M c_j n_{j,\theta}.
\]

For the waiting time, by Assumption~\ref{assum:indep} and the law of iterated expectations, we have
\[
    \EE_\theta[W_\pi]=\sum_{t\ge1} \EE_\theta[\ind\{t\le\tau\}T_{j_t,t}]=\sum_{t\ge1} \EE_\theta[\ind\{t\le\tau\}\mu_{j_t}]=\sum_{j=1}^M \mu_j n_{j,\theta}.\tag*{\Halmos}
\]

\vspace{0.1in}

\subsection{Proof of Proposition~\ref{prop:lower_bound}}

We start by proving the following inequality: Under any admissible policy $\pi$, we have
\begin{equation}\label{eq:g-Jensen-lower}
    \EE[C_\pi + g(W_\pi)] \ge    F(\bs{n}_A, \bs{n}_B)
\end{equation}
We first verify that $\EE[W_\pi]<\infty$. By Proposition~\ref{prop:cost-wait-decomp}, we have
\[
\EE[W_\pi]=\sum_{\theta\in\{A,B\}} \xi_\theta \sum_{j=1}^M \mu_j n_{j,\theta} \le \mu_{\max}\mathbb{E}[\tau] <\infty.
\]

Since $g:[0,\infty)\to[0,\infty)$ is convex by Assumption~\ref{assum:g}, and $W_\pi\ge0$ almost surely (since waiting times are nonnegative), Jensen's inequality for convex functions states that
\[
\EE[g(W_\pi)] = \xi_A\EE_A[g(W_\pi)] + \xi_B\EE_B[g(W_\pi)]\;\ge\; \xi_Ag(\EE[W_\pi^{(A)}]) + \xi_Bg(\EE[W_\pi^{(B)}]).
\]

Substituting the expression of $\EE[W_\pi]$ in \eqref{eq:cost-decomp-theta} into the above inequality, we have:
\[
\EE[g(W_\pi)]
\;\ge\;
\xi_Ag\!\Big(\sum_{j=1}^M \mu_j n_{j, A}\Big) + \xi_Bg\!\Big(\sum_{j=1}^M \mu_j n_{j, B}\Big).
\]

Inequality~\eqref{eq:g-Jensen-lower}  then follows from the above inequality and Proposition~\ref{prop:cost-wait-decomp}.

We now construct the following lower bound problem:
\begin{align}
    \min_{\bs{n}_A, \bs{n}_B} \ \ & F (\bs{n}_A, \bs{n}_B)\tag{LB'}\label{LB0} \\
   \mbox {s.t.}  \quad   & \sum_{j=1}^M I_{j,A} n_{j,A} \ge S_A(\alpha),\ \ \notag \\
 & \sum_{j=1}^M I_{j,B} n_{j,B} \ge S_B(\alpha), \ \ \ \notag \\
 & n_{j,A}, n_{j,B} \ge 0  \ \ \ \forall j\in [M]. \notag
\end{align}

By Propositions~\ref{prop:info-budget-asym} and inequality~\eqref{eq:g-Jensen-lower}, under any feasible policy, the expected numbers of queries $(\bs{n}_A, \bs{n}_B)$ are feasible to \eqref{LB0} and its expected cost is greater than that in \eqref{LB0}. Thus, we can derive that \eqref{LB0} is a lower bound.

Then, we show that any minimizer of \eqref{LB0} must satisfy both constraints with equality.
Suppose $(\bs{n}_A^\star, \bs{n}_B^\star)$ is a minimizer of \eqref{LB0} subject to
\eqref{eq:E-info-lb}, but
\[
  \sum_{j=1}^M I_{j,A} n_{j,A}^\star > S_A(\alpha).
\]
We will construct a feasible point with strictly smaller objective value, contradicting
optimality.

Define the scaling factor
\[
  \lambda := \frac{S_A(\alpha)}{\sum_{\ell=1}^M I_{\ell,A} n_{\ell,A}^\star}.
\]
Since $\sum_{\ell=1}^M I_{\ell,A} n_{\ell,A}^\star > S_A(\alpha)$ by assumption, we have
$\lambda < 1$.
Consider the scaled vector $\tilde{n}_{j,A}
  := \lambda \, n_{j,A}^\star$ for $j=1,\ldots,M$,
and keep $\tilde{\bs{n}}_B := \bs{n}_B^\star$ unchanged. Since the objective is non-decreasing in $n_{j, A}$'s, it can be verified that $(\tilde{\bs{n}}_A, \bs{n}_B^\star)$ is feasible and $F(\tilde{\bs{n}}_A, \bs{n}_B^\star) < F(\bs{n}_A^\star, \bs{n}_B^\star)$.
Therefore, any minimizer must satisfy the first constraint with equality. The same argument applies to the second constraint.

Lastly due to the convexity of $g$, $F(\bs{n}_A, \bs{n}_B)$ is also convex in $(\bs{n}_A, \bs{n}_B)$.  \hfill \Halmos

\vspace{0.1in}

\subsection{Proof of Proposition~\ref{prop:lower-mixture}}
For \eqref{ALO}, since $\Delta_M\times \Delta_M$ is compact, there exists a minimizer $(\bs{w}_A^\star, \bs{w}_B^\star)$.
Define
\[
z_A^\star:=S_A(\alpha)\sum_{j=1}^M \eta_{j, A}\cdot  w_{j, A}^\star, \qquad z_B^\star:=S_B(\alpha)\sum_{j=1}^M \eta_{j, B}\cdot  w_{j, B}^\star.
\]

\noindent
\textbf{KKT conditions rule out multiple actives generically:} The problem is convex with affine constraints, so the KKT conditions are necessary and sufficient.
Thus, there exist $\lambda_A\in\mathbb{R}$ and $\lambda_B\in \mathbb{R}$ such that for all $j$, we have
\begin{equation}\label{eq:KKT}
\begin{split}
&\xi_A S_A(\alpha) \kappa_{j, A} + \xi_A S_A(\alpha) \eta_{j, A} g'(z_A^\star)
\begin{cases}
=\lambda_A, & \text{if } w_{j, A}^\star>0,\\
\ge \lambda_A, & \text{if } w_{j, A}^\star=0,
\end{cases}, \\
&\xi_B S_B(\alpha) \kappa_{j, B} + \xi_B S_B(\alpha) \eta_{j, B} g'(z_B^\star)
\begin{cases}
=\lambda_B, & \text{if } w_{j, B}^\star>0,\\
\ge \lambda_B, & \text{if } w_{j, B}^\star=0.
\end{cases}
\end{split}
\end{equation}

In the following, we prove that there exists an optimal solution with $\bs{w}_A^\star=\bs{e}_{j^\star}$ and $\bs{w}_B^\star=\bs{e}_{k^\star}$ for some $j^\star$ and $k^\star$. We present the proof for $\bs{w}_A^\star=\bs{e}_{j^\star}$ in the following and the proof for $\bs{w}_B^\star=\bs{e}_{k^\star}$ is similar.

Suppose two distinct indices $j\ne k$ are active ($w_{j, A}^\star,w_{k, A}^\star>0$). Subtracting the equalities in \eqref{eq:KKT} gives
\begin{equation}\label{eq:key-balance}
\left(\kappa_{j, A} - \kappa_{k, A}\right) + g'(z_A^\star)\left(\eta_{j, A} - \eta_{k, A}\right) = 0.
\end{equation}

\noindent \textbf{Case 1: $\eta_{j, A}=\eta_{k, A}$.}
Then \eqref{eq:key-balance} becomes:
\[
\kappa_{j, A}- \kappa_{k, A} = 0,
\]
which implies $\kappa_{j, A} = \kappa_{k, A}$. In this case, LLMs $j$ and $k$ are equivalent for the optimization: they have the same cost-efficiency $\kappa_{j, A} = \kappa_{k, A}$ and the same time-efficiency $\eta_{j, A} = \eta_{k, A}$. Therefore, for any allocation $(w_{j, A}, w_{k, A})$ with $w_{j, A}, w_{k, A} > 0$, we can transfer all mass to a single LLM (say, set $\tilde{w}_{j, A} = w_{j, A} + w_{k, A}$ and $\tilde{w}_{k, A} = 0$) without changing the objective value, i.e., $G_\alpha(\bs{w}_A, \bs{w}_B)=G_\alpha(\tilde{\bs{w}}_A, \bs{w}_B)$. Thus, even when two such equivalent LLMs are active, there exists an equally optimal solution that is concentrated. If multiple LLMs are active, we can also find an equally optimal concentrated solution by moving all mass to one LLM.
\vspace{0.05in}

\noindent \textbf{Case 2: $\eta_{j, A}\ne\eta_{k, A}$.}
Equation \eqref{eq:key-balance} forces
\begin{equation}\label{eq:g-prime-K}
g'(z_A^\star)=K_{jk}:=\frac{\kappa_{k, A}-\kappa_{j, A}}{\eta_{j, A}-\eta_{k, A}}.
\end{equation}

We now show that for large $S_A(\alpha)$, $g'(z_A^\star)$ avoids all $K_{jk}$, so generically there is at most one active coordinate. Note that
\[
z_A^\star = S_A(\alpha)\sum_{j=1}^M \eta_{j, A}\cdot  w_{j, A}^\star \geq \eta_{\text{min}} \cdot S_A(\alpha),
\]
where $\eta_{\min} := \min_{j} \eta_{j, A} > 0$. Therefore, as $\alpha \to 0$, we have $z_A^\star \to \infty$. Recall that by Assumption~\ref{assum:g}, we have $g(x)=c\cdot x^\rho$. Then, we analyze two cases:

\medskip
\textbf{Subcase (a): $\rho > 1$.}
In this case, we have $g(x) = c\cdot x^\rho$ and $g'(x)=c\rho\cdot  x^{\rho-1}$. Thus, as $\alpha\to 0$, we have $z_A^\star\to \infty$ and hence $g'(z_A^\star)$ avoids all $K_{jk}$'s.

\medskip
\textbf{Subcase (b): $\rho=1$.} Since $g(x)=c\cdot x$, we have $g'(x)=c$ for any $x$. If $\min_{\eta_j\neq \eta_k}|c-K_{jk}|>0$, then $g'(x)$ avoids all $K_{jk}$.

\vspace{1mm}
Suppose $c=K_{jk}$ for some pair with $\eta_{j, A}\neq\eta_{k, A}$.
In this case, the KKT conditions allow both coordinates $j$ and $k$ to be active,
so we must check whether an interior split can actually be optimal.

Fix such a pair $(j,k)$ and hold all other coordinates constant.
Let $s:=w_{j, A}+w_{k, A}>0$. We only need to consider the reduced one-dimensional objective
\[
\phi_s(\theta)
:= \xi_A S_A(\alpha) \kappa_{j, A} \theta s+ \xi_A S_A(\alpha) \kappa_{k, A} (1-\theta)s
   + c\cdot S_A(\alpha)\cdot \Big(\xi_A \eta_{j, A} \theta s+ \xi_A \eta_{k, A}(1-\theta)s\Big),
\qquad \theta\in[0,1].
\]
Since $\phi_s(\theta)$ is linear in $\theta$, the minimum of $\phi_s$ lies at an endpoint, and the allocation can be concentrated on a single coordinate.

\medskip
\noindent\textbf{Conclusion:}
For all sufficiently large $S_A(\alpha)$, there exists a minimizer that has at most
one positive coordinate for each hypothesis. Moreover, the optimal coordinates $(j^\star, k^\star)$ can be computed as
\begin{align*}
        j^\star\in \argmin_{j\in [M]} \left\{S_A(\alpha)\kappa_{j, A} + g\!\left(S_A(\alpha)\eta_{j, A}\right)\right\} \qquad k^\star\in \argmin_{k\in [M]} \left\{S_B(\alpha)\kappa_{k, B} + g\!\left(S_B(\alpha)\eta_{k, B}\right)\right\}.
\end{align*}
\hfill \Halmos
\vspace{0.1in}

\subsection{Proof of Theorem~\ref{thm:two-llm-upper}}\label{sec:ub_proof}
As $\alpha\to 0$, the order for $\Psi_\alpha$ is straightforward to verify.

\noindent \textbf{Step 1: Decomposing and bounding the query cost and waiting time.}

Let $\pi = \pi(j_A^*,j_B^*)$. We start by decomposing the expected query cost and expected waiting time. Note that the total query cost can be written as
\[
  C_\pi = \sum_{t=1}^\tau c_{j_t}
    = c_{j_A^\star} N_{j^\star_A}(\tau) + c_{j_B^\star} N_{j^\star_B}(\tau).
\]
According to Proposition~\ref{prop:cost-wait-decomp} and the definition of $n_{j, \theta}$, if we take expectations under the prior $(\xi_A,\xi_B)$ using the law of total expectation, we have
\begin{align}
   \EE[C_\pi]
  &= \xi_A \EE_A[C_\pi] + \xi_B \EE_B[C_\pi]  \notag \\
  &= \xi_A (c_{j_A^\star} n_{j_A^\star,A} + c_{j_B^\star} n_{j_B^\star,A})
     + \xi_B (c_{j_A^\star} n_{j_A^\star,B} + c_{j_B^\star} n_{j_B^\star,B})  \notag \\
    &= c_{j_A^\star} \cdot  \xi_A \cdot n_{j_A^\star, A}  + c_{j_B^\star} \cdot \xi_B\cdot n_{j_B^\star, B} + \left(  c_{j_A^\star} \cdot \xi_B \cdot n_{j_A^\star, B} + c_{j_B^\star} \cdot \xi_A\cdot n_{j_B^\star, A} \right) \notag\\
&= \kappa_{j_A^\star, A} \cdot  \xi_A \cdot S_A(\alpha) \cdot \frac{n_{j_A^\star, A} \cdot I_{j_A^\star, A}}{S_A(\alpha)} + \kappa_{j_B^\star, B} \cdot \xi_B \cdot S_B(\alpha) \cdot \frac{n_{j_B^\star, B}\cdot I_{j_B^\star, B} }{S_B(\alpha)} \notag \\
&\hspace{8mm}+ \left(  c_{j_A^\star} \cdot \xi_B \cdot n_{j_A^\star, B} + c_{j_B^\star} \cdot \xi_A\cdot n_{j_B^\star, A} \right).
\label{C-pi}
\end{align}
\noindent
Following a similar logic, we can derive
\begin{align}
  \EE_\theta[W_\pi]
&= \eta_{j_\theta^\star, \theta}  \cdot S_\theta(\alpha) \cdot \frac{n_{j_\theta^\star, \theta} \cdot I_{j_\theta^\star, \theta}}{S_\theta(\alpha)} + \mu_{j_{-\theta}^\star}\cdot n_{j_{-\theta}^\star, \theta}, \label{W-pi}
\end{align}
where $-\theta$ is the alternative hypothesis rather than $\theta$.

In the following proposition, we bound the ``wrong-side'' query counts $n_{j_B^\star,A}$ and $n_{j_A^\star,B}$.

\begin{proposition}\label{prop:wrong-side-query}
    Under policy  $\pi(j_A^\star, j_B^\star)$, we can bound the wrong-side query counts $n_{j_B^\star,A}$ and $n_{j_A^\star,B}$ by some constant independent of $\alpha$, i.e., $n_{j_B^\star,A} = O(1)$ and $n_{j_A^\star,B} = O(1)$.
\end{proposition}
\begin{proof}{Proof of Proposition~\ref{prop:wrong-side-query}.}
We prove the bound on $n_{j_B^\star, A}$; the argument for $n_{j_A^\star, B}$ is symmetric.

By Lemma~\ref{lem:drift-decomp}, under $\PP_A$ the cumulative LLR admits the drift--martingale decomposition
\[
L_t=\sum_{s=1}^t I_{j_s,A}+M_t^{(A)},\qquad |M_t^{(A)}-M_{t-1}^{(A)}|\le 2C_\ell,\qquad \EE_A[M_t^{(A)}]=0.
\]
Set $m_A:=\min\{I_{j_A^\star,A},I_{j_B^\star,A}\}>0$. Since each conditional drift $I_{j_s,A}\ge m_A$, we have the \emph{pathwise} bound
\[
\sum_{s=1}^{t-1} I_{j_s,A}\ge (t-1)\,m_A\qquad\text{for every realization}.
\]
Thus on $\{L_{t-1}<-\delta\}$, we have
\[
M_{t-1}^{(A)}=L_{t-1}-\sum_{s=1}^{t-1}I_{j_s,A}<-\delta-(t-1)m_A.
\]
Define $t_0:=1+\lceil 2|\delta|/m_A\rceil$. For $t\ge t_0$, $\delta+(t-1)m_A\ge (t-1)m_A/2>0$, so
\[
\PP_A(L_{t-1}<-\delta)\le \PP_A\!\left(M_{t-1}^{(A)}<-(t-1)m_A/2\right).
\]
Apply Freedman's inequality (Lemma~\ref{lem:Freedman}) to the martingale $(M_s^{(A)})$ with bounded differences $|M_t^{(A)}-M_{t-1}^{(A)}|\le 2C_\ell$ and predictable variation $V_{t-1}^{(A)}\le v_\ell^2(t-1)$:
\begin{align*}
\PP_A\!\left(M_{t-1}^{(A)}<-(t-1)m_A/2\right)
&\le \exp\!\left(-\frac{((t-1)m_A/2)^2}{2(v_\ell^2(t-1)+2C_\ell\cdot (t-1)m_A/2)}\right)\\
&=\exp\!\left(-\frac{(t-1)\,m_A^2}{8(v_\ell^2+C_\ell m_A)}\right)= e^{-c_0(t-1)},
\end{align*}\vspace{2mm}
where $c_0:=m_A^2/[8(v_\ell^2+C_\ell m_A)]>0$ depends only on $m_A,v_\ell^2,C_\ell$ (not on $\alpha$).

By the sign-based policy, $j_t=j_B^\star$ if and only if $L_{t-1}<-\delta$. Therefore
\[
n_{j_B^\star, A}=\sum_{t=1}^\infty \PP_A(j_t=j_B^\star,\,t\le \tau)\le \sum_{t=1}^\infty \PP_A(L_{t-1}<-\delta)\le t_0+\sum_{t\ge t_0}e^{-c_0(t-1)}=:C_{\mathrm{neg}}^A,
\]
which is finite and depends only on $m_A,v_\ell^2,C_\ell,\delta$. The proof for the case when $\theta=B$ is symmetric, with $m_B:=\min\{I_{j_A^\star,B},I_{j_B^\star,B}\}$.\hfill \Halmos
\end{proof}
\vspace{0.1in}

Then, we prove that the ``correct-side'' query counts $I_{j_\theta^\star,\theta}\,n_{j_\theta^\star,\theta}$ are constant order away from $S_\theta(\alpha)$.

\begin{proposition}
\label{prop:query-count-asymptotics}
Under policy  $\pi(j_A^\star, j_B^\star)$, for all sufficiently
small $\alpha$, we have
\begin{align*}
I_{j_A^\star,A}\,n_{j_A^\star,A}=  S_A(\alpha) +O(1), \quad  I_{j_B^\star,B}\,n_{j_B^\star,B}= S_B(\alpha) + O(1).
\end{align*}
\end{proposition}
\begin{proof}{Proof of Proposition~\ref{prop:query-count-asymptotics}.}
We prove the statements under hypothesis $A$; the case $B$ is analogous.

First, according to \eqref{eq:L_tau}, under policy $\pi(j_A^\star, j_B^\star)$, we obtain $\EE_A[L_\tau] = I_{j_A^\star,A}\,n_{j_A^\star,A} + I_{j_B^\star,A}\,n_{j_B^\star,A}.$
Second, according to Lemma~\ref{lem:ELtau-vs-boundaries} and Proposition~\ref{prop:info-budget-asym}, we have $A_\alpha-C_0\le \EE_A[L_\tau] \le A_\alpha + C_\ell.$
Combining the above two results, we have:
\begin{equation}
\label{eq:drift-eq-with-wrong-side}
  A_\alpha-C_0\le I_{j_A^\star,A}\,n_{j_A^\star,A} + I_{j_B^\star,A}\,n_{j_B^\star,A}
  \le A_\alpha + C_\ell .
\end{equation}

First, since $I_{j_B^\star,A}>0$, we have
\[
  I_{j_A^\star,A}\,n_{j_A^\star,A} \le A_\alpha + C_\ell - I_{j_B^\star,A}\,n_{j_B^\star,A}
  \le A_\alpha + C_\ell = S_A(\alpha)+C_0 + C_\ell.
\]

Second, since $I_{j_B^\star,A}>0$ and $n_{j_B^\star, A}\le C_{\mathrm{neg}}^A$ (Proposition~\ref{prop:wrong-side-query}), we have
\[
I_{j_A^\star,A}\,n_{j_A^\star,A} \ge A_\alpha - C_0 - I_{j_B^\star,A}\,n_{j_B^\star,A} \ge A_\alpha -C_0- I_{j_B^\star,A}C_{\mathrm{neg}}^A
\]
Therefore, we have $I_{j_A^\star,A}\,n_{j_A^\star,A}=S_A(\alpha) + O(1)$.
 \hfill \Halmos
\end{proof}
\vspace{0.1in}

Combining Propositions~\ref{prop:wrong-side-query} and \ref{prop:query-count-asymptotics}, we can capture the asymptotic behaviors of the expectations~\eqref{C-pi} and~\eqref{W-pi} as follows:
\begin{equation}\label{eq:EC-EW-asym}
  \EE[C_\pi]
  = \xi_A S_A(\alpha)\,\kappa_{j_A^\star,A}
     + \xi_B S_B(\alpha)\,\kappa_{j_B^\star,B}
     + O(1), \quad
  \EE_\theta[W_\pi]
  = S_\theta(\alpha)\,\eta_{j_\theta^\star,\theta} + O(1).
\end{equation}

\noindent \textbf{Step 2: Bounding the waiting time penalty.}

We now address the challenge of matching the expected waiting time penalty incurred by $\pi$ to its deterministic benchmark $\xi_A g\! \left( \sum_{j=1}^M \mu_j n_{j,A}\right)+\xi_B g\!\left(\sum_{j=1}^M \mu_j n_{j,B}\right)$ underpinning the lower bound $\Psi_\alpha$. Let $W_\pi^{(\theta)}$ denote the total waiting time, conditioning on the hypothesis $\theta$. Let $N_\theta:=N_{j_\theta^\star}(\tau)$ denote the number of queries of LLM $j_\theta^\star$. To establish the results, the key step is to anchor
\[
W_\pi=\sum_{t=1}^\tau T_{j_t,t}=\sum_{\ell=1}^{N_A} T_{j_A^\star, \ell} + \sum_{\ell'=1}^{N_B} T_{j_B^\star, \ell'}  \mbox{ under each hypothesis }\theta,
\]
to its conditional mean
\[
m_\pi^\theta= \EE_\theta[W_\pi]=\mu_{j_A^\star}\cdot n_{j_A^\star, \theta} + \mu_{j_B^\star}\cdot n_{j_B^\star, \theta}.
\]
According to Proposition~\ref{prop:wrong-side-query}, under each hypothesis $\theta\in \{A, B\}$, when $\alpha\rightarrow 0$, the query will gradually concentrate on $j_\theta^\star$ so that $\EE_\theta[N_\theta] = \EE_\theta[\tau] + O(1)$. Motivated by this insight, we are going to match $W_\pi$ to its mean by establishing the convergence result $N_\theta\rightarrow \EE_\theta[\tau]$ for each $\theta$ and further control the deviation of the cumulative waiting time to its mean conditioned on the realized $\tau$.

Define $S_\pi:=\sum_{t=1}^\tau \mu_{j_t}$. In the following proposition, we provide properties of $W_\pi$ and $S_\pi$.
\begin{proposition}\label{prop:S_pi}
    For each $\theta\in \{A, B\}$, given $\pi=\pi(j_A^\star, j_B^\star)$, we have
    \begin{enumerate}
        \item Conditioning on $(\theta, \mathcal{H}_\tau)$, $W_\pi-S_\pi$ is sub-Gaussian with proxy $\psi^2\tau$, i.e., $\EE[e^{\lambda(W_\pi-S_\pi)}\mid \mathcal{H}_\tau, \theta]\le e^{\psi^2\lambda^2\tau/2}$.
        \item For each $\rho\ge 1$, we have $\EE_\theta[\tau^\rho]=O\left((\log(1/\alpha))^\rho\right)$.
        \item For each $p\ge 2$, we have $\EE_\theta|\tau-\EE_\theta[\tau]|^p=O(\left(\log(1/\alpha)\right)^{p/2})$ and $\EE_\theta|W_\pi-m_\pi^\theta|^p=O(\left(\log(1/\alpha)\right)^{p/2})$.
    \end{enumerate}
\end{proposition}
\begin{proof}{Proof of Proposition~\ref{prop:S_pi}.}
We start with the first statement. Fix $\theta\in\{A,B\}$. We will establish the moment generating function (MGF) bound conditional on $\theta$.
By definition, we have $W_\pi - S_\pi = \sum_{t=1}^\tau (T_{j_t,t} - \mu_{j_t})$.
By Assumption~\ref{assum:indep}, the array $(T_{j,t})_{j,t}$ is independent of $(\theta,(Y_t)_{t\ge1})$ and of the policy's internal randomization. Consequently, under the sign-based policy, conditional on $(\theta,\mathcal{H}_\tau)$:
\begin{itemize}
    \item The LLM indices $(j_t)_{t\le\tau}$ and the stopping time $\tau$ are $\mathcal{H}_\tau$-measurable, hence they are deterministic given $(\theta,\mathcal{H}_\tau)$.
    \item Given $j_{1:\tau}$, the waiting times $(T_{j_t,t})_{t\le\tau}$ are independent of $(\theta,\mathcal{H}_\tau)$ and mutually independent across $t$ (since $(T_{j,t})_{t\ge1}$ is i.i.d.\ in $t$ for each fixed $j$ by Assumption~\ref{assum:indep}).
\end{itemize}

Then, we can compute the conditional moment generating function as follows:
\begin{align*}
\EE_\theta\big[e^{\lambda (W_\pi - S_\pi)} \,\big|\, \mathcal{H}_\tau, \theta\big]
& = \EE_\theta\Big[\prod_{t=1}^\tau e^{\lambda (T_{j_t,t}-\mu_{j_t})} \,\Big|\, \mathcal{H}_\tau,\theta\Big] = \prod_{t=1}^\tau \EE_\theta\big[e^{\lambda (T_{j_t,t}-\mu_{j_t})} \,\big|\, \mathcal{H}_\tau,\theta\big]\\
&= \prod_{t=1}^\tau \EE\big[e^{\lambda (T_{j_t,t}-\mu_{j_t})}\mid j_t\big] \le \exp\left( \frac{\psi^2\lambda^2}{2} \tau\right).
\end{align*}

\vspace{0.2in}
Then, we prove the second statement. Similar to the proof of Proposition~\ref{prop:wrong-side-query}, for $t\ge t_1:=2A_\alpha/I_{\min}$, we have
\begin{align*}
    \PP_A(\tau>t) \le \PP_A(L_t< A_\alpha) \le \PP_A(L_t< I_{\min}t/2) \le e^{-c_3 t}.
\end{align*}
Thus, we have
\[
\EE_A[\tau^\rho] = \int_0^\infty \rho t^{\rho-1}\PP_A(\tau>t)dt \; \le \; t_1^\rho + \int_{t_1}^\infty \rho t^{\rho-1}e^{-c_3 t}dt = O\left((\log(1/\alpha))^\rho\right).
\]

\vspace{0.2in}
Lastly, we prove the third statement. First, we prove that $\EE_\theta|\tau - \EE_\theta[\tau]|^p=O\left((\log(1/\alpha))^{p/2}\right)$ for each $p\ge 2$. We fix $\theta=A$ and the proof for the case with $\theta=B$ is symmetric.

By Lemma~\ref{lem:drift-decomp}, we have
\[
L_\tau = I_{j_A^\star, A} N_A + I_{j_B^\star, A} N_B + M_t^{(A)},
\]
where $(M_t^{(A)})_{t\ge 0}$ is a martingale under $\PP_A$ with $\EE_A[M_A^{(A)}]=0$ and bounded increments. Since $\tau=N_A+N_B$, we have
\[
I_{j_A^\star, A}(\tau-\EE_A[\tau])=(L_\tau-\EE_A[L_\tau]) - (I_{j_B^\star, A}-I_{j_A^\star, A})(N_B-\EE_A[N_B]) - M_\tau^{(A)}.
\]
By the convexity of $|x|^p$, we have $|x+y+z|^p\le 3^{p-1}(|x|^p+|y|^p+|z|^p)$. Thus, we can derive that
\[
I_{j_A^\star, A}^p\EE_A|\tau-\EE_A[\tau]|^p \le 3^{p-1}\left(\EE_A|L_\tau-\EE_A[L_\tau]|^p + |I_{j_B^\star, A}-I_{j_A^\star, A}|^p\cdot \EE_A|N_B-\EE_A[N_B]|^p + \EE_A|M_\tau^{(A)}|^p \right).
\]
For the first term, by Lemmas~\ref{lem:error-bounds}, \ref{lem:overshoot} and \ref{lem:ELtau-vs-boundaries}, we have
\begin{align*}
\EE_A|L_\tau-\EE_A[L_\tau]|^p &= \PP_A(D=A) \cdot \EE_A[|L_\tau-\EE_A[L_\tau]|^p\mid D=A] + \PP_A(D=B) \cdot \EE_A[|L_\tau-\EE_A[L_\tau]|^p\mid D=B]\\
&\le \EE_A[(|L_\tau - A_\alpha| + |A_\alpha-\EE_A[L_\tau]|)^p\mid D=A] + c_{err}\alpha \EE_A[(|L_\tau - A_\alpha| + |A_\alpha-\EE_A[L_\tau]|)^p\mid D=B]\\
&= O(1) + O(\alpha\cdot (\log(1/\alpha))^p) = O(1).
\end{align*}

\vspace{0.05in}
For the second term, since $|x+y|^p\le 2^{p-1}(|x|^p+|y|^p)$, we have
\begin{align*}
    \EE_A|N_B-\EE_A[N_B]|^p \le 2^{p-1}\cdot \left(\EE_A[N_B^p]+(\EE_A[N_B])^p \right).
\end{align*}
By Proposition~\ref{prop:wrong-side-query}, we have $\EE_A[N_B]=O(1)$. According to the proof of Proposition~\ref{prop:wrong-side-query}, for $k\ge t_0$, we have
\begin{align*}
    \PP_A(N_B\ge k) \le \sum_{t\ge k}\PP_A(L_{t-1}<-\delta) \le \sum_{t\ge k} e^{-c_0(t-1)} = \frac{e^{-c_0(k-1)}}{1-e^{-c_0}} \le C_1 e^{-c_0 k},
\end{align*}
where $C_1$ is a constant independent of $\alpha$. Then, we can deduce that
\begin{align*}
    \EE_A[N_b^p] = \sum_{k=1}^\infty pk^{p-1} \PP_A(N_B\ge k) \le \sum_{k=1}^{t_0} pk^{p-1} + \sum_{k=t_0+1}^\infty pk^{p-1}C_1 e^{-c_0 k} = O(1).
\end{align*}
Thus, we have $\EE_A|N_B-\EE_A[N_B]|^p=O(1)$.

\vspace{0.05in}
For the third term, according to Burkholder's inequality (Theorem~2.10 in \citealt{hall2014martingale}), for $p>1$, there exists a constant $C_3$ such that
\begin{align*}
    \EE_A|M_\tau^{(A)}|^p \le C_3\EE_A\left|\sum_{t=1}^\tau (M_t^{(A)}-M_{t-1}^{(A)})^2\right|^{p/2}\le C_3 \cdot (2C_\ell)^{p}\cdot \EE_A|\tau|^{p/2} = O\left((\log(1/\alpha))^{p/2}\right).
\end{align*}

Therefore, we can derive that $\EE_A|\tau - \EE_A[\tau]|^p=O\left((\log(1/\alpha))^{p/2}\right)$. Similarly, we can prove that $\EE_B|\tau - \EE_B[\tau]|^p=O\left((\log(1/\alpha))^{p/2}\right)$.

\vspace{0.1in}
Since $|x+y|^p\le 2^{p-1}(|x|^p+|y|^p)$, we have
\begin{align*}
    \EE_\theta|W_\pi-m_\pi^\theta|^p \le 2^{p-1} \cdot (\EE_\theta|W_\pi-S_\pi|^p+\EE_\theta|S_\pi-m_\pi^\theta|^p).
\end{align*}
For the first term, since $W_\pi-S_\pi\mid \mathcal{H}_\tau, \theta$ is sub-Gaussian, we have
\begin{align*}
    \EE_\theta|W_\pi-S_\pi|^p &= \EE_\theta[\EE[|W_\pi-S_\pi|^p\mid \mathcal{H}_\tau, \theta]] = \EE_\theta\left[\int_0^\infty p t^{p-1} \PP(|W_\pi-S_\pi|>t\mid \mathcal{H}_t, \theta) dt \right]\\
    &=\EE_\theta\left[\int_0^\infty p t^{p-1} 2e^{-t^2/(2\psi^2\tau)} dt \right] \le O(\EE_\theta[\tau^{p/2}]) = O\left((\log(1/\alpha))^{p/2}\right).
\end{align*}

For the second term, we have
\begin{align*}
    \EE_A|S_\pi-m_\pi^A|^p &= \EE_A\left|\mu_{j_A^\star}(\tau-\EE_A[\tau])+(\mu_{j_B^\star}-\mu_{j_A^\star})(N_B-\EE_A[N_B)\right|^p\\
    &\le 2^{p-1}\left(\mu_{j_A^\star}^p \EE_A|\tau-\EE_A[\tau]|^p+|\mu_{j_B^\star}-\mu_{j_A^\star}|^p\cdot \EE_A|N_B-\EE_A[N_B]|^p\right)\\
    &= O\left((\log(1/\alpha))^{p/2}\right) + O(1) = O\left((\log(1/\alpha))^{p/2}\right).
\end{align*}

Therefore, we have $\EE_\theta|W_\pi-m_\pi^\theta|^p\le O\left((\log(1/\alpha))^{p/2}\right)$. \hfill \Halmos
\end{proof}

\vspace{3mm}

We exploit the polynomial assumption imposed on $g$ to control the fluctuation of $g(W_\pi)$ around $g(\EE[W_\pi])$. Although $g$ is convex and increasing, its variation along the sample path remains regulated by the exponential tail of $M_t^{(\theta)}$. In particular, since $\EE_\theta[g'(m_\pi^\theta)(W_\pi-m_\pi^\theta)] = g'(m_\pi^\theta)\cdot \EE_\theta[W_\pi-m_\pi^\theta] = 0$, we have
\begin{eqnarray*}
    & & \EE_\theta[g(W_\pi)] - g(m_\pi^\theta) \\
    & & \hspace{5mm} = \EE_\theta\left[g(W_\pi) - g(m_\pi^\theta) - g'(m_\pi^\theta)(W_\pi-m_\pi^\theta) \right] \\
    & & \hspace{5mm} = c\cdot \left(m_\pi^\theta\right)^\rho\cdot \EE_\theta[\phi(W_\pi/m_\pi^\theta)],
\end{eqnarray*}
where $\phi(\cdot)$ is defined in Appendix~\ref{append:wait_func}. By Lemma~\ref{lem:wait_func} and Proposition~\ref{prop:S_pi}, we have
\begin{eqnarray}
    & & \EE_\theta[g(W_\pi)] - g(m_\pi^\theta) \notag \\
    & & \hspace{5mm} = c\cdot \left(m_\pi^\theta\right)^\rho\EE_\theta[\phi(W_\pi/m_\pi^\theta)] \le c\cdot c_1\cdot \left((m_\pi^\theta)^{\rho-2}\EE_\theta\left[(W_\pi - m_\pi^\theta)^2\right] + (m_\pi^\theta)^{\rho-q}\EE_\theta\left[(W_\pi - m_\pi^\theta)^q\right] \right) \notag \\
    && \hspace{5mm} \le O\left((\log(1/\alpha))^{\rho-1}\right) + O\left((\log(1/\alpha))^{\rho-q/2}\right) \notag\\
    & & \hspace{5mm} = O\left((\log(1/\alpha))^{\rho-1}\right) \label{eq:gW-asym-a}
\end{eqnarray}
where the last equality follows since $q\ge 2$.

Moreover, by Propositions~\ref{prop:wrong-side-query} and \ref{prop:query-count-asymptotics}, we have $\EE_A[\tau]=n_{j_A^\star, A}+n_{j_B^\star, A} = \Theta(\log(1/\alpha))$, which implies that
\begin{equation}\label{eq:gW-asym-b}
    g(m_\pi^\theta)=\Theta\left((\log(1/\alpha))^{\rho}\right).
\end{equation}

\noindent \textbf{Step 3: Combining the bounds.}

By equation~\ref{eq:EC-EW-asym},~\eqref{eq:gW-asym-a} and~\eqref{eq:gW-asym-b}, we have
\begin{align*}
    \mathcal{R}(\pi;\alpha) &= \EE[C_\pi]+\EE[g(W_\pi)] = \EE[C_\pi]+\xi_Ag(\EE_A[W_\pi]) + \xi_Bg(\EE_B[W_\pi]) + O\left((\log(1/\alpha))^{\rho-1}\right)\\
    &= \xi_AS_A(\alpha)\kappa_{j_A^\star, A} + \xi_BS_B(\alpha)\kappa_{j_B^\star, B} + O(1) \\
    &\qquad + \xi_Ag\left(S_A(\alpha)\eta_{j_A^\star, A} +O(1) \right) + \xi_B g\left(S_B(\alpha)\eta_{j_B^\star, B} +O(1) \right) + O\left((\log(1/\alpha))^{\rho-1}\right)\\
    &= \xi_AS_A(\alpha)\kappa_{j_A^\star, A} + \xi_BS_B(\alpha)\kappa_{j_B^\star, B} + \xi_Ag\left(S_A(\alpha)\eta_{j_A^\star, A} \right) + \xi_B g\left(S_B(\alpha)\eta_{j_B^\star, B} \right) + O\left((\log(1/\alpha))^{\rho-1}\right)\\
    &= \Psi_\alpha + O\left((\log(1/\alpha))^{\rho-1}\right).
\end{align*}
which establishes \eqref{eq:two-llm-additive-rate}.  \hfill \Halmos

\vspace{0.1in}

 \subsection{Proof of Theorem~\ref{thm:two-llm-tight}}
Assume $c>0$ and $\rho>1$. We work first under $A$ and temporarily suppress the hypothesis subscript on $I_j,\kappa_j,\eta_j$. Define
\[
\eta^*:=\min_j\eta_j>0,\qquad
\kappa^*:=\min_{j:\eta_j=\eta^*}\kappa_j,\qquad
v_{\min}:=\min_j\Var_A(\ell_j(Y))>0.
\]
For the chosen policy, write $n_j=\EE_A[N_j(\tau)]$, $m=\EE_A[W_\pi]$, and
\[
Q(\bs n):=\sum_j I_jn_j,\qquad
F_A(\bs n):=\sum_j c_jn_j+g\!\left(\sum_j\mu_j n_j\right),
\qquad
\Psi_{\alpha,A}:=\min_{\bs x\ge0:\,Q(\bs x)\ge S_A(\alpha)}F_A(\bs x).
\]
The definition of $F_A$ agrees with the conditional objective in the proof of Proposition~\ref{prop:lower_bound}.

Set $d_j:=\mu_j-\eta^*I_j=I_j(\eta_j-\eta^*)\ge0$. Define stopped processes, with zero initial values, by
\begin{align*}
V_t&:=\sum_{s=1}^{t}\mathbf1\{s\le\tau\}\left(T_{j_s,s}-\eta^*\ell_{j_s}(Y_s)\right),\\
\Gamma_t&:=\sum_{s=1}^{t}\mathbf1\{s\le\tau\}d_{j_s},
\qquad Z_t:=V_t-\Gamma_t.
\end{align*}
Let $\Gamma:=\Gamma_\tau$ and $Z:=Z_\tau$. Then
\begin{equation}\label{eq:waiting_decomp}
W_\pi-\eta^*L_\tau=\Gamma+Z.
\end{equation}
The nonnegative term $\Gamma$ measures excess expected delay relative to the best time per unit of evidence. The martingale term $Z$ captures fresh delay noise and evidence noise. Sources tied in time efficiency have $d_j=0$ and cause no difficulty in this decomposition.

\subsubsection{Supporting Results}
We give the moment and deviation properties of~\eqref{eq:waiting_decomp}.
\begin{lemma}\label{lem:decomposition}
There are constants $\underline\sigma^2=(\eta^*)^2v_{\min}>0$, $C_Z<\infty$, and $K_0<\infty$, independent of the policy and $\alpha$, such that
\[
\Gamma_t\text{ is nondecreasing},\qquad
\EE_A[\Gamma]=\sum_j d_jn_j,
\]
\[
\EE_A[Z]=0,\qquad
\EE_A[Z^2]\ge\underline\sigma^2 t_\alpha,\qquad
\EE_A[Z^4]\le C_Z \log(1/\alpha)^2,
\]
and, on $\{D=A\}$,
\[
|W_\pi-m|\ge|Z|-\Gamma-\EE_A[\Gamma]-K_0.
\]
Moreover, $\EE_A[Z^2]=O(\log(1/\alpha))$ uniformly over policies.
\end{lemma}
\noindent\textbf{Proof.}
The claims about $\Gamma$ follow from its nonnegative predictable increments and Tonelli's theorem. The increments of $Z_t$ are
\[
\Delta Z_t=\mathbf1\{t\le\tau\}
\left[(T_{j_t,t}-\mu_{j_t})-\eta^*(\ell_{j_t}(Y_t)-I_{j_t})\right].
\]
They have conditional mean zero. Given the pre-query history, the two centered terms are independent, by Assumption~\ref{assum:indep}. Their conditional MGFs and variances therefore give
\[
\EE_A[e^{\lambda\Delta Z_t}\mid\mathcal G_{t-1}]
\le\exp\!\left(\frac{\widetilde\psi^2\lambda^2}{2}\mathbf1\{t\le\tau\}\right),
\qquad
\widetilde\psi^2:=\psi^2+4(\eta^*)^2C_\ell^2,
\]
\[
\EE_A[(\Delta Z_t)^2\mid\mathcal G_{t-1}]
=\mathbf1\{t\le\tau\}\left(\sigma_{j_t}^2+(\eta^*)^2\Var_A(\ell_{j_t}(Y_t))\right)
\ge\underline\sigma^2\mathbf1\{t\le\tau\}.
\]
The MGF bound also bounds the conditional variance above by $\widetilde\psi^2\mathbf1\{t\le\tau\}$. Orthogonality of martingale increments consequently gives
\[
\EE_A[Z_t^2]\le\widetilde\psi^2\EE_A[t\wedge\tau].
\]
The martingale is Cauchy in $L^2$, because the mean square of its increment after time $t$ is at most $\widetilde\psi^2\EE_A[(\tau-t)^+]\to0$. Its almost-sure terminal value $Z$ thus has mean zero, and
\[
\underline\sigma^2\EE_A[\tau]\le\EE_A[Z^2]
\le\widetilde\psi^2\EE_A[\tau]=O(\log(1/\alpha)).
\]
Use $\tau\ge t_\alpha$ for the lower bound.

Moreover, the conditional sub-Gaussian bound on $\Delta Z_t$ implies that,
for some constant $C_4<\infty$ independent of $\alpha$ and the policy,
\[
\EE_A[|\Delta Z_t|^4\mid\mathcal G_{t-1}]
\le C_4\mathbf 1\{t\le\tau\}.
\]
By the standard fourth-moment inequality for martingales,
\[
\EE_A[|Z_{n\wedge\tau}|^4]
\le
C\left(\EE_A[(n\wedge\tau)^2]+\EE_A[n\wedge\tau]\right).
\]
The same concentration argument as in the proof of Proposition~\ref{prop:S_pi},
using $I_{\min,A}:=\min_j I_{j,A}>0$, yields uniformly over admissible
policies $\EE_A[\tau^2]=O(\log^2(1/\alpha)).$
Hence,
\[
\sup_n\EE_A[|Z_{n\wedge\tau}|^4]
=O(\log^2(1/\alpha)).
\]
Letting $n\to\infty$ and applying Fatou's lemma gives
\[
\EE_A[Z^4]=O(\log^2(1/\alpha)).
\]

Finally, centering~\eqref{eq:waiting_decomp} gives
\[
W_\pi-m=Z+(\Gamma-\EE_A\Gamma)+\eta^*(L_\tau-\EE_A L_\tau).
\]
On $\{D=A\}$, the last term has magnitude at most $\eta^*(C_0+C_\ell)$ by the overshoot and expected-evidence bounds. Taking this as $K_0$ and applying the triangle inequality proves the deviation claim.\hfill\Halmos

\begin{lemma}\label{lem:wait_lb_1}
For all sufficiently small $\alpha$, the conditional benchmark is
\[
\Psi_{\alpha,A}=S_A(\alpha)\kappa^*+g(S_A(\alpha)\eta^*).
\]
There is a policy-independent $k_1>0$ such that
\[
F_A(\bs n)-\Psi_{\alpha,A}
\ge\frac{\iota_\alpha}{2}\EE_A[\Gamma]
\ge k_1 \log(1/\alpha)^{\rho-1}\EE_A[\Gamma],
\qquad \iota_\alpha:=g'(S_A(\alpha)\eta^*).
\]
\end{lemma}
\noindent\textbf{Proof.}
Let $S=S_A(\alpha)$ and
$D_j:=(\kappa_j-\kappa^*)+\iota_\alpha(\eta_j-\eta^*)$.
When $\eta_j=\eta^*$, $D_j=\kappa_j-\kappa^*\ge0$. When $\eta_j>\eta^*$, the divergence of $\iota_\alpha=c\rho(S\eta^*)^{\rho-1}$ implies
\[
D_j\ge\frac{\iota_\alpha}{2}(\eta_j-\eta^*)\ge0
\]
for all sufficiently small $\alpha$, uniformly across the finite source pool.

For any feasible allocation $\bs x$, the supporting-line inequality for $g$ at $S\eta^*$ yields
\begin{align*}
F_A(\bs x)-[S\kappa^*+g(S\eta^*)]
&\ge\sum_j I_jx_jD_j+(\kappa^*+\iota_\alpha\eta^*)[Q(\bs x)-S]\ge0.
\end{align*}
Equality is attained by allocating $S/I_j$ queries to a source with $\eta_j=\eta^*$ and $\kappa_j=\kappa^*$. This proves the benchmark formula. For the expected allocation induced by a policy, Proposition~\ref{prop:info-budget-asym} gives $Q(\bs n)\ge S$. Retaining the first term of the supporting-line bound gives
\[
F_A(\bs n)-\Psi_{\alpha,A}
\ge\frac{\iota_\alpha}{2}\sum_j I_jn_j(\eta_j-\eta^*)
=\frac{\iota_\alpha}{2}\EE_A[\Gamma].
\]
Since $\iota_\alpha=\Theta(\log(1/\alpha)^{\rho-1})$, the result follows. This argument retains the nonnegative contribution of excess evidence, rather than identifying the allocation gap solely with the use of a different source.\hfill\Halmos

\begin{lemma}\label{lem:waiting_lb_2}
There are constants $k_2>0$ and $\alpha_2>0$, independent of the policy, such that every admissible policy with $\EE_A[\Gamma]\le1$ satisfies
\[
\EE_A[g(W_\pi)]-g(\EE_A[W_\pi])\ge k_2\log(1/\alpha)^{\rho-1},
\qquad 0<\alpha<\alpha_2.
\]
\end{lemma}
\noindent\textbf{Proof.}
Proposition~\ref{prop:cost-wait-decomp} implies
\[
\mu_{\min}t_\alpha\le m\le\mu_{\max}\EE_A[\tau]=O(\log(1/\alpha)),
\]
so $m=\Theta(\log(1/\alpha))$ uniformly over policies. For sufficiently small $\alpha$, $t_\alpha\ge b_0\log(1/\alpha)$ for some fixed $b_0>0$, for example $b_0=1/(2C_\ell)$. Define
\[
u_\pi:=\frac12\sqrt{\EE_A[Z^2]}.
\]
Lemma~\ref{lem:decomposition} gives both $u_\pi\ge\tfrac12\sqrt{\underline\sigma^2b_0\log(1/\alpha)}$ and $u_\pi=O(\sqrt{\log(1/\alpha)})$. Applying the Paley--Zygmund inequality to $Z^2$ yields
\[
\PP_A(|Z|\ge u_\pi)
\ge\frac{9(\EE_A Z^2)^2}{16\EE_A Z^4}
\ge\frac{9\underline\sigma^4 b_0^2}{16C_Z}.
\]
Choose a fixed positive $p_0$ no larger than this last constant and $1/2$. Unlike a bound retaining $t_\alpha$ without its logarithmic denominator, this choice is independent of $\alpha$ and the policy.

Markov's inequality and Lemma~\ref{lem:error-bounds} give
\[
\PP_A(\Gamma>u_\pi/4)\le4/u_\pi,
\qquad\PP_A(D=B)\le c_{\mathrm{err}}\alpha.
\]
On $\{D=A,\Gamma\le u_\pi/4,|Z|\ge u_\pi\}$, Lemma~\ref{lem:decomposition} therefore implies
\[
|W_\pi-m|\ge u_\pi-u_\pi/4-1-K_0\ge u_\pi/2
\]
for sufficiently small $\alpha$, uniformly over the policies in question. Consequently,
\begin{equation}\label{eq:good_event}
\PP_A(|W_\pi-m|\ge u_\pi/2)
\ge p_0-4/u_\pi-c_{\mathrm{err}}\alpha\ge p_0/2.
\end{equation}

Since $\EE_A[W_\pi-m]=0$, the normalized convexity remainder gives
\[
\EE_A[g(W_\pi)]-g(m)=cm^\rho\EE_A[\phi(W_\pi/m)].
\]
For sufficiently small $\alpha$, $u_\pi/(2m)\le1/2$. On the event in~\eqref{eq:good_event}, Lemma~\ref{lem:wait_func} bounds the integrand below by $c_2(u_\pi/(2m))^2$; elsewhere it is nonnegative. Hence
\[
\EE_A[g(W_\pi)]-g(m)
\ge\frac{cc_2p_0}{8}m^{\rho-2}u_\pi^2
\ge k_2\log(1/\alpha)^{\rho-1}.
\]
Both upper and lower bounds on $m/\log(1/\alpha)$ were established, so the last step holds also when $1<\rho<2$.\hfill\Halmos

\subsubsection{Proof of the Lower Bound~\eqref{eq:rate-tightness}}
For each hypothesis, decompose the conditional gap as
\begin{equation}\label{eq:xi}
\Xi_\theta:=\EE_\theta[C_\pi+g(W_\pi)]-\Psi_{\alpha,\theta}
=\bigl[F_\theta(\bs n_\theta)-\Psi_{\alpha,\theta}\bigr]
+\bigl[\EE_\theta g(W_\pi)-g(\EE_\theta W_\pi)\bigr].
\end{equation}
Both terms are nonnegative. Under $A$, if $\EE_A[\Gamma]\ge1$, Lemma~\ref{lem:wait_lb_1} bounds the allocation term below by $k_1\log(1/\alpha)^{\rho-1}$. If $\EE_A[\Gamma]<1$, Lemma~\ref{lem:waiting_lb_2} gives a lower bound of this order on the fluctuation term. Thus $\Xi_A\ge K_A\log(1/\alpha)^{\rho-1}$ with a fixed $K_A>0$, uniformly over all policies.

Under $B$, repeat the argument with $-L_t$ as evidence in favor of the truth and with the $B$-specific information ratios. This gives a fixed $K_B>0$. The relaxed problem separates by hypothesis, so $\Psi_\alpha=\xi_A\Psi_{\alpha,A}+\xi_B\Psi_{\alpha,B}$. Therefore,
\[
\mathcal R(\pi;\alpha)-\Psi_\alpha
=\xi_A\Xi_A+\xi_B\Xi_B
\ge(\xi_AK_A+\xi_BK_B)\log(1/\alpha)^{\rho-1}.
\]
Taking the infimum over policies proves the theorem. The variance lower bound used only positive evidence variance; deterministic per-query response times and ties in time efficiency are both allowed.\hfill\Halmos

\vspace{0.1in}

\section{Omitted Proofs in Section~\ref{sec:adaptivity}}
\subsection{Proof of Proposition~\ref{prop:socp_static_lb_partI}}
First, we present the Berry-Esseen bound in \cite{shevtsova2010improvement}.
\begin{lemma}[Berry--Esseen Bound for Independent Summation; \citealt{shevtsova2010improvement}]\label{lem:BE_partI}
Let $X_1,\dots,X_t$ be independent with $\EE[X_i]=0$, variances
$\Var(X_i)=\sigma_i^2$, and third absolute moments
$\EE[|X_i|^3]<\infty$. Let $S_t:=\sum_{i=1}^t X_i$ and
$s_t^2:=\sum_{i=1}^t \sigma_i^2>0$. Then, for all $x\in\RR$,
\[
  \left|\PP\!\left(\frac{S_t}{s_t}\le x\right)-\Phi(x)\right|
  \le C_{\mathrm{BE}}\,
  \frac{\sum_{i=1}^t \EE[|X_i|^3]}{s_t^3},
\]
where $\Phi$ is the standard normal CDF and one may take
$C_{\mathrm{BE}}=0.56$.
\end{lemma}

Since $\ell_j(Y)\in\{\ell_j(A),\ell_j(B)\}$ and $\gamma_{j, \theta}\in (1/2, 1)$, the increments are bounded by $C_\ell<\infty$ and have finite third moments.
Under $\theta=A$ write
\[
L_\omega-\mu_A(\bs{n})=\sum_{j=1}^M\sum_{i=1}^{n_j} X_{j,i},
\qquad
X_{j,i}:=\ell_j(Y_{j,i})-\EE_A[\ell_j(Y)],
\]
where $\{Y_{j,i}\}$ are independent draws from model $j$ (static sampling).
The $X_{j,i}$ are independent, mean zero, and satisfy $|X_{j,i}|\le 2C_\ell$.
Let
\[
\mathcal{T}_\theta(\bs{n}):=\sum_{j=1}^M n_j\,\EE_\theta\!\left[\left|\ell_j(Y)-\EE_\theta[\ell_j(Y)]\right|^3\right].
\]
Then Lemma~\ref{lem:BE_partI} implies that for all $x\in\RR$,
\[
\left|\PP_A\!\left(\frac{L_\omega-\mu_A(\bs{n})}{\sigma_A(\bs{n})}< x\right)-\Phi(x)\right|\le C_{\mathrm{BE}}\frac{\mathcal{T}_A(\bs{n})}{\sigma_A(\bs{n})^3} \le \frac{\mathcal{T}_A(\bs{n})}{\sigma_A(\bs{n})^3}.
\]

As mentioned, a feasible policy should satisfy $\PP_A(D=B)\le \alpha/\xi_A$ and $\PP_B(D=A)\le \alpha/\xi_B$. Under $\theta=A$, by Berry-Esseen bound in Lemma~\ref{lem:BE_partI}, we have
\[
\PP_A(D=B)
=
\PP_A\!\left(\frac{L_n-\mu_A(\bs{n})}{\sigma_A(\bs{n})}< \frac{-\delta-\mu_A(\bs{n})}{\sigma_A(\bs{n})}\right) \ge \Phi\left(-\frac{\delta+\mu_A(\bs{n})}{\sigma_A(\bs{n})}\right) - \frac{\mathcal{T}_A(\bs{n})}{\sigma_A(\bs{n})^3}.
\]
Thus we have
\begin{align*}
    \frac{\mu_A(\bs{n})+\delta}{\sigma_A(\bs{n})} \ge \nu\left(\frac{\alpha}{\xi_A} + \frac{\mathcal{T}_A(\bs{n})}{\sigma_A(\bs{n})^3}\right).
\end{align*}

Similarly, for $\theta=B$, by Berry-Esseen bound in Lemma~\ref{lem:BE_partI}, we have
\[
\frac{\mu_B(\bs{n})-\delta}{\sigma_B(\bs{n})} \ge \nu\left(\frac{\alpha}{\xi_B} + \frac{\mathcal{T}_B(\bs{n})}{\sigma_B(\bs{n})^3}\right).
\]

\hfill \Halmos
\vspace{0.1in}

\subsection{Proof of Theorem~\ref{thm:VoA}}

For $j\in[M]$, let
\[
s_j:=a_j+b_j>0,\qquad r_j:=\frac{c_j}{s_j^2},\qquad
r^*:=\min_j r_j>0,\qquad j^*\in\operatorname*{arg\,min}_j r_j.
\]
Let $z_\alpha:=\nu(\alpha)=\Phi^{-1}(1-\alpha)$. The following two propositions state precisely the lower- and upper-limit conclusions needed for the theorem.

\begin{proposition}[Asymptotic static cost]\label{prop:asym_static}
With zero waiting cost and weak-signal scaling,
\[
\lim_{\alpha\to0}\liminf_{\epsilon\to0}
\frac{\epsilon^2\mathcal R^{\mathrm{st},*}(\alpha,\epsilon)}{z_\alpha^2}=r^*,
\qquad
\lim_{\alpha\to0}\limsup_{\epsilon\to0}
\frac{\epsilon^2\mathcal R^{\mathrm{st},*}(\alpha,\epsilon)}{z_\alpha^2}=r^*.
\]
\end{proposition}
\begin{proposition}[Asymptotic adaptive cost]\label{prop:asym_adaptive}
Under the same scaling,
\[
\lim_{\alpha\to0}\liminf_{\epsilon\to0}
\frac{\epsilon^2\mathcal R^*(\alpha,\epsilon)}{z_\alpha^2}=\frac{r^*}{4},
\qquad
\lim_{\alpha\to0}\limsup_{\epsilon\to0}
\frac{\epsilon^2\mathcal R^*(\alpha,\epsilon)}{z_\alpha^2}=\frac{r^*}{4}.
\]
\end{proposition}

The propositions are proved in the next two subsections. To combine them rigorously, define
\[
a_{\alpha,\epsilon}:=\frac{\epsilon^2\mathcal R^{\mathrm{st},*}(\alpha,\epsilon)}{z_\alpha^2},
\qquad
b_{\alpha,\epsilon}:=\frac{\epsilon^2\mathcal R^*(\alpha,\epsilon)}{z_\alpha^2}.
\]
The adaptive lower bound below makes $\liminf_{\epsilon\to0}b_{\alpha,\epsilon}>0$ for sufficiently small fixed $\alpha$, and both upper limits are finite. Therefore,
\begin{align*}
\frac{\liminf_{\epsilon\to0}a_{\alpha,\epsilon}}{\limsup_{\epsilon\to0}b_{\alpha,\epsilon}}
&\le\liminf_{\epsilon\to0}\frac{a_{\alpha,\epsilon}}{b_{\alpha,\epsilon}}\\
&\le\limsup_{\epsilon\to0}\frac{a_{\alpha,\epsilon}}{b_{\alpha,\epsilon}}
\le\frac{\limsup_{\epsilon\to0}a_{\alpha,\epsilon}}{\liminf_{\epsilon\to0}b_{\alpha,\epsilon}}.
\end{align*}
As $\alpha\to0$, the first and last expressions both converge to $r^*/(r^*/4)=4$. This proves both ordered-limit statements in Theorem~\ref{thm:VoA}, without assuming that an ordinary inner limit exists for every fixed $\alpha$.\hfill\Halmos

\vspace{0.1in}

\subsection{Proof of Proposition~\ref{prop:asym_static}}
\noindent\textbf{Weak-signal expansions.}
Fix a sufficiently small upper bound on $\epsilon$ that keeps all accuracies strictly inside $(1/2,1)$. For each LLM, write $p_{j,A}=1/2+a_j\epsilon$ and $p_{j,B}=1/2-b_j\epsilon$ for the probability of output $A$ under the respective hypotheses. Taylor expansion of the two logarithms gives
\begin{align*}
\ell_j(A)&=\log\frac{1+2a_j\epsilon}{1-2b_j\epsilon}, \ \  \ell_j(B) =\log\frac{1-2a_j\epsilon}{1+2b_j\epsilon},\\
I_{j,\theta}&=2s_j^2\epsilon^2+O(\epsilon^4), \ \
 V_{j,\theta} =4s_j^2\epsilon^2+O(\epsilon^4),\\
\beta_{j,\theta}&:=\EE_\theta|\ell_j(Y)-\EE_\theta\ell_j(Y)|^3
=8s_j^3\epsilon^3+O(\epsilon^5).
\end{align*}
For the last two expansions, if $d_j:=\ell_j(A)-\ell_j(B)$ and $p=p_{j,\theta}$, the exact Bernoulli formulas are
\[
V_{j,\theta}=p(1-p)d_j^2,
\qquad
\beta_{j,\theta}=p(1-p)\bigl[p^2+(1-p)^2\bigr]|d_j|^3,
\]
with $d_j=4s_j\epsilon+O(\epsilon^3)$.

Define the total signal budget
\[
Q_\epsilon(\bs n):=\epsilon^2\sum_j s_j^2n_j.
\]
Since $0<s_{\min}\le s_j\le s_{\max}<\infty$, there is a fixed $K$ such that, for \emph{every} nonnegative allocation and sufficiently small $\epsilon$,
\begin{align*}
\mu_\theta(\bs n)&=2Q_\epsilon(\bs n)[1+O(\epsilon^2)],\\
\sigma_\theta^2(\bs n)&=4Q_\epsilon(\bs n)[1+O(\epsilon^2)],\\
\mathcal T_\theta(\bs n)&\le K\epsilon Q_\epsilon(\bs n),
\qquad
\Delta_\theta(\bs n)\le\frac{K\epsilon}{\sqrt{Q_\epsilon(\bs n)}}\quad(Q_\epsilon(\bs n)>0).
\end{align*}
All relative error constants here are independent of the allocation.

\noindent\textbf{Normal-quantile scale.}
For $u\to0$, let $z_u=\Phi^{-1}(1-u)$ and let $\varphi$ be the standard normal density. The Mills-ratio bounds used in the manuscript \citep{mills1926table} give
\[
\frac{z_u}{1+z_u^2}\varphi(z_u)\le u\le\frac{\varphi(z_u)}{z_u}.
\]
Taking logarithms yields $z_u^2\sim2\log(1/u)$. Consequently, for every fixed $k>0$,
\[
\frac{\nu(k\alpha)}{z_\alpha}\longrightarrow1,
\qquad
\frac{H_\alpha}{z_\alpha^2}\longrightarrow\frac12
\quad\text{as }\alpha\to0.
\]

\noindent\textbf{Lower bound.}
Fix $0<\alpha<\xi_{\min}/2$. The information constraints established above and the uniform mean expansion imply
$Q_\epsilon(\bs n)\ge b_\alpha/3>0$ for every feasible design when $\epsilon$ is sufficiently small. Thus its total query count diverges at least on the order of $\epsilon^{-2}$, and the Berry--Esseen correction converges to zero \emph{uniformly over feasible allocations}. Proposition~\ref{prop:socp_static_lb_partI} gives
\begin{equation}\label{eq:be_lower}
\frac{\mu_A(\bs n)+\delta}{\sigma_A(\bs n)}\ge\nu(\alpha/\xi_A+\Delta_A(\bs n)),
\qquad
\frac{\mu_B(\bs n)-\delta}{\sigma_B(\bs n)}\ge\nu(\alpha/\xi_B+\Delta_B(\bs n)).
\end{equation}
Writing $Q=Q_\epsilon(\bs n)$, the uniform expansions give
\begin{equation}\label{eq:limit_expression}
\frac{\mu_A(\bs n)+\delta}{\sigma_A(\bs n)}
+\frac{\mu_B(\bs n)-\delta}{\sigma_B(\bs n)}
=2\sqrt Q+O(\epsilon^2\sqrt Q)+O(\epsilon^2/\sqrt Q).
\end{equation}
The leading prior offsets cancel; their residual effect is included in the last term. Since $Q\ge b_\alpha/3$, the last term is $o(1)$. The right-hand sides of~\eqref{eq:be_lower} sum to at least
$\nu(\alpha/\xi_A)+\nu(\alpha/\xi_B)+o(1)$. Bounding the middle error term in~\eqref{eq:limit_expression} above by $K\epsilon^2\sqrt Q$ and rearranging therefore gives
\[
\sqrt Q\ge\frac{\nu(\alpha/\xi_A)+\nu(\alpha/\xi_B)+o(1)}{2(1+K\epsilon^2)}.
\]
The quantiles are positive for the chosen $\alpha$, so this inequality can be squared. Also,
\[
\epsilon^2\sum_jc_jn_j\ge r^*\epsilon^2\sum_js_j^2n_j=r^*Q.
\]
Taking the infimum over feasible designs and then the lower limit yields
\[
\liminf_{\epsilon\to0}\epsilon^2\mathcal R^{\mathrm{st},*}(\alpha,\epsilon)
\ge\frac{r^*}{4}\bigl[\nu(\alpha/\xi_A)+\nu(\alpha/\xi_B)\bigr]^2.
\]
The uniformity above makes this step valid for allocations chosen optimally as functions of $\epsilon$.

\noindent\textbf{Upper bound.}
Use the source $j^*$ in every query and select an integer sample size
\[
n^*_{\alpha,\epsilon}:=
\left\lceil\frac{x_\alpha^2}{s_{j^*}^2\epsilon^2}\right\rceil,
\qquad
x_\alpha:=\max_{\theta\in\{A,B\}}\nu\!\left(\frac{\alpha}{2\xi_\theta}\right)+|\delta|+1.
\]
For the fixed $\alpha<\xi_{\min}/2$, $x_\alpha>1$, and this allocation has
$Q_\epsilon=x_\alpha^2+O(\epsilon^2)$. Under the two hypotheses, its standardized distances to the decision boundary converge to
\[
x_\alpha+\frac{\delta}{2x_\alpha},\qquad
x_\alpha-\frac{\delta}{2x_\alpha}.
\]
Both are at least $\max_\theta\nu(\alpha/(2\xi_\theta))+1$. The limiting normal error probabilities are therefore strictly below $\alpha/(2\xi_\theta)$, and the Berry--Esseen correction tends to zero. For sufficiently small $\epsilon$, the actual errors satisfy $e_\theta\le\alpha/(2\xi_\theta)$, implying $\xi_Ae_A+\xi_Be_B\le\alpha$. This is an explicit feasible static policy, not merely a feasible point of the relaxation. Its cost gives
\[
\limsup_{\epsilon\to0}\epsilon^2\mathcal R^{\mathrm{st},*}(\alpha,\epsilon)
\le\lim_{\epsilon\to0}\epsilon^2c_{j^*}n^*_{\alpha,\epsilon}
=r^*x_\alpha^2.
\]
Dividing both bounds by $z_\alpha^2$ and then letting $\alpha\to0$ proves both claims in Proposition~\ref{prop:asym_static}, because
$[\nu(\alpha/\xi_A)+\nu(\alpha/\xi_B)]/(2z_\alpha)\to1$ and $x_\alpha/z_\alpha\to1$.\hfill\Halmos

\vspace{0.1in}

\subsection{Proof of Proposition~\ref{prop:asym_adaptive}}
Write
\[
\bar H_\alpha:=\xi_AA_\alpha+\xi_BB_\alpha
=H_\alpha+(\xi_B-\xi_A)\delta.
\]
We derive bounds that are uniform over adaptive policies before taking either limit. The single-query expansions above imply, for some fixed $K$ and every $j,\theta$,
\[
I_{j,\theta}(\epsilon)\le2s_j^2\epsilon^2(1+K\epsilon^2),
\qquad C_\ell(\epsilon)=O(\epsilon).
\]
All stopping-time and drift identities from Appendix~C apply for each fixed $\epsilon>0$. Their fixed-instance remainder constants need not be uniform in $\epsilon$; the explicit bounds below avoid that requirement.

\noindent\textbf{Lower bound.}
Let $\pi$ be any admissible policy. Put $U_{A,\tau}=L_\tau$, $U_{B,\tau}=-L_\tau$, and let $e_\theta=\PP_\theta(D\ne\theta)$. From $c_j\ge r^*s_j^2$, the expected-count identities imply
\begin{align*}
\epsilon^2\mathcal R(\pi;\alpha,\epsilon)
&=\epsilon^2\sum_\theta\xi_\theta\sum_jc_j n_{j,\theta} \ge\frac{r^*}{2(1+K\epsilon^2)}
\sum_\theta\xi_\theta\sum_j I_{j,\theta}n_{j,\theta}
=\frac{r^*}{2(1+K\epsilon^2)}\sum_\theta\xi_\theta\EE_\theta[U_{\theta,\tau}].
\end{align*}
The overshoot argument, retaining the actual error probabilities, gives
\[
\EE_A[L_\tau]\ge A_\alpha-(2H_\alpha+C_\ell(\epsilon))e_A,
\qquad
\EE_B[-L_\tau]\ge B_\alpha-(2H_\alpha+C_\ell(\epsilon))e_B.
\]
Every admissible posterior-threshold policy has Bayes error $\xi_Ae_A+\xi_Be_B\le\alpha$. Therefore,
\[
\epsilon^2\mathcal R(\pi;\alpha,\epsilon)
\ge\frac{r^*}{2(1+K\epsilon^2)}
\left[\bar H_\alpha-\alpha(2H_\alpha+C_\ell(\epsilon))\right].
\]
The inequality holds for all policies, whether or not an optimal policy is attained. Taking their infimum and then $\liminf_{\epsilon\to0}$ yields
\[
\liminf_{\epsilon\to0}\epsilon^2\mathcal R^*(\alpha,\epsilon)
\ge\frac{r^*}{2}(\bar H_\alpha-2\alpha H_\alpha).
\]
The right-hand side is positive for all sufficiently small $\alpha$. In particular, the wrong-boundary contribution is kept throughout this argument; it is not replaced by a vanishing overshoot constant.

\noindent\textbf{Upper bound.}
Use only source $j^*$ and stop at the original posterior boundaries, with stopping time $\tau_0$. This is admissible and has finite expected query count for each fixed $\epsilon$. The stopped drift and overshoot identities give
\[
I_{j^*,A}\EE_A[\tau_0]=\EE_A[L_{\tau_0}]\le A_\alpha+C_\ell(\epsilon),
\qquad
I_{j^*,B}\EE_B[\tau_0]=\EE_B[-L_{\tau_0}]\le B_\alpha+C_\ell(\epsilon).
\]
Since $I_{j^*,\theta}=2s_{j^*}^2\epsilon^2[1+O(\epsilon^2)]$ and $C_\ell(\epsilon)\to0$,
\begin{align*}
\limsup_{\epsilon\to0}\epsilon^2\mathcal R^*(\alpha,\epsilon)
&\le\lim_{\epsilon\to0}\epsilon^2c_{j^*}
\left[\frac{\xi_A(A_\alpha+C_\ell(\epsilon))}{I_{j^*,A}}
+\frac{\xi_B(B_\alpha+C_\ell(\epsilon))}{I_{j^*,B}}\right]\\
&=\frac{r^*}{2}\bar H_\alpha.
\end{align*}
Finally, $\bar H_\alpha/z_\alpha^2\to1/2$ and $2\alpha H_\alpha/z_\alpha^2\to0$ as $\alpha\to0$. Dividing the bounds by $z_\alpha^2$ proves both assertions in Proposition~\ref{prop:asym_adaptive}. \hfill\Halmos

\section{Further Details on the Numerical Experiments}\label{append:bellman}

This appendix describes the Bellman benchmarks and waiting time distribution used in Section~\ref{sec:numerics}. Hypothesis $A$ means that the candidate solution is correct, whereas $B$ means that it is incorrect. The simulation uses the symmetric prior $\xi_A=\xi_B=1/2$, so $\delta=0$ and $A_\alpha=B_\alpha=H_\alpha:=\log((1-\alpha)/\alpha)$. The waiting penalty is $g(w)=C_w w^\rho$, with $C_w>0$ and $\rho\ge1$. Here, $C_w$ is the coefficient denoted by $c$ in Assumption~\ref{assum:g}.

\subsection{Reconstruction of Waiting Time Distribution}
\label{append:monte-carlo}

The raw query-level waiting times are not observable. We therefore
reconstructed the waiting-time distribution of each LLM using an equal mixture of ten Gamma distributions, each parameterized to match one of the ten available repeat-specific average waiting times and variances. At each query to a given LLM, one of the ten Gamma components is selected with equal probability and a waiting time is sampled from that component. This construction preserves positivity and incorporates the repeat-specific average waiting times, variances, and between-repeat variability. Because the reconstruction is based on summary statistics rather than raw query-level observations, it should be interpreted as a moment-matched reconstruction rather than an empirical bootstrap.

The Gamma-mixture reconstruction is used only for the numerical evaluation. Because an untruncated Gamma distribution is not sub-Gaussian, this reconstructed waiting time model does not directly satisfy Assumption~\ref{assum:wait}. The simulation should therefore be interpreted as an empirical robustness evaluation under a positive, moment-matched waiting time distribution rather than as a numerical instantiation of the theoretical tail assumption.

\vspace{0.1in}

\subsection{Definition of the Bellman benchmark}\label{append:bellman-defn}

At LLR state $L$, the posterior probability of hypothesis $A$ is
\[
p_A(L)
=
\frac{\exp(\delta+L)}
     {1+\exp(\delta+L)}.
\]
Conditional on querying LLM $j$, the predictive probability of observing output $A$ is
\[
q_j(L)
=
p_A(L)\gamma_{j,A}
+
\bigl[1-p_A(L)\bigr]\bigl(1-\gamma_{j,B}\bigr).
\]
The two terms correspond respectively to a correct output of $A$ when $\theta=A$ and an incorrect output of $A$ when $\theta=B$.

The deterministic-latency benchmark sets $T_{j,t}\equiv\mu_j$. At state $(L,w)$, $w$ is the delay that has already accumulated. Define $V_\alpha(L,w)$ as the minimum expected sum of \emph{future} monetary query costs and the waiting penalty assessed on the \emph{total} delay at stopping. Previously paid query costs are excluded, but previously incurred delay remains relevant through the terminal penalty. Thus, for $-B_\alpha<L<A_\alpha$,
\begin{align}
V_\alpha(L,w)
=
\min_{j\in[M]}
\Bigl\{
&c_j
+
q_j(L)
V_\alpha\bigl(L+\ell_j(A),w+\mu_j\bigr)
\nonumber\\
&+
\bigl[1-q_j(L)\bigr]
V_\alpha\bigl(L+\ell_j(B),w+\mu_j\bigr)
\Bigr\},
\label{eq:augmented-state-bellman}
\end{align}
with a terminal condition
\[
V_\alpha(L,w)=g(w)
\quad
\text{if }
L\ge A_\alpha
\text{ or }
L\le-B_\alpha.
\]
The minimum is over all models. Although each query has a deterministic duration in this benchmark, the outputs, subsequent model choices, stopping round, and total accumulated delay can still be random. The benchmark therefore retains uncertainty in evidence acquisition while removing the additional randomness in each query's duration.

Equivalently, define the transformed value function
\[
U_\alpha(L,w)
=
V_\alpha(L,w)-g(w).
\]
Its terminal value is zero, and the incremental waiting cost associated with
querying LLM $j$ is
\[
g(w+\mu_j)-g(w).
\]
For $g(w)=C_w w^\rho$, this increment becomes
\[
g(w+\mu_j)-g(w)
=
C_w\left[(w+\mu_j)^\rho-w^\rho\right].
\]
These increments telescope to the final waiting penalty minus $g(w)$ at the initial state of the recursion. Thus, the transformation subtracts a state-dependent constant, preserves the minimizing model choices, and uses a zero terminal value without charging the waiting penalty twice.

When $\rho=1$, the incremental waiting cost reduces to $C_w\mu_j$ and is independent of the accumulated waiting time $w$. Therefore, $w$ need not be included in the state, and the transformed value function satisfies the single-dimensional Bellman equation
\begin{align}
U_\alpha(L)
=
\min_{j\in[M]}
\Bigl\{
&c_j+C_w\mu_j
+
q_j(L)U_\alpha\bigl(L+\ell_j(A)\bigr)
\nonumber\\
&+
\bigl[1-q_j(L)\bigr]
U_\alpha\bigl(L+\ell_j(B)\bigr)
\Bigr\},
\label{eq:one-dimensional-bellman}
\end{align}
with a terminal condition
\[
U_\alpha(L)=0
\quad
\text{whenever }
L\ge A_\alpha
\text{ or }
L\le-B_\alpha.
\]

The initial benchmark values are $U_\alpha(0)$ for the one-dimensional formulation and $V_\alpha(0,0)=U_\alpha(0,0)$ for the augmented-state formulation. The reported quantity $\mathcal R_{\mathrm{Bellman}}^{\mathrm{det}}$ in Section~\ref{sec:numerics} is their numerical approximation. For the $\rho=1$, varying-$C_w$ experiment, the LLR grid has spacing $\Delta L=0.0025$, and no waiting-time grid or truncation is used. For the varying-$\rho$ experiment with $1\le\rho\le2$, the augmented-state calculation uses approximately 4,001 LLR grid points, $\Delta w=0.125$ seconds, and $W_{\max}=500$ seconds. Its LLR grid includes both stopping boundaries and $L=0$ exactly; the resulting spacing ranges from approximately $0.00110$ to $0.01151$ over $x=\log_{10}(1/\alpha)\in\{1,\ldots,10\}$. The truncation level is a Bellman-grid parameter, not a posterior stopping boundary.

\end{document}